\documentclass[12pt]{article}
\usepackage[margin=1in]{geometry}

\usepackage[T1]{fontenc} 
\usepackage{lmodern}
\usepackage[utf8]{inputenc} 
\usepackage[english]{babel} 

\usepackage{amsfonts,dsfont} 
\usepackage{amsthm} 
\usepackage{amsmath} 
\usepackage{braket} 
\usepackage{mathtools} 
\usepackage{bm} 
\usepackage{mathrsfs} 

\usepackage{microtype} 
\usepackage{natbib} 
\usepackage[colorlinks=true,allcolors=blue]{hyperref} 
\usepackage{setspace} 
\usepackage{multirow} 
\usepackage{enumitem} 
\usepackage{caption}

\theoremstyle{definition}
\newtheorem{definition}{Definition}
\newtheorem{remark}[definition]{Remark}

\theoremstyle{plain}
\newtheorem{theorem}[definition]{Theorem}
\newtheorem{proposition}[definition]{Proposition}
\newtheorem{lemma}[definition]{Lemma}

\newcommand{\R}{\mathbb{R}}
\newcommand{\indic}{\mathds{1}}
\newcommand{\Prob}{\mathbb{P}}
\newcommand{\E}{\mathbb{E}}
\newcommand{\X}{\mathbb{X}}
\newcommand{\crm}{\tilde \mu}

\newcommand{\lesspace}{\vspace{-0.125\baselineskip}}
\newcommand{\cut}{\vspace{0pt}}

\begin{document}
	
	\pagestyle{plain}
	\pagenumbering{arabic}
	
	\title{\bf Principled Estimation and Prediction \\ with Competing Risks: \\ a Bayesian Nonparametric Approach }
	
	\author{Claudio Del Sole$^1$ \qquad Antonio Lijoi$^2$ \qquad Igor Pr\"unster$^2$ \\
		{\small $^1$Department of Economics, Management and Statistics, University of Milano-Bicocca} \\[-12pt]
		{\small $^2$Bocconi Institute for Data Science and Analytics, Bocconi University}}
	
	\date{}
	
	\maketitle
	
	\begin{abstract}
		\noindent Competing risks occur in survival analysis when multiple causes of death are present. They play a prominent role in several domains extending beyond biostatistics to encompass epidemiology, actuarial sciences, and reliability theory. This paper adopts a multi--state modeling framework to competing risks. We introduce a class of flexible nonparametric priors, defined through hierarchical completely random measures, to model the transition probabilities, and identify the specific (conditionally) conjugate member of this general class. Furthermore, we determine the joint marginal distribution of the data and of a latent random partition, and characterize the posterior distribution of the model. 
		Leveraging these distributional results, we evaluate the predictive probability that a future event is of a specific type (e.g.~death from a particular cause), as a function of the time at which the event occurs. The resulting function, derived on sound principles, is termed the \textit{prediction curve}, and represents a major innovation in the literature. In addition, we provide posterior estimates for the survival function, and for the cause--specific incidence and subdistribution functions.
		Suitable simulation algorithms for posterior inference are also devised.
		The model's performance, as well as the algorithms' effectiveness, is evaluated through simulation studies. Finally, we illustrate our approach on clinical datasets.
	\end{abstract}

	\medskip
	
	
	\vfill
	\newpage
	
	\section{Introduction}
	\label{sec:intro}
	
	Survival analysis must often account for multiple competing sources of risk, where the observed failure (or death) is determined by one among several causes (or types of events). For instance, in clinical studies, a fatal outcome after treatment may be due to a number of different causes, that are either related or unrelated to the treatment, and only one of those causes can be actually observed.
	\emph{Competing events} or \emph{competing risks} arise when the observation of an event of interest, and of the associated survival time, is potentially prevented by the occurrence of one of a set of distinct alternative events \citep{kalbfleisch2002, geskus2024}. The natural application of competing risks analysis is in biomedical and clinical studies, in which subjects may die from different causes. However, competing risks play a prominent role in several domains, another notable setting being reliability--testing in industry, where the possible breakdown of a complex system is due to the failure of one of its components.
	
	\subsection{Approaches to competing risks}
	\label{sec:competing_risks}

	Throughout, the observed survival time is denoted by $T \in \R^+$ and assumed to be a random variable (r.v.) with absolutely continuous distribution (with respect to the Lebesgue measure). Moreover, $\Delta \in \set{1, \dots, D}$ is the categorical r.v.~identifying the observed event type, or equivalently the cause of death/failure. For each event type $\delta\in\set{1,\ldots,D}$, the \emph{cause--specific hazard rate} at time $t \ge 0$ is defined as
	\begin{equation}
		\label{eq:cause_hazard}
		h_\delta(t) := 
		\lim_{s \to 0} \ \frac{1}{s} \ 
		\Prob\left( t \le T < t + s, \, \Delta = \delta \mid T > t\right),
	\end{equation} 
	representing the instantaneous rate of occurrence of an event of type $\delta$ at time $t$, given survival up to that time from \emph{all} possible types of events. Hence, individuals experiencing a competing event are no longer at risk. The cause--specific hazard rates are the building blocks of several important quantities. For instance, the \emph{survival function} of $T$, evaluated at time $t\ge0$, is given by \lesspace
	\begin{equation}
		\label{eq:overall_surv}
		S(t) := \Prob\left(T > t\right) = \exp \left( - \sum_{\delta=1}^D \int_0^t h_\delta (s) \, ds \right). 
	\end{equation}
	Similarly, the cause--specific \emph{incidence functions}, representing the infinitesimal probability of occurrence of each event type, are defined in terms of the cause--specific hazard rates as
	\begin{equation*} 
		f_\delta(t) := \lim_{s \to 0} \ \frac{1}{s} \ \Prob \left(  t \le T < t + s, \, \Delta = \delta \right) = h_\delta(t) \, \exp \left( - \sum_{\ell=1}^D \int_0^t h_\ell(s) \, ds \right), \qquad t \ge 0.
	\end{equation*}
	The probability of occurrence for each event type $\delta$, that is, the proportion $\pi_\delta$ of subjects eventually experiencing $\delta$, is the limit of the corresponding \emph{cumulative incidence function}, or \emph{subdistribution},
	\begin{equation*}
		\pi_\delta = \Prob\left(\Delta = \delta\right) = \lim_{t \to \infty} \int_0^t f_\delta(s) \, ds,
	\end{equation*}	
	with the obvious constraint $\sum_{\delta=1}^D \pi_\delta = 1$.
	This setup is consistent with the \emph{multi--state approach} to statistical modeling of competing risks. Indeed, the cause--specific hazards can be seen as transition rates in a multi--state model \citep{andersen2002, putter2007} for a Markov process $\{M(t) \colon t\ge 0\}$ with one transient state, `alive', labeled as $0$, and $D$ absorbing states `death from cause $\delta$', with label $\delta\in\set{1,\ldots,D}$. Hence, $h_\delta$ is the transition intensity from $0$ to $\delta$. Specifically, if $P_{\ell,k}(0,t)=\Prob\left(M(t)=k \mid M(0)=\ell\right)$ denotes the probability of moving from state $\ell$ to state $k$ in the time interval $(0,t)$, then
	\begin{equation}
		\label{eq:transitions}
		P_{0,0}(0,t)=S(t),  \qquad P_{0,\delta}(0,t)=\int_0^t f_\delta(s)\,ds,
	\end{equation} 
	are respectively the survival function in \eqref{eq:overall_surv} and the cause--specific cumulative incidence function.
	
	An alternative approach to competing risks data \citep{crowder2012} considers latent, or \emph{potential}, survival times $\bm{Y}_{1:D}=(Y_1,\ldots,Y_D)$, one for each event type. They are assumed to be distinct almost surely (a.s.)
	and the observed survival time is $T = \min\{Y_1,\ldots,Y_D\}$. Whilst intuitive, this approach has been criticized for its lack of plausibility and interpretability in biomedical applications, as it assumes all event types eventually occur for each individual \citep{geskus2024}. Additionally, the marginal distributions of the latent survival times cannot be identified without further assumptions on the dependence structure  \citep{kalbfleisch2002}. Even considering a dependent parametric model for $\bm{Y}_{1:D}$, observed data do not allow to discern this model from one with independent risks \citep{cox1959, tsiatis1975, crowder1991}. For this reason, several authors suggest to focus only on the estimation of observable quantities, often assuming independence of competing events for convenience and interpretability. However, this approach is  useful to construct simulated datasets; see  \citealp[e.g.][]{beyersmann2009} and Section~\ref{sec:simulation_three_risks}. Section~\ref{supsec:latent_times} outlines a Bayesian nonparametric approach to modelling latent failure times.
	
	The statistical analysis of competing risks data typically focuses on: (i) the estimation of the cause--specific cumulative incidence functions, which allows to estimate the probabilities that each competing event occurs within some time interval; (ii) the association of these quantities with different treatments or predictors of interest \citep{finegray1999}; (iii) the estimation of the survival function $S$ in \eqref{eq:overall_surv}. These statistical problems are widely discussed in the frequentist literature, as evidenced by numerous reviews and textbooks; see, e.g., \cite{kalbfleisch2002}, \cite{lawless2003}, \cite{crowder2012}, \cite{geskus2015,geskus2024}, \cite{cooklawless2018}. 
	In contrast, the literature on Bayesian nonparametric approaches to competing risks is rather limited. Bayesian models are often based on gamma or beta process priors, which were successfully employed to univariate survival data 
	\citep{dykstra1981, lo1989, hjort1990}. Recently, \cite{arfe2019} introduced a beta--Stacy process generalization \citep{walker1997} for cumulative incidence functions on a discrete timescale. \cite{xu2016} proposed a model for evaluating dynamic treatment regimes based on dependent Dirichlet processes, while \cite{xu2020} develop a similar approach for semi--competing risks. Instead, independent gamma processes are considered in \cite{oganisian2024} for modelling sequential treatment decisions with time-varying confounders. Further pointers and discussions can be found in \cite{mullerquintana2004,muller2015}. An alternative construction based on Bayesian Additive Regression Trees was proposed in \cite{sparapani2020}; we discuss and compare it to our method in Section~\ref{supsec:sparapani}, also through numerical illustrations.
	
	Unlike previous contributions, we take a comprehensive approach by specifying a prior on the transition probabilities of the associated multi--state model, successfully addressing all major inferential problems in the competing risks literature.
	Additionally, our principled framework leads to a significant breakthrough in prediction: we introduce the concept of \textit{prediction curve}, which corresponds to the conditional probability of a future event being of a specific type given the past observations, as a function of the event's occurrence time. Prediction curves are not only of foundational and methodological interest, but also offer great potential for practical decision--making due to their clear interpretability.
	
	\subsection{A Bayesian nonparametric model}
	
	Our model for competing risks data employs a dependent prior for the cause--specific hazard rates, which induces a prior on the transition probabilities $P_{0,\delta}$ and $P_{0,0}$ in \eqref{eq:transitions}. For each $\delta\in\set{1,\ldots,D}$, assuming $\tilde h_\delta$ is a random cause--specific hazard, the induced prior coincides with the distribution of
	\begin{equation*}
		\tilde P_{0,\delta}(0,t)
		=\int_0^t \tilde h_\delta(s) \, \exp\left(-\sum_{\ell=1}^D\int_0^s \tilde h_\ell(u)\, du \right) ds,
	\end{equation*}
	while $\tilde P_{0,0}(0,t)$ is expressed as in \eqref{eq:overall_surv}, replacing each $h_\delta$ with its random counterpart $\tilde h_\delta$.
	For modeling the random hazard rates, we adopt a kernel mixture specification, the most popular approach in Bayesian Nonparametrics. In the univariate case, this strategy was pioneered by \cite{dykstra1981} and \cite{lo1989} using gamma processes, and unraveled in its full potential by \cite{ishwaran2004} and \cite{james2005} for generic kernels and completely random measures as mixing distributions. 
	In our more complex multivariate setup, we consider dependent cause--specific mixture hazard rates by specifying dependent mixing measures.
	In particular, among the various possibilities \citep{quintana2022}, we adopt a hierarchical specification and resort to hierarchical completely random measures \citep{camerlenghi2021}, known for their analytical tractability.
	This dependent  prior is consistent with an assumption of conditional independence of $\tilde h_1,\ldots,\tilde h_D$, given a common random measure $\tilde\mu_0$ at the hierarchy's root, leading to two key advantages: (a) it induces unconditional dependence among the hazard rates; (b) it enables borrowing of information across different causes or types of events $\delta\in\set{1,\ldots,D}$. This specification allows substantial flexibility in modeling the hazard rates without imposing restrictive dependence assumptions.

	We provide a complete picture of the posterior behavior of the proposed model and identify the (conditionally) conjugate prior within our general class. The latter is based on generalized gamma completely random measures, and parallels the role played by the beta process \citep{hjort1990} in nonparametric models for cumulative hazards and by the beta-Stacy process \citep{walker1997} in neutral--to--the--right models. Moreover, we obtain Bayesian estimates of the cause--specific cumulative incidence functions and of the survival function, along with a characterization of the latent random partition induced by the model. 
	Leveraging the partition distribution, we introduce the concept of \textit{prediction curve}, which corresponds to the probability of a future failure or death due to a specific cause, as a function of the event's occurrence time, $t\mapsto \Prob\left(\Delta_{n+1}=\delta\mid T_{n+1}=t,\mbox{ data}\right)$. Prediction curves exhibit a neat form that reflects the composition structure generated by hierarchical discrete random measures. This principled tool is a major strength of our approach; as a preliminary peek at its practical relevance and immediate interpretability, Figure~\ref{fig:prediction_melanoma} displays the estimates of the prediction curves, with 0.95 pointwise credible bands, for the ``melanoma dataset'' analyzed in Section~\ref{supsec:melanoma}. This dataset originates from a clinical study at the Odense University Hospital on stage I melanoma patients: the two competing risks are death from melanoma and from other causes.

	\begin{figure}
		\centering
		\includegraphics[width=0.6\textwidth]{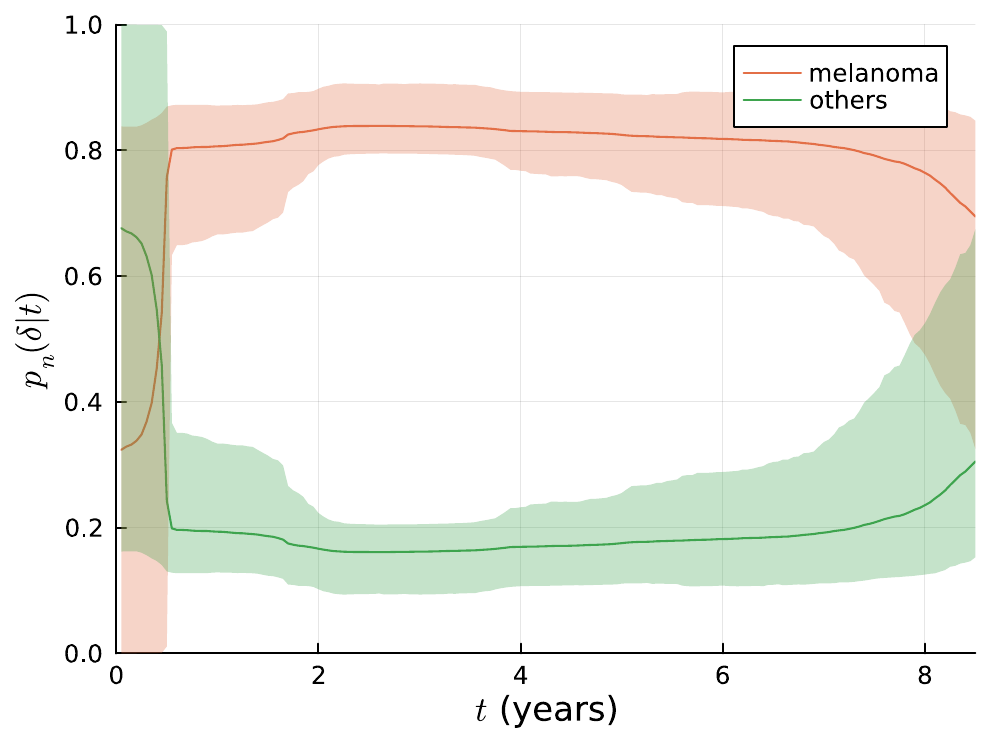}
		\captionsetup{width=0.87\textwidth,font=small}
		\caption{Prediction curves corresponding to melanoma and other causes of death for a future observation $\Delta_{n+1}$, as a function of the time $T_{n+1}=t$ at which the event occurs.} \label{fig:prediction_melanoma} 
	\end{figure}

	\subsection{Outline} 
	
	Section~\ref{sec:mixture_hazard} introduces kernel mixtures with respect to random measures to construct a prior on cause--specific hazard rates, which induces a prior on the space of transition probabilities of the multi--state model. In Section~\ref{sec:hierarchical_CRM}, we define  hierarchical completely random measures, introduce the atoms' marking approach as a key technical tool to derive the main results, and establish consistency results for specific kernel choices.
	Section~\ref{sec:marginal} investigates the latent clustering structure, which is pivotal for predictive inference, particularly in determining prediction curves, and for characterizing the posterior distribution, as detailed in Section~\ref{sec:posterior}. Simulation algorithms for posterior inference are devised in Section~\ref{sec:gibbs}. Finally, in Section~\ref{sec:simulation_study}, we present two simulation studies and a clinical application to evaluate the model's performance and the effectiveness of our samplers.
	References to the Supplementary Material are identified by the prefix `S', e.g.~Section S1. The Julia implementation for the proposed algorithms is available at: \url{github.com/claudiodelsole/CompetingRisks.jl}.
	
	\section{Cause--specific hazards and prediction curves}
	\label{sec:mixture_hazard}
	
	Kernel mixtures represent the most popular Bayesian nonparametric approach for exchangeable single--cause, and possibly right--censored, survival data \citep{dykstra1981,james2005}. 
	In our setting,	such a mixture specification applies to the cause--specific random hazard rates
	\begin{equation} \label{mixture_hazard}
		\tilde h_\delta (t) = \int_\X k(t;x) \, \crm_\delta(dx), \qquad \delta \in \set{1,\dots,D},
	\end{equation}
	where $\crm_1, \dots, \crm_D$ are random measures on $\X$ and $k \colon \R^+ \times \X \mapsto \R^+$ is a deterministic non-negative kernel. The distribution of $\tilde h_\delta$ serves as the prior on the space of hazards corresponding to cause $\delta$. However, when extending from a single--cause to the multivariate setup with multiple competing risks, we would like to specify dependent cause--specific mixture hazards. This is achieved by acting at the level of the mixing measures: we assume $\bm{\crm}=(\crm_1,\ldots,\crm_D)$ to be a vector of dependent random measures with hierarchical dependence structure, which corresponds to a hCRM; see Sections~\ref{supsec:crm} and~\ref{sec:dependent_cmrs}.
	Consequently, the prior on the transition probabilities $P_{0,\delta}$ and $P_{0,0}$ is expressed in terms of $\bm{\crm}$ as the law of $\{\tilde P_{0,\delta}(0,t) \colon t\ge 0\}$ and $\{\tilde P_{0,0}(0,t) \colon t\ge 0\}$, respectively, where
	\begin{gather} 
		\tilde P_{0,\delta}(0,t)
		=\int_0^t \int_\X k(s; x)\,\crm_\delta(dx)\:
		\exp \left( -\sum_{\ell=1}^D
		\int_0^s \int_\X k(u; x)\,\crm_\ell(dx)\,du \right) ds, \nonumber \\
		\tilde P_{0,0}(0,t) =
		\tilde S(t) = \exp \left( - \sum_{\ell=1}^D \int_0^t \int_\X k(s;x) \, \tilde \mu_\ell(dx)\,ds \right). \label{survival}
	\end{gather}
	For increasing hazards, the standard choice is the Dykstra–Laud kernel  $k(t; x) = \gamma(x) \, \indic_{\{t \ge x\}}$, with $\X = \R^+$ and $\gamma(\,\cdot\,)$ a positive right-continuous function, introduced in \cite{dykstra1981}. 
	Typically, $\gamma(x) = \gamma$ for any $x\in\R^+$ is assumed, a simplifying assumption we also adopt; however, results can be easily extended to any choice of  $\gamma(\,\cdot\,)$. Alternative kernel choices leading to different shapes of the hazard rates are the rectangular kernel $k(t; x) = \gamma(x) \,\indic_{\{0 \le t - x\le \tau\}}$, with bandwidth $\tau > 0$, and the Ornstein-Uhlenbeck kernel $k(t; x) = \sqrt{2 \kappa} \, \exp(- \kappa(t-x)) \,\indic_{\{t \ge x\}}$, with rate $\kappa > 0$ \citep{ishwaran2004, peccati2008, deblasi2009}. Further methodological and computational investigations can be found, e.g., in \cite{ishwaran2004, nieto2004, catalano2020}.
	
	We consider the transition probabilities \eqref{survival} to model competing risks data. Specifically, for exchangeable observations $(T_i, \, \Delta_i)$, where $T_i \in \R^+$ is the $i$--th survival time and $\Delta_i \in \set{1, \dots, D}$ is the corresponding event type, our nonparametric model takes the form
	\begin{equation} \label{model}
		(T_1, \Delta_1), \, \dots \, , (T_n, \Delta_n) \mid \bm{\crm} \: \overset{\text{\scriptsize iid}}{\sim} \: \tilde p, \qquad
		\boldsymbol \crm \: \sim \: \mbox{hCRM},
	\end{equation}
	where for any $i=1,\ldots,n$ and $\delta\in\{1,\ldots,D\}$,  
	\begin{equation}
		\label{directing_measure}
		\tilde p(dt, \delta)
		= \tilde f_\delta(t) \,dt =
		\int_\X k(t;x) \, \tilde \mu_\delta(dx) \: \exp \left( - \sum_{\ell=1}^D \int_0^t \int_\X k(s;x) \, \tilde \mu_\ell (dx)\,ds \right) dt.
	\end{equation}
	The prior distribution on $\tilde p$ is thus induced by the distribution of $\bm{\crm}$, characterized in \eqref{hierarchical_crm}. As a consequence, the prior probability of observing an event of type $\delta$, conditionally on $T$ and $\bm{\crm}$, is \begin{equation} \label{eq:relative_hazards}
		\Prob\left(\Delta_i=\delta \mid T_i=t,\,\bm{\crm}\right) \: = \: \frac{ \displaystyle 
		\int_\X k(t;x)\,\tilde\mu_\delta(dx)}{ \displaystyle 
		\sum_{\ell=1}^D \int_\X k(t;x)\,\tilde\mu_\ell(dx)}, \:
	\end{equation}
	for any $\delta\in\set{1,\ldots,D}$ and $t > 0$. 
	In the literature on multi--state models, the ratios of hazards in \eqref{eq:relative_hazards} correspond to the transition probabilities given the transition time and model parameters \citep[Section 2.5]{cooklawless2018}. These are useful for simulation purposes: indeed, the conditional distribution of the event type is multinomial.
	Importantly, from a Bayesian perspective, the function $t\mapsto \Prob\left(\Delta_i=\delta \mid T_i=t \right)$ can be interpreted as the prior predictive probability of observing an event of a given type $\delta$ as the survival time $T_i$ varies. 
	A primary goal of the paper, achieved in Section~\ref{sec:marginal}, is the evaluation of the updated (posterior) function, given the data, which we call the \emph{prediction curve}.
	
	\begin{definition} \label{def:pred_curve}
		The \textit{prediction curve} associated to cause $\delta\in\set{1,\ldots,D}$, for $n\ge 1$, is the function
		\begin{equation}
			\label{eq:predict_curve}
			t \, \mapsto \,  p_n(\delta\mid t) = \Prob \left( \Delta_{n+1}=\delta \mid (T_i, \Delta_i)_{i=1}^n,T_{n+1}=t \right).
		\end{equation}
	\end{definition}
	
	\noindent Hence, conditionally on the observed data, the prediction curve $p_n(\delta \mid t)$ describes the probability of the next event being of type $\delta$ as a function of the survival time $T_{n+1}$. This serves as a valuable tool for evaluating the relative probabilities of the different event types over time.
	
	The rest of the paper develops a comprehensive distributional theory for this class of models including: (i) a characterization of the marginal distribution of data and latent variables, essential for determining the prediction curve \eqref{eq:predict_curve}; (ii) the posterior distribution of $\tilde p$ in \eqref{directing_measure}, enabling Bayesian estimation and uncertainty quantification of both the survival function and cause--specific incidence functions. The fundamental technical tools to achieve these goals are provided in the next section.
	
	\section{Dependent prior specification}
	\label{sec:hierarchical_CRM}

	We construct a tractable and effective dependent prior for $\bm{\crm}=(\crm_1,\ldots,\crm_D)$, inducing a prior on the space of cause-specific hazard rates and transition probabilities, as described in Section~\ref{sec:mixture_hazard}. Specifically, we adopt hierarchical structures, a popular approach for introducing dependence among \textit{completely random measures} (CRMs). Then, we introduce the marked process, investigate some key prior properties, and provide a consistency result. An overview of CRMs is given in Section~\ref{supsec:crm}.
	
	\subsection{Hierarchical completely random measures}
	\label{sec:dependent_cmrs}

	A prominent research stream in Bayesian Nonparametrics focuses on heterogeneous data or latent structures associated with different yet related experimental conditions. These settings can be framed within partial exchangeability, which assumes within--group homogeneity and across--group heterogeneity. Success stories of this approach include topic modeling, ANOVA, change--point detection, reliability analysis, species sampling and numerous biostatistics and bioinformatics applications \citep[see][for a recent review]{quintana2022}. A popular strategy to model heterogeneity nonparametrically leverages vectors of dependent CRMs, as prior dependence among CRMs enables borrowing of information across different groups. Several effective proposals have appeared in the literature. In survival analysis, \cite{lijoi2014a} model each vector component as a superposition of a common and an independent idiosyncratic CRMs; similar additive structures are employed in \cite{muller2004} and \cite{lijoi2014b} to model dependent random probabilities. More complex constructions based on multivariate L\'evy intensities are found in \cite{epifani2010} and \cite{riva2018, riva2021}, relying on L\'evy copulas, and in \cite{griffin2017}, where compound random measures (CoRMs) are introduced. CoRMs represent a general class of dependent random measures encompassing both superpositions of CRMs \citep{lijoi2014a, lijoi2014b} and weighted CRMs \citep{chen2013}. See \cite{riva2022} for an application to survival regression models. Instead, \cite{nipoti2018} define group--specific baseline hazard rates in frailty models employing independent CRMs. Further Bayesian nonparametric constructions of dependent random measures with biological and medical applications, though not directly based on CRMs, can be found in \cite{jara2010,jara2011,soriano2019,cassese2019,christensen2020}.
	
	Unlike contributions above, we employ dependent random measures in an exchangeable framework to jointly model survival times $T_i$ and competing causes $\Delta_i$. Specifically, we focus on hierarchical structures that satisfy a conditional independence assumption, paralleling the standard Bayesian paradigm arising from de~Finetti's representation theorem for exchangeable random elements. A vector of random measures $\boldsymbol \crm =(\crm_1,\ldots,\crm_D)$ has hierarchical structure if $\tilde \mu_1, \dots, \tilde \mu_D \mid \crm_0 \overset{\text{\scriptsize iid}}{\sim} \tilde{\mathcal{M}}_{\crm_0}$ and $\crm_0 \sim \mathcal{M}_0$, where $\tilde{\mathcal{M}}_{\crm_0}$ is the conditional distribution of each $\tilde \mu_\delta$ given the random measure $\crm_0$, which is the hierarchy's root with distribution $\mathcal{M}_0$. 
	The idea of infinite--dimensional hierarchical constructions is due to \cite{teh2006}, who introduced the \emph{hierarchical Dirichlet process}. A general distribution theory for hierarchical random probability measures was developed in \cite{camerlenghi2019}. Instead, we build on the hierarchical construction of CRMs in \cite{camerlenghi2021}, considering a more general formulation.
	
	\begin{definition} \label{def:hCRM}
		Let $\crm_0$ be a CRM on $\X$ with L\'evy intensity  $\nu_0(ds,dx)=\rho_0(ds)\,\Lambda_0(dx)$, for some $\sigma$--finite diffuse $\Lambda_0$ on $\X$, which is termed base measure. A vector of \textit{hierarchical completely random measures} (hCRM) $\bm{\crm}=(\crm_1,\ldots,\crm_D)$ is defined as
		\begin{equation} \label{hierarchical_crm}
			\tilde \mu_1, \dots, \tilde \mu_D \, \mid \ \tilde \mu_0 \ \overset{\text{\scriptsize iid}}{\sim} \, \text{CRM}(\tilde \nu), \qquad 
			\tilde \mu_0  \ \sim  \ \text{CRM}(\nu_0),
		\end{equation}
		where $\tilde \nu(ds,dx) = \rho(ds) \, \tilde \mu_0(dx)$ and $\mbox{CRM}(\nu)$ denotes a CRM with intensity measure $\nu$.
	\end{definition} 

	\noindent Partially exchangeable data are naturally suited for the application of hCRMs, as shown in \cite{camerlenghi2021} for estimating the marginal survival function of each population. In contrast, we apply hCRMs in a seemingly incompatible setup, namely exchangeable survival data with competing risks, as detailed in Section~\ref{sec:mixture_hazard}. 	
	For illustration, we consider the \emph{generalized gamma} hCRM, a large subclass of hCRM arising, for any $\beta, \beta_0 > 0$ and $\sigma, \sigma_0 \in [0,1)$, from the specification 
	\begin{equation} \label{generalized_gamma}
		\rho(ds) = \frac{1}{\Gamma(1-\sigma)} \, s^{-1-\sigma} \, e^{-\beta s} \, ds, \qquad \rho_0(ds) = \frac{1}{\Gamma(1-\sigma_0)} \, s^{-1-\sigma_0} \, e^{-\beta_0 s} \, ds.
	\end{equation}
	Special cases include the \emph{gamma} hCRM, which is retrieved by setting $\sigma = \sigma_0 = 0$, the \emph{$\sigma$-stable} hCRM, characterized by $\beta = \beta_0 = 0$, and \emph{inverse Gaussian} hCRM, when  $\sigma = \sigma_0 = 1/2$.
	
	\subsection{The marked process}
	\label{sec:marked_process}
	
	A key feature of hierarchical constructions is that the random measures $\tilde \mu_1, \dots, \tilde \mu_D$ share the same sequence of locations. Indeed, each measure can be represented as $\crm_\delta=\sum_{h\ge 1} S_{\delta h}\,\delta_{X_{\delta h}^*}$ where, conditionally on $\crm_0$, the atoms $(X_{\delta h}^*)_{h\ge 1}$ are i.i.d.~from a distribution obtained as a transformation of $\crm_0$. Notably, the jumps $S_{\delta h}$ and $S_{\delta \ell}$ are indistinguishable if $X_{\delta h}^*=X_{\delta \ell}^*=X_j^*$, i.e.~if they are associated to the same atom from $\crm_0$: indeed, the actual jump of $\crm_\delta$ at $X_j^*$ depends on $S_{\delta h}+S_{\delta \ell}$, and the individual components cannot be disentangled from their sum. We overcome this issue and considerably simplify the treatment of hCRMs by associating to each jump $S_{\delta h}$ a diffuse independent mark $Z_{\delta h}^*$, which acts as unique label. This is achieved by extending $\tilde \mu_1, \dots, \tilde \mu_D$ to the enlarged space $[0,1] \times \X$ with corresponding augmented intensity $\tilde \nu^e(ds,dz,dx) = \rho(ds) \, H(dz) \, \crm_0(dx)$, where $H$ is an arbitrary diffuse probability measure on $[0,1]$. Consequently, each extended random measure at the bottom of the hierarchy is 
	\begin{equation} \label{hierarchical_crm_sequence}
		\crm^e_\delta = \sum_{h \ge 1} S_{\delta h} \, \delta_{(Z^*_{\delta h}, X^*_{\delta h})}, 
		\qquad \delta \in \set{1, \dots, D},
	\end{equation}
	where $(S_{\delta h})_{h \ge 1}$ and $(Z^*_{\delta h})_{h \ge 1}$ are independent sequences of random jumps and marks, respectively, while $(X^*_{\delta h})_{h \ge 1}$ is a sequence of independent locations whose distribution depends on $\crm_0$. 
	The original $\crm_\delta$ is recovered from $\crm^e_\delta$ by marginalizing over the marks. From a technical perspective, this augmentation device enables a direct and straightforward proof strategy; see Section~\ref{supsec:proofs}. Moreover, the same technique may also simplify arguments in other setups, such as those in e.g.~\cite{camerlenghi2019,camerlenghi2021}.
	
	\subsection{Distributional properties and weak consistency}
	\label{sec:consistency}
	
	Firstly, we investigate the properness of the survival function in \eqref{survival}, namely $\lim_{t \to +\infty} \tilde S(t) = 0$ a.s., which depends on both the kernel and the underlying random measures. Assuming a deterministic and non-negative kernel satisfying $\displaystyle \int_0^t k(s; x)\,ds < \infty$ for any $t \ge 0$ and $x \in \X$, the following result provides sufficient conditions.

	\begin{proposition} \label{prop:survival}
		The survival function $\tilde S$ in \eqref{survival} is proper if either of the following conditions holds: \lesspace
		\begin{itemize}
			\item[{\rm (i)}] the 
			$\crm_1,\ldots,\crm_D$ and $\crm_0$ are infinitely active and $\displaystyle \lim _{t \to \infty} \int_0^t k(s; x) \, ds = \infty$ for every $x \in \X$;
			\item[{\rm (ii)}] the 
			measure $\Lambda_0$ is infinite and $\displaystyle \lim _{t \to \infty} \int_0^t k(s; x) \, ds \ge C$ for every $x \in \X$ and some $C > 0$.
		\end{itemize}
	\end{proposition}

	\noindent Specifically, the Dykstra-Laud kernel with generalized gamma hCRMs \eqref{generalized_gamma} satisfies (i), while the rectangular and Ornstein-Uhlenbeck kernels with $\Lambda_0$ proportional to the Lebesgue measure meet (ii). A necessary and sufficient condition, although less transparent, is derived within the proof in Section~\ref{supsec:proofs}.

	As anticipated in Section~\ref{sec:mixture_hazard}, different kernels encode distinct prior properties in the model; e.g. the Dykstra–Laud kernel is suited to distributions with increasing hazards. This naturally motivates an investigation of the prior support under different kernel choices. Let $\mathcal{P}$ be the prior induced by the distribution of the survival function $\tilde S$ in \eqref{survival} on the space of probability densities on $\mathbb{R}^+$. A density $f_0$ belongs to the Kullback-Leibler support of $\mathcal P$ if $\mathcal P$ assigns positive probability to any Kullback--Leibler neighbourhood of $f_0$. Schwartz's Theorem (see \citealp[Chapter 6]{ghosal2017}) then guarantees weak consistency at any ``true'' data--generating density $f_0$ in the Kullback--Leibler support of $\mathcal P$.
	Building on the results for single--cause mixture hazard models in \cite{deblasi2009}, we identify, for different kernel choices, broad classes of hazards $h_0$ such that the corresponding survival density $f_0$ belongs to the Kullback-Leibler support of $\mathcal P$, ensuring weak consistency. For simplicity, consider the generalized gamma hCRM in \eqref{generalized_gamma} with $\Lambda_0$ proportional to the Lebesgue measure. These model assumptions guarantee some technical conditions, but do not compromise the generality of the result; refer to the proof in Section~\ref{supsec:proofs} for details.
	
	\begin{theorem} \label{prop:support}
		Let $\boldsymbol{\crm}$ be the generalized gamma hCRM \eqref{generalized_gamma} and let $\Lambda_0$ be proportional to the Lebesgue measure on $\R^+$. Then, the following hold:
		\begin{itemize}
			\item[{\rm (a)}] \emph{Dykstra--Laud kernel $k(t;x) = \gamma\,\mathds{1}_{\set{t \ge x}}$:} the prior $\mathcal{P}$ is weakly consistent at each $f_0$ such that $\displaystyle \int_0^\infty t^2 f_0(t)\,dt < \infty$, with $h_0(0) = 0$ and $h_0(t)$  strictly positive and non-decreasing for any $t > 0$;
			\item[{\rm (b)}] \emph{Rectangular kernel $k(t;x) = \gamma \,\mathds{1}_{\set{ 0 \le t-x \le \tau}}$:} if the bandwidth $\tau$ is random and its prior assigns positive probability to $[0,L]$ for some $L > 0$, then $\mathcal{P}$ is weakly consistent at each $f_0$ such that $\displaystyle \int_0^\infty t f_0(t)\,dt < \infty$, with $h_0(0) = 0$ and $h_0(t)$ strictly positive, bounded and Lipschitz for any $t > 0$;
			\item[{\rm (c)}] \emph{Ornstein-Uhlenbeck kernel $k(t; x) = \sqrt{2 \kappa} \, \exp(- \kappa(t-x)) \,\mathds{1}_{\set{t \ge x}}$:} the prior $\mathcal{P}$ is weakly consistent at each $f_0$ such that $\displaystyle \int_0^\infty t f_0(t)\,dt < \infty$, with $h_0(0) = 0$ and $h_0(t)$ strictly positive and differentiable for any $t > 0$; also, for any $t > 0$ for which $h_0'(t) < 0$, one has $\vert h_0'(t) \vert /h_0(t) < \kappa \sqrt{2\kappa}$.
		\end{itemize}
	\end{theorem}

	\noindent The class of hazard functions included in the Kullback-Leibler support of $\mathcal{P}$, for which weak consistency holds, depends essentially on the choice of the kernel. The rectangular kernel guarantees consistency for basically any bounded Lipschitz hazard, while the Ornstein-Uhlenbeck kernel covers any differentiable hazard subject to a certain upper bound on the local exponential decay. In contrast, the Dykstra-Laud kernel encompasses any non-decreasing hazard satisfying mild conditions. Remarkably, Theorem~\ref{prop:support} remains valid in presence of independent right--censoring, provided the censoring distribution has full support \citep[see][]{deblasi2009}. An empirical confirmation of this consistency result for the Dykstra-Laud kernel is presented in Section~\ref{supsec:consistency}.
	
	Finally, the hierarchical structure in \eqref{hierarchical_crm}, combined with the conditional probability of cause $\Delta_i$ given the survival time $T_i$ in \eqref{eq:relative_hazards}, implies a discrete uniform prior predictive for $\Delta_i$, that is $\Prob\left(\Delta_i=\delta \mid T_i=t\right) = 1/D$ for $\delta \in \set{1,\ldots,D}$ and any $t > 0$. However, this uniformity can be easily removed by assuming either (i) $\crm_1,\ldots,\crm_D$ conditionally independent but not identically distributed, or (ii) distinct kernels for each event type. Closed--form expressions for the prior predictives $\Prob\left(\Delta_i=\delta \mid T_i=t\right)$ are typically unavailable in these cases, but they can be approximated via numerical integration. The results in the following sections extend naturally to this more general setting, with the only drawback being the added notational complexity.
	
	\section{Latent clustering and prior predictive structures}
	\label{sec:marginal}

	The nonparametric mixture hazards model \eqref{model}--\eqref{directing_measure} induces a random partition, with a nested structure, of latent random elements. This is somewhat reminiscent of random partitions induced by dependent discrete random probability measures, often sequentially represented using the popular \textit{Chinese restaurant franchise} metaphor \citep{teh2006}. Yet, the random partition here is fundamentally different, as it arises from random measures rather than random probability measures: this requires both a distinct analytical treatment and a novel interpretation. 
	To study the distributional properties of the random partition associated to \eqref{model}, we rely on a novel latent structure that facilitates analytical derivations. These latent variables are also pivotal for the posterior characterization in Section~\ref{sec:posterior}.
	The multiplicative intensity likelihood \citep{kalbfleisch2002} for our model is the product of the directing random probability  \eqref{directing_measure} evaluated at the observations, and is expressed as
	\begin{equation} 
		\label{beginning_likelihood}
		\mathcal{L}(\bm{\crm}; \boldsymbol T_{1:n}, \boldsymbol \Delta_{1:n}) 
		= \prod_{i=1}^n \left[ \int_\X k(T_i;x_i) \, \tilde \mu_{\Delta_i}(dx_i) \: \exp \left( - \sum_{\ell=1}^D \int_0^{T_i} \int_\X k(s;y) \, \tilde \mu_\ell (dy)\,ds \right) \right].
	\end{equation}
	Throughout $\bm{x}_{1:n}=(x_1,\ldots,x_n)$ stands for vectors of dimension $n$, for any $n\ge 1$. We assume the observed survival times to be exact to avoid notational burden, although the inclusion of censored data is straightforward; see Section~\ref{sec:applications} and e.g.~\cite{james2005, lijoi2014a}. The likelihood \eqref{beginning_likelihood} can be equivalently reparameterized in terms of $\bm{\crm^e}=(\crm^e_1, \dots, \crm^e_D)$: 
	\begin{multline*}
		\mathcal{L}(\bm{\crm^e} ; \boldsymbol T_{1:n}, \boldsymbol \Delta_{1:n}) 
		\\
		= \prod_{i=1}^n \left[\int_{[0,1] \times \X} k(T_i;x_i) \, \tilde \mu^e_{\Delta_i}(dz_i,dx_i) \: \exp \left( - \sum_{\ell=1}^D \int_0^{T_i} \int_{[0,1] \times \X} k(s;y) \, \tilde \mu^e_\ell (dw,dy)\,ds \right)
		\right].
	\end{multline*}
	Standard approaches for the Bayesian analysis of such type of likelihoods rely on suitable sequences of latent random elements  $\boldsymbol X_{1:n} = (X_1, \dots, X_n)$, which essentially correspond to the (latent) locations sampled from the mixing measures and allow to remove the integration on $\X$. See, e.g., \cite{james2005}, \cite{lijoi2014a} and \cite{camerlenghi2021}.
	In contrast, we resort here to a double augmentation: in addition to $\boldsymbol X_{1:n}$, we introduce $[0,1]$--valued latent r.v.~$\boldsymbol Z_{1:n} = (Z_1, \dots, Z_n)$, leading to the following doubly--augmented likelihood
	\begin{multline*}
		\mathcal{L}(\bm{\crm^e}; \boldsymbol T_{1:n}, \boldsymbol \Delta_{1:n}, \boldsymbol X_{1:n}, \boldsymbol Z_{1:n}) \\
		= \prod_{i=1}^n \left[ k(T_i;X_i) \, \tilde \mu^e_{\Delta_i}(dZ_i,dX_i) \:
		\exp \left( - \sum_{\ell=1}^D \int_0^{T_i} \int_{[0,1] \times \X} k(s;x) \, \tilde \mu^e_\ell (dz,dx)\,ds \right) \right].
	\end{multline*}
	The a.s.~discreteness of $\bm{\crm^e}$ naturally induces a random nested partition structure, which is identified by the tied values displayed by the latent sequences $\boldsymbol X_{1:n}$ and $\boldsymbol Z_{1:n}$. Specifically, the coarser partition is induced by elements in $\boldsymbol X_{1:n}$ and is denoted as
	\begin{equation*}
		\Pi_{\boldsymbol X_{1:n}}=\set{C_{j}^X \colon j=1,\ldots,k}, \qquad \mbox{with } \: C_{j}^X=\set{i \colon X_{i}=X^*_{j}},
	\end{equation*}
	for some $k\in\set{1,\ldots,n}$. Hence,  $n_j=\mbox{card}(C_j^X)$ is the frequency of elements in $\bm{X}_{1:n}$ coinciding with the $j$-th distinct value $X_j^*$, under the constraint $n_1 + \cdots + n_k = n$. The series representation \eqref{hierarchical_crm_sequence} of each $\crm_\delta^e$ entails that these distinct values can be regarded as the latent locations of $\crm_0$ at the root of the hierarchy. Since locations are shared among the $\crm_\delta$'s at the bottom of the hierarchy, the same location can be associated to different competing events. Let $C_{\delta j}^X=\set{i \colon (\Delta_i,X_i)=(\delta,X_j^*)}$, so that $n_{\delta j}=\mbox{card}(C_{\delta j}^X)$ stands for the frequency of observations such that $\Delta_i=\delta$ and $X_i=X_j^*$, with $n_{1j}+\,\cdots\,+n_{Dj}=n_j$. For each $j\in\set{1,\ldots,k}$, one may have $n_{\delta j} \ge 1$ and $n_{\delta'j}\ge 1$, which corresponds to events of type $\delta$ and $\delta'$ sharing the same latent location $X^*_j$. In view of this, $X_i = X_\ell$ does not necessarily imply $\Delta_i = \Delta_\ell$. Therefore, there is no structural relationship between the random latent partition $\set{C_{j}^X \colon j=1,\ldots,k}$ and the observed partition $\set{C_\delta \colon \delta=1,\ldots,D}$ encoded in $\boldsymbol \Delta_{1:n}$, where $C_\delta=\set{i \colon \Delta_i=\delta}$. 
	On the contrary, this relationship is accounted for at the finer level of the partition, induced by elements in $\boldsymbol Z_{1:n}$ and denoted as
	\begin{equation*}
		\Pi_{\bm{Z}_{1:n}}=\set{C_{\delta jh}^Z \colon \delta=1,\ldots,D;\, j=1,\ldots,k;\, h=1,\ldots,r_{\delta j} },
	\end{equation*}
	with $C_{\delta jh}^Z=\set{i\in C_{\delta j} \colon Z_{i}=Z^*_{\delta jh}}$. This means that, for any $\delta$ and $j$ such that $n_{\delta j}\ge 1$, the $n_{\delta j}$ observations for which $\Delta_i = \delta$ and $X_i = X_j^*$ are further partitioned into $r_{\delta j}$ clusters according to the distinct values $Z_{\delta j1}^*, \dots, Z_{\delta jr_{\delta j}}^*$ in $\boldsymbol Z_{1:n}$; clearly, if $q_{\delta jh}=\mbox{card}(C_{\delta jh}^Z)$, one has $q_{\delta j1} + \cdots + q_{\delta jr_{\delta j}} = n_{\delta j}$. The representation in \eqref{hierarchical_crm_sequence} suggests a straightforward interpretation of these distinct $Z_{\delta j1}^*, \dots, Z_{\delta jr_{\delta j}}^*$ as the latent marks from the random measure $\crm^e_\delta$ that are paired with the same location $X_j^*$ from $\crm_0$. Hence, $r_j=r_{1j}+\,\cdots\,+r_{Dj}$ is the total number of different marks across random measures $\crm^e_1, \dots, \crm^e_D$ associated with the same location $X_j^*$, for any latent group $C_j^X$. From these definitions, it is apparent that the partitions $\Pi_{\bm{X}_{1:n}}$ and $\Pi_{\bm{Z}_{1:n}}$ are nested. Specifically, if two observations $T_i$ and $T_\ell$ are tied in their marks, i.e.~$Z_i = Z_\ell$, they must also share the same location, i.e.~$X_i = X_\ell$; the converse is not true. This description of the latent partition 
	allows to express the augmented likelihood as
	\begin{align} \label{likelihood}
		\mathcal{L}(\bm{\crm^e}; \mathcal{F}^*_n)
		& = Q_n(\boldsymbol T_{1:n}, \boldsymbol X_{1:n})\:
		\exp \left( - \sum_{\ell=1}^D \int_{[0,1] \times \X} K_n(x) \, \tilde \mu^e_\ell(dz,dx) \right) \nonumber \\
		& \qquad \qquad \times \prod_{\ell=1}^D \prod_{j=1}^k \prod_{h=1}^{r_{\ell j}} \tilde \mu^e_\ell(dZ^*_{\ell jh},dX^*_j)^{q_{\ell jh}},
	\end{align}
	where, in order to ease notation, we have set $\mathcal{F}^*_n=(\bm{T}_{1:n}, \boldsymbol \Delta_{1:n}, \Pi_{\bm{X}_{1:n}}, \Pi_{\bm{Z}_{1:n}}, \boldsymbol X^*, \bm{Z}^*)$, with $\bm{X}^* = (X^*_1, \ldots, X^*_k)$ and $\bm{Z}^* = (Z_{\delta jh})_{\delta jh}$, and \cut
	\begin{equation} 
		\label{kernel_quantities}
		Q_n(\boldsymbol T_{1:n}, \boldsymbol X_{1:n}) 
		= \prod_{j=1}^k \, \prod_{i\in C_j^X} k(T_i;X^*_j), \qquad
		K_n(x) = K_n(x; \boldsymbol T_{1:n}) 
		= \sum_{i=1}^n \int_0^{T_i} k(s;x)\,ds;
	\end{equation}
	for details on the derivation of \eqref{likelihood} from \eqref{beginning_likelihood}, see Section~\ref{supsec:proofs}.
	The expression \eqref{likelihood} neatly singles out the individual contributions to the likelihood due to the observed survival times and event types, and to the sequences of latent variables. 
	Specifically, the survival times $\boldsymbol T_{1:n}$ impact the likelihood through $Q_n(\boldsymbol T_{1:n}, \boldsymbol X_{1:n})$ and $K_n(x)$. On the other hand, the event types $\boldsymbol \Delta_{1:n}$ impose constraints on the possible configurations of the nested partition structure, which is encoded into the multiplicities $(q_{\delta jh})_{\delta jh}$ at the finer $\bm{Z}$-level, and $(r_{\delta j})_{\delta j}$ at the coarser $\bm{X}$-level.
	
	The joint distribution of the observations $(\bm{T}_{1:n},\bm{\Delta}_{1:n})$ and the latent variables $(\bm{X}_{1:n},\bm{Z}_{1:n})$ is obtained by integrating out $\bm{\crm^e}$ from the likelihood. To this end, let $\psi$ and $\tau$ be the Laplace exponent and cumulants of the random measures $\crm_1, \dots, \crm_D$ at the bottom of the hierarchy, namely
	\begin{equation} \label{cumulants}
		\psi(u) = \int_{\R^+} (1-e^{-us}) \, \rho(ds), \qquad \tau(m; u) = \int_{\R^+} s^m \, e^{-us} \, \rho(ds).
	\end{equation}
	The quantities $\psi_0$ and $\tau_0$ for $\crm_0$ at the hierarchy's root are obtained by replacing $\rho$ with $\rho_0$.
	
	\begin{theorem} 
		\label{prop:marginal}
		The joint marginal distribution of the observations $(\boldsymbol T_{1:n}, \boldsymbol \Delta_{1:n})$ and the latent variables $(\boldsymbol X_{1:n}, \boldsymbol Z_{1:n})$, that is, the marginal distribution of $\mathcal{F}^*_n$, is given by
		\begin{multline*}
			\Prob (\mathcal{F}^*_n) = Q_n(\boldsymbol T_{1:n}, \boldsymbol X_{1:n}) \, \exp \left( - \int_\X \psi_0 \big( D \psi(K_n(x)) \big) \, \Lambda_0(dx) \right) \\
			\times \prod_{j=1}^k \left( \prod_{\ell=1}^D \prod_{h=1}^{r_{\ell j}} \tau(q_{\ell jh}; K_n(X_j^*)) \, H(dZ_{\ell jh}^*) \right) \tau_0 \big( r_j ; D \psi(K_n(X_j^*)) \big) \, \Lambda_0(dX_j^*).
		\end{multline*}
	\end{theorem}

	\noindent Note that, conditionally on $\Pi_{\bm{Z}_{1:n}}$, the distinct values $\bm{Z}^*$ can be integrated out, since they are iid from the diffuse probability  $H$, and hence have no impact on the latent partition structure. In contrast, the random partition $\Pi_{\bm{X}_{1:n}}$ and the distinct values $\boldsymbol X^*$ are not independent and both impact the marginal distribution. Additionally, the distinct values are tied to the observed survival times through $Q_n(\boldsymbol T_{1:n}, \boldsymbol X_{1:n})$. Hence, marginalizing with respect to $\bm{Z}^*$, Theorem~\ref{prop:marginal} can be restated as \lesspace
	\begin{multline} \label{marginal}
		\Prob (\mathcal{F}_n) = 
		Q_n(\boldsymbol T_{1:n}, \boldsymbol X_{1:n}) \, \exp \left( - \int_\X \psi_0 \big( D \, \psi(K_n(x)) \big) \, \Lambda_0(dx) \right) \\
		\times \prod_{j=1}^k \left( \prod_{\ell =1}^D \prod_{h=1}^{r_{\ell j}} \tau(q_{\ell jh}; K_n(X_j^*)) \right) \tau_0 \big( r_j ; D \, \psi(K_n(X_j^*)) \big) \, \Lambda_0(dX_j^*),
	\end{multline}
	where $\mathcal{F}_n = (\bm{T}_{1:n}, \boldsymbol \Delta_{1:n}, \Pi_{\bm{X}_{1:n}}, \Pi_{\bm{Z}_{1:n}}, \boldsymbol X^*)$ is obtained from $\mathcal{F}^*_n$ by removing the distinct values $\bm{Z}^*$.
	This expression serves as the starting point for deriving the computational results in Section~\ref{sec:gibbs}: indeed, a marginal Gibbs sampling scheme is obtained from the full conditional distributions derived from \eqref{marginal}.
	
	As an important by--product of Theorem~\ref{prop:marginal}, we determine the predictive distribution for the event type associated with the $(n+1)$-th observation, that is the predictive distribution of $\Delta_{n+1}$, given the survival time $T_{n+1}$ and $\mathcal{F}_n$. 
	
	\begin{theorem} \label{prop:predict_competing}
		The predictive distribution of $\Delta_{n+1}$, conditional on $T_{n+1}=t$ and on $\mathcal{F}_n$, has probability mass function equal to
		\begin{multline}
			\label{eq:predict_competing}
			p(\delta \mid t, \mathcal{F}_n)
			\propto \sum_{j=1}^k k(t;X_j^*)\, \sum_{h=1}^{r_{\delta j}}\frac{\tau(q_{\delta jh}+1;K_{n+1}(X_j^*))}
			{\tau(q_{\delta jh};K_{n+1}(X_j^*))} \\
			+ \sum_{j=1}^k k(t;X_j^*)\, \tau(1;K_{n+1}(X_j^*))\,\frac{\tau_0(r_j+1;D\psi(K_{n+1}(X_j^*)))}
			{\tau_0(r_j;D\psi(K_{n+1}(X_j^*)))} \\
			+ \int_\X k(t;x)\,\tau(1;K_{n+1}(x))\,\tau_0(1;D\psi(K_{n+1}(x)))\, \Lambda_0(dx), 
		\end{multline}
		for any $\delta\in\{1,\ldots,D\}$, where $ \displaystyle K_{n+1}(x)=K_n(x)+\int_0^t k(s; x)\,ds$, since $T_{n+1}=t$. 		
	\end{theorem}

	\noindent The conditional distribution \eqref{eq:predict_competing} is essential to derive the prediction curves $p_n(\delta\mid\cdot\,)$ in Definition~\ref{def:pred_curve}, which provide a measure of how likely each future cause of death $\delta\in\set{1,\ldots,D}$ is as a function of the survival time. Although this derivation is not analytically possible, $p_n(\delta\mid\cdot\,)$ can be easily obtained from the MCMC algorithm in Section~\ref{sec:gibbs}.

	\subsection{Predictive structure of generalized gamma mixtures}
	
	To conclude, we specialize Theorems~\ref{prop:marginal} and~\ref{prop:predict_competing} to the noteworthy subclass obtained by specifying $\bm{\crm}$ as a \emph{generalized gamma} hCRM \eqref{generalized_gamma} and $k(\,\cdot\,; \,\cdot\,)$ as the Dykstra-Laud kernel. In this case, the Laplace exponent and cumulants \eqref{cumulants} equal $\psi(u) =\big((\beta + u)^\sigma - \beta^\sigma\big)/\sigma$ and $\tau(m; u) = (\beta + u)^{\sigma-m}\,\Gamma(m - \sigma)/\Gamma(1-\sigma)$, respectively. Note that $\lim_{\sigma\downarrow 0}	\psi(u)=\log(1+u/\beta)$, while $\psi_0$ and $\tau_0$ are obtained by replacing $\beta$ and $\sigma$ with $\beta_0$ and $\sigma_0$. Moreover, with the Dykstra-Laud kernel $k(t;x) = \gamma \, \indic_{\{t\ge x\}}$, the quantities in \eqref{kernel_quantities} reduce to
	\begin{gather*}
		Q_n(\boldsymbol T_{1:n}, \boldsymbol X_{1:n}) = \prod_{j=1}^k \prod_{i\in C_j^X}\,\gamma \,  \indic_{\{T_i\ge X_j^*\}}=\gamma^n \,\prod_{j=1}^k \indic_{\{\min_{i\in C_j^X} T_i\ge X_j^*\}}, \\
		K_n(x) = \gamma \, \sum_{i=1}^n \max(T_i - x, \, 0).
	\end{gather*}
	Therefore, the joint marginal distribution in \eqref{marginal} reduces to the expression \cut
	\begin{multline*}
		\Prob (\mathcal{F}_n) = \gamma^n \, \prod_{j=1}^k \left( \prod_{\ell=1}^D \prod_{h=1}^{r_{\ell j}} \frac{\Gamma(q_{\ell jh} - \sigma)}{\Gamma(1-\sigma)} \right) \frac{\Gamma(r_j - \sigma_0)}{\Gamma(1-\sigma_0)} \, \indic_{\{\min_{i\in C_j^X} T_i \, \ge X^*_j\}} \: \Lambda_0(dX_j^*) \\
		\times \prod_{j=1}^k \left(\beta_0 + \frac{D}{\sigma}  \Big( (\beta + K_n(X^*_j))^\sigma - \beta^\sigma \Big) \right)^{\sigma_0-r_j} (\beta + K_n(X^*_j))^{r_j \sigma-n_j} \\
		\times \, \exp \left( - \frac{1}{\sigma_0} \int_{\R^+} \left( \left( \beta_0 + \frac{D}{\sigma} \Big( (\beta + K_n(x))^\sigma - \beta^\sigma \Big) \right)^{\sigma_0} - \beta_0^{\sigma_0} \right) \Lambda_0(dx) \right).
	\end{multline*}
	Recall that the prior predictive distribution on $\Delta$ is non--informative, in the sense that $\Delta \mid (T = t)$ has a discrete uniform distribution on $\set{1,\ldots,D}$, for every $t \ge 0$. The predictive distribution for $\Delta_{n+1}$ in \eqref{eq:predict_competing} allows to assess how the data have modified our prior belief; its expression boils down to
	\begin{multline*}
		p(\delta \mid t,\mathcal{F}_n) 
		\propto 
		\sum_{j=1}^k\indic_{\{t\ge X_j^*\}}\,\left\{\frac{n_{\delta j}-r_{\delta j}\sigma}{\beta+K_{n+1}(X_j^*)}
		+\frac{(r_j-\sigma)\,(\beta+K_{n+1}(X_j^*))^{\sigma-1}}
		{\beta_0+\frac{D}{\sigma}((\beta+K_{n+1}(X_j^*))^\sigma-\beta^\sigma)}\right\} \\[4pt]
		+ \int_0^t (\beta+K_{n+1}(x))^{\sigma-1} \left(\beta_0 + \frac{D}{\sigma} \Big( (\beta + K_{n+1}(x))^\sigma - \beta^\sigma \Big) \right)^{\sigma_0-1}\Lambda_0(dx),
	\end{multline*}
	for any $\delta\in\set{1,\ldots,D}$, where notably $K_{n+1}(x)=K_n(x)+\gamma\,\max(t-x,0)$ for any $x$ in $\R^+$.
		
	\section{Posterior analysis}
	\label{sec:posterior}
		
	\subsection{General posterior characterization} 
	\label{sec:post_general}	
	
	The predictive system of Theorem~\ref{prop:predict_competing} provides an indirect glimpse of the posterior of the hCRM model \eqref{model}, though it was obtained from a probabilistic characterization of the underlying nested random partitions. In this section, we tackle the problem directly and derive a posterior characterization of $\bm{\crm}$ and $\crm_0$, conditionally on both the observations $(\bm{T}_{1:n},\bm{\Delta}_{1:n})$ and the latent variables $(\bm{X}_{1:n}, \bm{Z}_{1:n})$. This allows to estimate functionals of interest, such as the cause--specific incidence and subdistribution functions, as well as the survival function, and quantify the associated uncertainty.
	We show that $\crm_1, \dots, \crm_D$ and $\crm_0$ remain CRMs a posteriori, with significant updates, reflecting the information conveyed by the data. Indeed, each updated CRM consists of two independent components: a \textit{non-homogeneous} CRM without fixed points of discontinuity and a finite number of random jumps at fixed locations. Moreover, the hierarchical structure is preserved a posteriori, as the random measures at the bottom of the hierarchy are still conditionally independent given the random measure at the root. Consequently, our model is structurally conjugate in the sense of \cite{lijoi2010}, although the shift from \textit{homogeneous} to \textit{nonhomogeneous} CRMs marks a significant deviation from the prior structure; this effect does not occur with e.g.~normalized random measures, whether in the exchangeable \citep{james2009} or partially exchangeable \citep[e.g.][]{camerlenghi2019} settings.
	Recall that $\mathcal{F}_n=(\bm{T}_{1:n},\bm{\Delta}_{1:n},\Pi_{\bm{X}_{1:n}},\Pi_{\bm{Z}_{1:n}}, \bm{X}^*)$ and denote an updated quantity by the superscript $^+$. 

	\begin{theorem} \label{prop:posterior}
		Given the hCRM model \eqref{model}, then:
		\begin{itemize}
			\item[{\rm (a)}] the posterior distribution of $\crm_0$, given  $\mathcal{F}_n$, coincides with the distribution of
			\begin{equation}
				\label{eq:posterior_mu0}
				\crm_0^+ =\crm_0^*+ \sum_{j=1}^k V_j \,\delta_{X_j^*},
			\end{equation}
			where $\crm_0^*\sim\mbox{\rm CRM}(\nu_0^+)$ with $\nu_0^+(ds,dx) = e^{ - D \psi(K_n(x))\,s} \, \rho_0(ds)\,\Lambda_0(dx)$, and $V_1,\ldots,V_k$ are mutually independent r.v.'s with densities $f_{V_j}(s)\,ds \propto s^{r_j}\,\exp\{ - D \psi(K_n(X_j^*))\,s\} \, \rho_0(ds)$; 
			\item[{\rm (b)}] the posterior distribution of $\bm{\crm}$, given $\mathcal{F}_n$ and $\crm_0^+$ in \eqref{eq:posterior_mu0}, equals the distribution of a vector of random measures, whose $\delta$-th component is
			\begin{equation}
			\label{eq:posterior_mu_d}
				\crm_\delta^+=\crm_\delta^* + \sum_{j=1}^k \sum_{h=1}^{r_{\delta j}} S_{\delta jh} \,\delta_{X_j^*},
			\end{equation}
			where $\crm_\delta^* \mid \crm_0^+ \stackrel{\mbox{\rm \scriptsize iid}}{\sim} \mbox{\rm CRM}(\tilde\nu^+)$ with $\tilde\nu^+(ds,dx)= e^{-K_n(x)\,s} \, \rho(ds)\,\crm_0^+(dx)$ and $(S_{\delta j h})_{jh}$ are mutually independent r.v.'s with densities $f_{S_{\delta jh}}(s)\,ds \propto s^{q_{\delta jh}}\,e^{-K_n(X_j^*)\,s} \rho(ds)$. \end{itemize}		
	\end{theorem}

	\noindent Theorem~\ref{prop:posterior} shows that \textit{a posteriori} the hierarchical structure is preserved, though the CRMs at the hierarchy's bottom  are no longer conditionally identical in distribution. Indeed, $\crm_1^+, \dots, \crm_D^+$ are conditionally independent CRMs given $\crm_0^+$, but the component with fixed locations on the right-hand side of \eqref{eq:posterior_mu_d} depends on $\delta$, meaning conditional identity in distribution no longer holds. Nevertheless, $\bm{\crm}^*=(\crm_1^*,\ldots,\crm_D^*)$ remains a hCRM in the sense of Definition~\ref{def:hCRM}. Additionally, the result sheds light on the impact of the nested partition structure on the posterior distribution of hCRMs. At the hierarchy's root,
	the distribution of each jump $V_j$ depends on the multiplicity $r_j=r_{1j}+\,\cdots\,+r_{Dj}$, i.e.~the number of distinct marks across $\crm_1^e, \ldots, \crm_D^e$ associated with the latent location $X^*_j$. At the hierarchy's bottom, the distribution of each jump $S_{\delta jh}$ depends on the number of observations $q_{\delta jh}$ associated with cause $\delta$, location $X_j^*$ and latent mark $Z^*_{\delta jh}$. The interpretation is clear: (i) the posterior distribution of each $\crm_\delta$ depends on the (finer) partition $\Pi_{\bm{Z}_{1:n}}$ of the observations for which $\Delta_i = \delta$, i.e.~on the association of the observations with the marks; (ii) the posterior distribution of the root measure $\crm_0$ depends on how the finer partition $\Pi_{\bm{Z}_{1:n}}$ is nested within the coarser partition $\Pi_{\bm{X}_{1:n}}$, i.e.~on the association of the marks with the locations.
	Moreover, the components of the posterior with random locations, $\crm^*_1, \dots, \crm^*_D$ and $\crm^*_0$, are \textit{non-homogeneous}, as evidenced by the structure of $\tilde\nu^+$ and $\nu_0^+$. Indeed, the updates of their L\'evy intensities: (i) involve non-homogeneous exponential tiltings of the prior homogeneous L\'evy intensities; (ii)  depend solely on the observed survival times $\boldsymbol T_{1:n}$ via $x \mapsto K_n(x)=K_n(x; \boldsymbol T_{1:n})$ in \eqref{kernel_quantities}, and remain unaffected by the latent partition structure. Notably, at both levels of the hierarchy, the distributions of the jumps depend on the $\bm{Z}$-level partition and on the fixed latent location $X^*_j$ through $K_n(X^*_j)$, for $j \in \set{1,\dots,k}$.

	\subsection{Posterior characterization of generalized gamma mixtures} 
	\label{sec:post_GG}

	We specialize the general characterization to the subclass of \emph{generalized gamma} hCRMs \eqref{generalized_gamma} discussed in Section~\ref{sec:marginal}: the prior to posterior update becomes straightforward. Specifically, $\crm^*_1, \dots, \crm^*_D$ and $\crm^*_0$ remain (extended) generalized gamma CRMs, with the same $\sigma$ and $\sigma_0$ and updated rate parameters. The latter become functions on the latent space $\X$, due to nonhomogeneity \emph{a posteriori}, and depend only on $\bm{T}_{1:n}$ through $K_n$. Indeed, $\beta^+(x) = \beta + K_n(x)$ and $\beta_0^+(x) = \beta_0 + D \{(\beta + K_n(x))^\sigma - \beta^\sigma\}/\sigma$. 
	Furthermore, for each fixed latent location $X^*_j$, the jumps $(S_{\delta jh})_{\delta h}$ and $V_j$ are gamma distributed
	\begin{equation*}
		S_{\delta jh} 
		\sim \text{Gamma} \big(q_{\delta jh} - \sigma, \, \beta^+(X_j^*) \big), \qquad V_j 
		\sim \text{Gamma} \big( r_j - \sigma_0, \, \beta_0^+(X^*_j) \big).
	\end{equation*}
	The generalized gamma hCRM can thus be regarded as the (conditionally) conjugate prior when used as mixing measure for the kernel mixtures defining cause--specific hazards in \eqref{model}. Moreover, with the Dykstra-Laud kernel, $\displaystyle K_n(x) = \gamma \, \sum_{i=1}^n \max(T_i - x, \, 0)$ is a non-increasing function and $K_n(x)=0$ for any $x\ge \max_i T_i$. This, combined with $u\mapsto \psi(u)$ being non--decreasing, implies that the effect of the exponential tilting update in the posterior L\'evy intensities $\tilde\nu^+$ and $\nu_0^+$ is larger for smaller values of $x$ and becomes negligible for larger values of $x$. In other words, for any $\varepsilon>0$, the function $\displaystyle x\mapsto \rho((\varepsilon,+\infty ))-\int_\varepsilon^{+\infty}e^{-K_n(x)s}\rho(ds)$ is non-increasing and actually equals $0$ if $x\ge \max_i T_i$. Hence, for large enough $x$, the prior and posterior Lévy intensities coincide. The same argument applies to $\rho_0$ at the hierarchy's root.
	Moreover, each of the fixed latent locations $X^*_1, \dots, X^*_k$ is necessarily smaller than the largest observed survival time, due to the kernel product term $Q_n(\boldsymbol T_{1:n}, \boldsymbol X_{1:n})$. Therefore, the observations do not provide any information beyond the largest survival time $\max_i T_i$, where the prior and posterior distributions of the random measures coincide. Similar considerations apply to other kernels mentioned in Section~\ref{sec:mixture_hazard}. See also \cite{deblasi2009} for related discussions. 

	\subsection{Conditional functionals estimation given the latent variables}
	\label{sec:functionals}

	Theorem~\ref{prop:posterior} represents the key ingredient for the estimation of important functionals of $\boldsymbol{\crm}$ such as the survival function \eqref{survival} and the cause--specific incidence functions \eqref{directing_measure}. Their posterior estimates, given the observations and the latent variables, are obtained by marginalization with respect to the posterior random measures. To this end, we conveniently single out the jump part of the posterior L\'evy intensities $\tilde\nu^+$ and $\nu_0^+$ as $\rho^+(ds \,\vert\, x) = \exp(-K_n(x)\,s) \, \rho(ds)$ and $\rho_0^+(ds \,\vert\, x) = \exp(- D \psi(K_n(x))\,s)\, \rho_0(ds)$, with corresponding Laplace exponents and cumulants given by
	\begin{equation} \label{eq:cumulants_posterior}
		\psi_x^+(u) = \int_{\R^+} (1-e^{-us})\,\rho^+(ds \,\vert\, x), \qquad \tau_x^+(m; u) = \int_{\R^+} s^m\,e^{-us}\, \rho^+(ds \,\vert\, x);
	\end{equation}
	the quantities $\psi_{0,x}^+$ and $\tau_{0,x}^+$ are obtained replacing $\rho^+$ with $\rho^+_0$ in \eqref{eq:cumulants_posterior}.
	
	First, consider the survival function $\tilde S(t)$ in \eqref{survival}. Since $\tilde S$ is the exponential of a linear functional of $\crm_1,\ldots,\crm_D$, its posterior estimate, under squared error loss, is an immediate consequence of Theorem~\ref{prop:posterior}. 
	
	\begin{proposition} \label{prop:posterior_survival}
		The posterior estimate, given $\mathcal{F}_n$, of the survival function $\tilde S(t)$ at $t\ge0$ is
		\begin{multline*}
			\E \left[ \tilde S(t) \mid \mathcal{F}_n \right] = \exp \left( - \int_\X \psi_{0,x}^+ \big( D \psi_x^+(K_1(x)) \big) \, \Lambda_0(dx) \right) \\
			\times \prod_{j=1}^k \left( \prod_{\ell=1}^D \prod_{h=1}^{r_{\ell j}} \frac{\tau_{X_j^*}^+(q_{\ell jh}; K_1(X_j^*))}{\tau_{X_j^*}^+(q_{\ell jh}; 0)} \right) \frac{\tau_{0,X_j^*}^+\big(r_j; D \psi_{X_j^*}^+(K_1(X_j^*))\big)}{\tau_{0,X_j^*}^+(r_j; 0)},
		\end{multline*}
		where $\displaystyle K_1(x)=K_1(x;t) = \int_0^t k(s;x)\,ds$, for any $t \ge 0$ and $x$ in $\X$.
	\end{proposition}

	\noindent Moreover, we obtain a posterior estimate of the cause--specific incidence function $\tilde f_\delta(t)$ in \eqref{directing_measure}. Although its dependence on $\crm_1, \ldots, \crm_D$ is complex, marginalization remains feasible.

	\begin{proposition} \label{prop:posterior_incidence}
		For any $t\ge0$ and $\delta\in \set{1,\ldots,D}$, the posterior estimate of $\tilde p(dt, \delta) = \tilde f_\delta(t)\,dt$ is
		\begin{multline*}
			\E \left[ \, \tilde p(dt, \delta) \mid \mathcal{F}_n \right] = \E \left[ \tilde S(t) \mid \mathcal{F}_n \right] \bigg\{ \sum_{j=1}^k k(t;X_j^*) \, \sum_{h=1}^{r_{\delta j}} \, \frac{\tau_{X_j^*}^+(q_{\delta jh}+1; K_1(X_j^*))}{\tau_{X_j^*}^+(q_{\delta jh}; K_1(X_j^*))} \\
			+ \sum_{j=1}^k k(t;X_j^*) \, \tau_{X_j^*}^+(1; K_1(X_j^*)) \, \frac{\tau_{0,X_j^*}^+( r_j+1; D \psi_{X_j^*}^+(K_1(X_j^*)))}{\tau_{0,X_j^*}^+( r_j; D \psi_{X_j^*}^+(K_1(X_j^*)))} \\
			+ \int_\X k(t;x) \,\tau_x^+(1; K_1(x)) \,\tau_{0,x}^+(1; D \psi_x^+(K_1(x))) \,\Lambda_0(dx) \bigg\} \,dt.
		\end{multline*}
	\end{proposition}

	\begin{remark} \label{remark}
		Proposition~\ref{prop:posterior_incidence} also allows to evaluate the predictive distribution for the event type $\Delta_{n+1}$ associated with the additional survival time $T_{n+1}$ given $(T_{n+1},\mathcal{F}_n)$, which unsurprisingly coincides with $p(\delta \mid t,\mathcal{F}_n)$ in Theorem~\ref{prop:predict_competing}. Details are provided in Section~\ref{supsec:proofs}.
	\end{remark} 

	\noindent Propositions~\ref{prop:posterior_survival} and~\ref{prop:posterior_incidence} provide explicit posterior estimates, yet conditional on the latent variables $(\Pi_{\bm{X}_{1:n}},\Pi_{\bm{Z}_{1:n}}, \bm{X}^*)$. To obtain posterior estimates depending only on the observations $(\bm{T}_{1:n},\bm{\Delta}_{1:n})$, one needs to integrate these expressions with respect to the posterior distribution of these latent variables. This marginalization is performed via the Gibbs sampling scheme detailed in the next section.   

	\section{Simulation algorithms for posterior inference}
	\label{sec:gibbs}

	The latent partition structure identified in Section~\ref{sec:marginal} is the key device to characterize marginal, predictive and posterior quantities, which are expressed conditionally on $(\Pi_{\bm{X}_{1:n}}, \Pi_{\bm{Z}_{1:n}}, \bm{X}^*)$. 
	To obtain unconditional quantities, we employ a Gibbs sampling scheme that generates these latent variables conditionally on the observations. Specifically, we derive a marginal algorithm, which leverages the explicit expression of the law of $\mathcal{F}_n^*$ obtained in Theorem~\ref{prop:marginal} by integrating out the underlying hCRM. We also present a sampling procedure for generating trajectories of $\crm_0^+$ and $\bm{\crm}^+$ from their distributions in Theorem~\ref{prop:posterior}, conditionally on the partition structure. These algorithms are tailored to the \emph{generalized gamma} hCRM in Section~\ref{supsec:sampling}.

	\subsection{Marginal sampling algorithm}
	\label{sec:marginalMCMC}

	The marginal likelihood in Theorem~\ref{prop:marginal} is the starting point for devising a Gibbs sampler that generates the posterior latent variables $(X_i,Z_i)$, for $i \in \set{1, \dots, n}$. Marginal structures are typically characterized in the nonparametric literature via sequential urn schemes. In the exchangeable case, these are known as generalized P\'olya urn schemes or Chinese restaurant processes, derived for the Dirichlet process in \cite{blackwell1973}, the Pitman-Yor process in \cite{pitman1996} and normalized CRMs in \cite{james2009}. In the partially exchangeable case, with a hierarchical dependence structure, they are known as Chinese restaurant franchise processes, and were introduced by \cite{teh2006} for the hierarchical Dirichlet process and by \cite{camerlenghi2019} for hierarchical normalized CRMs.
	
	Although we also aim at an urn--type sequential description of the marginal structure, its derivation requires a fundamentally different approach. Indeed, since the underlying model does not build upon normalized (hierarchical) CRMs, we should rely on different tools as well as on the essential augmentation via the latent marks $\bm{Z}_{1:n}$. The full conditionals correspond to the predictive distributions deduced from Theorem~\ref{prop:marginal} and the marginal of $\mathcal{F}_n$ in \eqref{marginal}. 
	Specifically, the full conditional generating $(X_i,Z_i)$ can be described in terms of the standard notation where $q_{\delta jh}^{-i}$, $r_{\delta j}^{-i}$,  $r_j^{-i}$ and $k^{-i}$ have the same frequency interpretation as in Theorem~\ref{prop:marginal} after removing $(X_i,Z_i)$. 
	Similarly, we set $\mathcal{F}_n^{-i} = (\bm{T}_{1:n}, \bm{\Delta}_{1:n}, \bm{X}_{1:n}^{-i},\bm{Z}_{1:n}^{-i})$ with the $i$-th component removed from $\bm{X}_{1:n}$ and $\bm{Z}_{1:n}$. Adopting this notation, the predictive distribution of $(X_i,Z_i)$ is identified by the following cases:

	\begin{itemize}
		\item[(1)] $X_{i}$ and $Z_{i}$ coincide with conditioning  latent variables, say $X_{i}= X_j^*$ and $Z_{i}=Z_{\delta jh}^*$, for some $j$ and $h$ and such that $\Delta_i = \delta$, with probability
		\begin{equation*}
			\Prob \big( X_{i}= X_j^*, \, Z_{i} = Z_{\delta jh}^* \mid \mathcal{F}_n^{-i} \big) \propto k(T_i; X_j^*) \, \frac{\tau(q_{\delta jh}^{-i}+1; K_{n}(X_j^*))}{\tau(q_{\delta jh}^{-i}; K_{n}(X_j^*))};
		\end{equation*}
		\item[(2)] $X_{i}$ equals one of the tied values in $\boldsymbol X_{1:n}^{-i}$, say $X_{i} = X_j^*$, while $Z_{i}$ takes on a new value, not included in $\boldsymbol Z_{1:n}^{-i}$, with probability
		\begin{equation*}
			\Prob \big( X_{i} = X_j^*, \, Z_{i} = \text{`new'} \mid \mathcal{F}_n^{-i} \big) \propto k(T_{i}; X_j^*) \,\tau(1; K_{n}(X_j^*)) \, \frac{\tau_0\big( r_j^{-i}+1; D \psi(K_{n}(X_j^*)) \big)}{\tau_0 \big( r_j^{-i}; D \psi(K_{n}(X_j^*)) \big)};
		\end{equation*}
		\item[(3)] $X_{i}$ and $Z_{i}$ take on new values, i.e.~not included in $\boldsymbol X_{1:n}^{-i}$ and $\boldsymbol Z_{1:n}^{-i}$, respectively, with probability
		\begin{equation*}
			\Prob \big( X_{i} = \text{`new'}, \, Z_{i} = \text{`new'} \mid \mathcal{F}_n^{-i} \big) \propto \int_\X k(T_{i}; x) \, \tau(1; K_{n}(x))\, \tau_0 \big( 1; D \psi(K_{n}(x)) \big) \, \Lambda_0(dx).
		\end{equation*}
	\end{itemize}

	\noindent Recall from Section~\ref{sec:marginal} that the new value possibly assigned to the latent mark $Z_{i}$ does not affect the predictive distributions and, hence, the full conditionals: the only relevant information in $\bm{Z}_{1:n}$
	is the induced partition structure. In contrast, the new value possibly assigned to the latent $X_{i}$ in case (3) directly impacts the allocation probabilities in the Gibbs sampling scheme, and is sampled from the diffuse probability measure proportional to $k(T_{i}; x) \, \tau(1; K_n(x))\, \tau_0 \big( 1; D \psi(K_n(x)) \big) \, \Lambda_0(dx)$; with the kernel choices considered in Section~\ref{sec:mixture_hazard}, this reduces to a univariate density on a bounded interval.
	Each iteration of the algorithm yields an update of the latent nested partition structure $(\Pi_{\bm{X}_{1:n}}, \Pi_{\bm{Z}_{1:n}})$. A further \emph{acceleration step}, which is known to improve mixing \citep{escobar1998, maceachern1998}, consists in resampling independently each latent location $X^*_j$ for $j \in \set{1, \dots, k}$, given the observations and latent partition structure, proportionally to 
	\begin{equation*}
		\bigg( \prod_{i \in C_j^X} k(T_i; x) \bigg) \left( \prod_{\ell=1}^D \prod_{h=1}^{r_{\ell j}} \tau(q_{\ell jh}; K_n(x)) \right) \tau_0 \big( r_j ; D \psi(K_n(x)) \big) \, \Lambda_0(dx).
	\end{equation*}
	Finally, hyperpriors are usually specified for the parameters of the Lévy intensities $\tilde\nu$ and $\nu_0$ of the underlying hCRM \eqref{hierarchical_crm} and the kernel $k(\,\cdot\,; \,\cdot\,)$. 
	In particular, if $\Lambda_0(dx) = \theta\,dx$ is proportional to the Lebesgue measure, a gamma prior on $\theta$ is conjugate to the marginal distribution \eqref{marginal}. Indeed, if $\theta\sim\text{Gamma}(a,b)$, its full conditional depends only on the observed survival times $\boldsymbol T_{1:n}$ and the number $k$ of distinct latent locations, being $${\rm Gamma}\left( a + k, \ b + \int_\X \psi_0 ( D \psi(K_n(x))) \, dx \right).$$
	Posterior estimates of the survival function, cause--specific incidence functions and prediction curves are then readily obtained. Specifically, at each iteration, (conditional) posterior estimates are computed by plugging the updated latent variables, resampled from their posterior distribution, into the expressions of Propositions~\ref{prop:posterior_survival} and~\ref{prop:posterior_incidence} and Theorem~\ref{prop:predict_competing}; thus, unconditioned posterior estimates are obtained as Monte Carlo averages.

	\subsection{Posterior conditional sampling of random measures}
	\label{sec:conditionalMCMC}
	
	Marginal algorithms cannot be used for uncertainty quantification around posterior estimates of functionals involving the random measures, as marginalizing over $\crm_1,\ldots,\crm_D$ averages out the associated uncertainty. Indeed, the conditional estimates computed at each iteration as in Section~\ref{sec:functionals} are not a sample from the posterior of the corresponding functionals, but rather of their conditional expectations, given the latent variables. 
	To address this issue, one needs to generate samples from the actual posterior of the random measures: this properly accounts for posterior uncertainty and enables the computation of reliable credible bands. Specifically, at each iteration of the marginal algorithm, we sample independently from the posterior hCRMs, given the observations and latent variables. Based on these samples, valid posterior draws of functionals can be computed from their definitions \eqref{survival} and \eqref{directing_measure}.
	
	The practical implementation of this method requires sampling from the posterior hCRMs in Theorem~\ref{prop:posterior}. The Ferguson-Klass algorithm \citep{ferguson1972} is the most popular approach for generating realizations of CRMs, as it ensures their jumps are a.s.~decreasing, allowing to control the associated approximation error. See \cite{barrios2013} and references therein, and \cite{bruck2025} for a novel approach to sampling based on generative neural networks. For non-homogeneous CRMs, \cite{wolpert1998} proposed a simpler alternative, which however does not generate a.s.~decreasing jumps. Nevertheless, when non--homogeneity is induced by exponential tilting of a homogeneous prior, as in our case, the sequence of jumps can still be upper--bounded by a decreasing sequence obtained from the prior \citep{camerlenghi2021}. This theoretical guarantee on the truncation error underpins our sampling method.

	Consider the non-homogeneous random measure $\crm^+_0$ at the root of the posterior hierarchy (Theorem~\ref{prop:posterior}), and perform the following steps: 
	\begin{itemize}
		\item[(i$_0$)] generate the sequence of random locations $(X^{(0)}_h)_{h \ge 1}$ independently from $\Lambda_0/\Lambda_0(\X)$;
		\item[(ii$_0$)] generate the sequence $(N^{(0)}_h)_{h \ge 1}$ of jump times of a unit--rate Poisson process, namely $N^{(0)}_0 = 0$ and, for $h \ge 1$, the increments $N^{(0)}_h - N^{(0)}_{h-1}$ are i.i.d.~exponentially distributed with unit mean;
		\item[(iii$_0$)] generate the sequence of random jumps $(V^{(0)}_h)_{h \ge 1}$ by solving the equations
		\begin{equation*}
			N^{(0)}_h = \Lambda_0(\X) \int_{V^{(0)}_h}^\infty \exp \left(- D \psi \big( K_n\big(X^{(0)}_h\big) \big)\,s \right) \rho_0(ds), \qquad h \ge 1.
		\end{equation*}
	\end{itemize}
	Since one cannot generate infinite sequences in (i$_0$)--(iii$_0$), the algorithm is stopped as soon as
		$\displaystyle N^{(0)}_{H_0} > \Lambda_0(\X) \int_{\varepsilon}^\infty \rho_0(ds)$,
	for some $H_0 \ge 1$. This guarantees each of the discarded jumps $V^{(0)}_{H_0+1}, V^{(0)}_{H_0+2}, \dots$ to be smaller than a fixed threshold $\varepsilon > 0$. 
	Moreover, if $\Lambda_0$ is an infinite measure, the sampling procedure must be restricted to a suitable subset $\mathbb{S} \subseteq \X$ such that $\Lambda_0(\mathbb{S}) < \infty$; for the kernels in Section~\ref{sec:mixture_hazard}, with $\X = \R^+$, we consider $\mathbb{S} = [\,0,T_{\rm max}\,]$ where $T_{\rm max}$ is larger than the largest observed survival time. An approximate realization of the root random measure $\crm_0$ from its posterior distribution is thus 
	\begin{equation} \label{approx_root}
		\crm_0^+ \approx \sum_{j=1}^k V_j \,\delta_{X_j^*}+ \sum_{h=1}^{H_0} V^{(0)}_h \, \delta_{X^{(0)}_h}, 
	\end{equation}
	where the jumps $V_1, \dots, V_k$ at the fixed locations $X^*_1, \dots, X^*_k$ are sampled from the corresponding densities in Theorem~\ref{prop:posterior}. Note that, since $\Lambda_0$ is diffuse, the random locations $X^{(0)}_1, \dots, X^{(0)}_{H_0}$ are distinct a.s., and also differ from the fixed locations $X^*_1, \dots, X^*_k$. Similarly, at the bottom of the posterior hierarchy, for each non-homogeneous $\crm^+_\delta$ with $\delta \in \set{1,\ldots,D}$:
	\begin{itemize}
		\item[(i$_\delta$)] generate $(X^{(\delta)}_h)_{h \ge 1}$ independently from $\crm_0^+/\crm_0^+(\X)$; \item[(ii$_\delta$)] generate the sequence  $(N^{(\delta)}_h)_{h \ge 1}$
		of jump times of a unit--rate Poisson process as in ii$_0$);
		\item[(iii$_\delta$)] generate $(S^{(\delta)}_h)_{h \ge 1}$ by solving the equations
		\begin{equation*}
			N^{(\delta)}_h = \crm_0^+(\X) \int_{S^{(\delta)}_h}^\infty e^{ - K_n(X^{(\delta)}_h)\,s} \rho(ds), \qquad h \ge 1,
		\end{equation*}
	\end{itemize}
	The procedure above is stopped as soon as
		$\displaystyle N^{(\delta)}_{H_\delta} > \crm_0^+(\X) \int_{\varepsilon}^\infty \rho(ds)$,
	for some $H_\delta \ge 1$ and some fixed $\varepsilon > 0$. Clearly, in (i$_\delta$)--(iii$_{\delta}$) we replace $\crm_0^+$ with its finite approximation on the right--hand--side of \eqref{approx_root}. Hence, an approximate realization of $\crm_\delta$ from its posterior distribution is given by
	\begin{equation} \label{approx_lower}
		\crm_\delta^+ \approx \sum_{j=1}^k \sum_{h=1}^{r_{\delta j}} S_{\delta jh} \,\delta_{X_j^*}+ \sum_{h=1}^{H_\delta} S^{(\delta)}_h \, \delta_{X^{(\delta)}_h},
	\end{equation}
	where the jumps at fixed locations are sampled as specified in Theorem~\ref{prop:posterior}. 
	
	\begin{remark}
		The truncation of the infinite sequence of jumps in \eqref{approx_lower}, and the conditioning on a truncated realization of $\crm_0^+$, entail two levels of approximation. Moreover, the random locations $X^{(\delta)}_1, \dots, X^{(\delta)}_{H_\delta}$ display ties with positive probability, and take values among the finite collection of locations $X^{(0)}_1, \dots, X^{(0)}_{H_0}$ and $X^*_1, \dots, X^*_k$ characterizing the discrete measure approximating $\crm_0^+$. An alternative approach for sampling realizations of each $\crm^+_\delta$ consists in sampling its $k+H_0$ cumulative random jumps associated to the locations inherited from the approximation in \eqref{approx_root}. This exact sampling procedure avoids the approximation at the root of the hierarchy, and is explored for normalized random measures in \cite{lijoi2020} and more recently in \cite{catalano2025}. 
	\end{remark}

	\section{Applications to synthetic and real datasets} \label{sec:simulation_study}
	
	We consider synthetic datasets to assess the performance of our Bayesian nonparametric model, compared to both frequentist and Bayesian alternatives. This evaluation underscores how the borrowing of information across hazard rates, induced by the hierarchical structure, enhances the accuracy of posterior estimates. 
	Moreover, we present an application to a real dataset obtained from the bone marrow transplant registry of the European Blood and Marrow Transplant (EBMT) Group. 
	Section~\ref{supsec:applications} provides further details on the simulated datasets, diagnostic plots, and additional figures.
	Another application to a dataset of melanoma patients is reported in Section~\ref{supsec:melanoma}. Finally, we compare our method to that of \cite{sparapani2020} on both synthetic and real datasets in Section~\ref{supsec:sparapani}.
	
	\subsection{A simulation study with three competing risks}
	\label{sec:simulation_three_risks}
	
	First, we consider a simulation scenario with $D = 3$ competing sources of risk. The simulated data are generated according to the latent failure times approach (see Section~\ref{sec:intro}). In particular, 
	\begin{equation} \label{data_generation}
		\begin{gathered}
			Y_{i,1} \stackrel{\mbox{\scriptsize iid}}{\sim} \text{Weibull} \big(\xi_{1} = 1.2, \,\lambda_{1}=1\big), \qquad
			Y_{i,2} \stackrel{\mbox{\scriptsize iid}}{\sim} \text{Weibull}\big(\xi_{2} = 1.6,\,\lambda_{2}=1\big), \\
			Y_{i,3} \stackrel{\mbox{\scriptsize iid}}{\sim} \text{Weibull}\big(\xi_{3} = 2.4,\,\lambda_{3}=1\big),
		\end{gathered}	
	\end{equation}
	with $i=1,\ldots,n$, where $\xi_{\delta}$ and $\lambda_{\delta}$ are shape and scale parameters, for $\delta\in\set{1,2,3}$. The observed survival time is $T_i=\min\{Y_{i,1},Y_{i,2},Y_{i,3}\}$ and $\Delta_i = \delta$ if and only if $T_i = Y_{i,\delta}$. The sample size is $n = 300$. Note that the number of subjects incurring in each competing event is a r.v.~itself, and the parameters in \eqref{data_generation} are chosen to yield a fairly balanced allocation of the observations to the different event types.
	
	We rely on the \emph{generalized gamma} hCRM prior with Dykstra-Laud kernel, and implement both algorithms outlined in Section~\ref{sec:gibbs}. The hyperparameters of our model are $\beta = \beta_0 = 1.0$ and $\sigma = \sigma_0 = 0.25$, while $\Lambda_0(dx) = \theta\,dx$ for $\theta > 0$ is proportional to the Lebesgue measure. Moreover, exponential hyperpriors with large means, reflecting a non--informative specification, are assigned to both $\theta$ and the kernel parameter $\gamma$. 
	The MCMC procedure is run for $25,000$ iterations, with a burn-in of $5,000$; after thinning, we collect $2,000$ samples from the latent partition structure to compute posterior estimates with the marginal algorithm.  
	For the quantification of posterior uncertainty, we generate $10$ independent realizations from the posterior of the hCRMs, for each sample from the latent structure. Diagnostic tools suggest convergence of the Markov chains, which show a fairly good mixing; indeed, for the marginal algorithm, the average effective sample size for posterior survival function estimates, evaluated on a grid of points, is about $1,393$. 
	
	\begin{figure}
		\centering
		\includegraphics[width=0.6\textwidth]{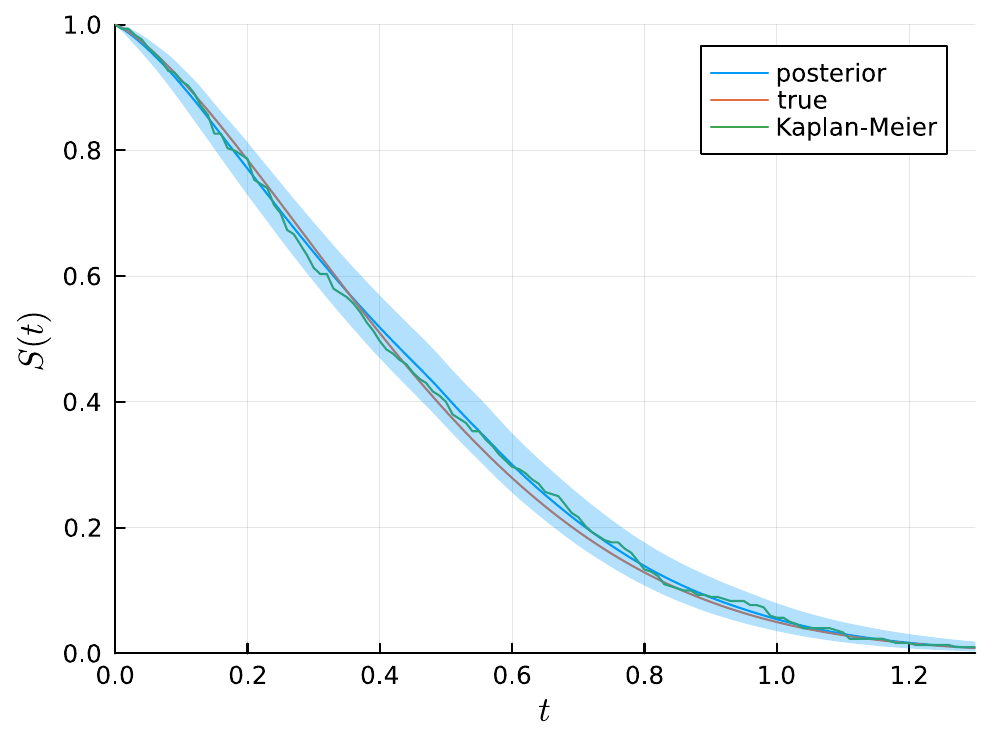}
		\captionsetup{width=0.85\textwidth,font=small}
		\caption{Survival function for a synthetic dataset with $n=300$ data points: posterior estimates obtained via the marginal algorithm (blue), true curve (orange), and Kaplan-Meier estimate (green); the pointwise $0.95$ posterior credible band is also displayed.} 
		\label{fig:survival}
	\end{figure}

	Figure~\ref{fig:survival} compares the true survival function implied by \eqref{data_generation}, its Kaplan-Meier estimate and our Bayesian posterior estimate, obtained via the marginal algorithm. The latter successfully recovers the ground truth: its Kolmogorov distance from the true curve ($0.0242$) is substantially smaller than that of the Kaplan-Meier estimate ($0.0330$). Furthermore, the pointwise $0.95$ credible band entirely covers the true function. 
	Figure~\ref{fig:incidence} compares the ground truth with the posterior estimates of the cause--specific incidence functions, subdistributions and prediction curves. The proposed model accurately recovers the true curves. Indeed, the crossing time (around $t = 0.6$) of the cause--specific incidence functions is correctly detected, with pointwise credible bands slightly overlapping for times smaller than the crossing. Moreover, the posterior estimates of the subdistribution functions align with the frequentist Aalen-Johansen nonparametric estimates, and the corresponding pointwise credible bands entirely contain the true subdistribution functions. 
	Finally, the prediction curves slightly depart from the true proportions for large values of $t$, as the observations do not inform prediction curves beyond the largest observed survival time; refer to the discussion in Section~\ref{sec:post_GG}.
	
	\begin{figure}[t]
		\centering
		\noindent\makebox[\textwidth]{
			\includegraphics[width=0.5\textwidth]{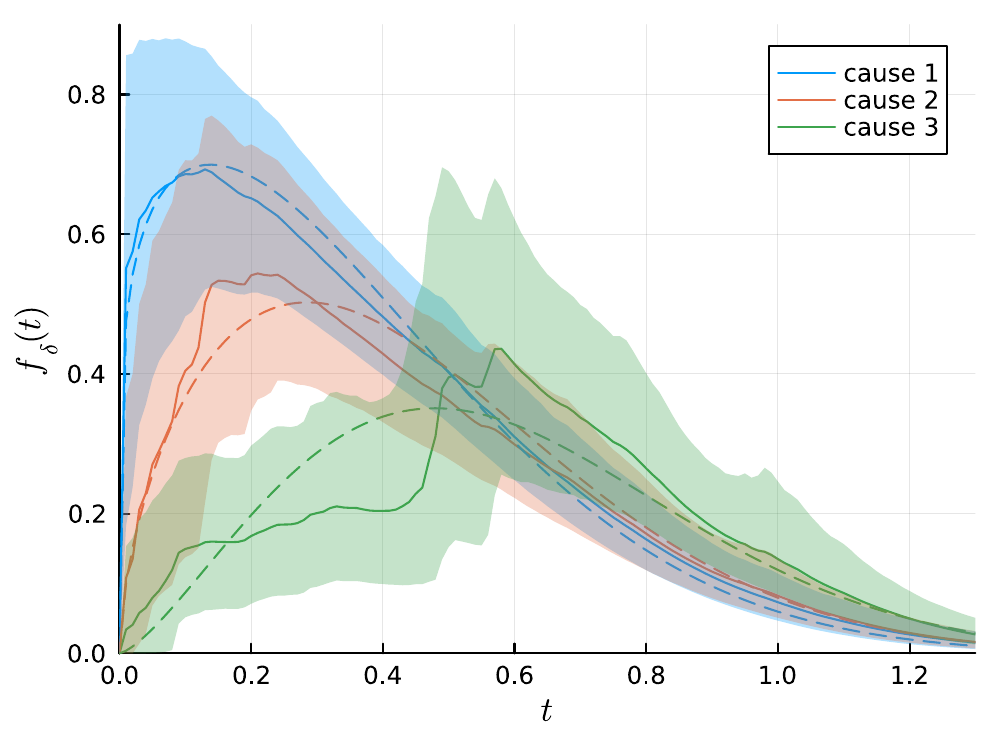}
			\includegraphics[width=0.5\textwidth]{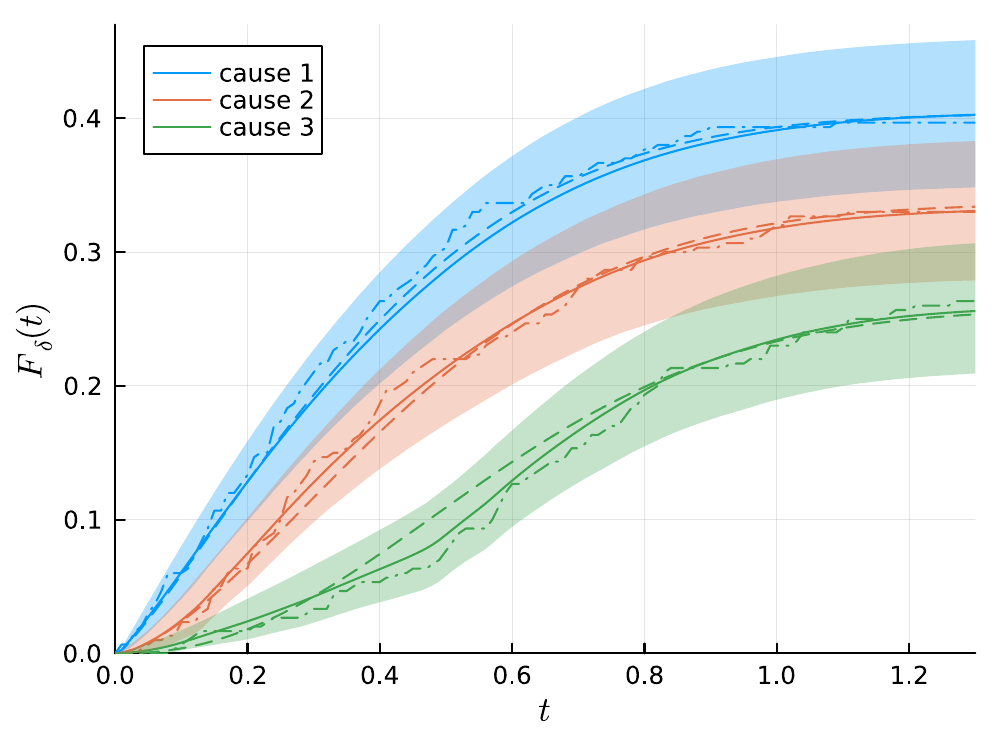}
		}
		\includegraphics[width=0.5\textwidth]{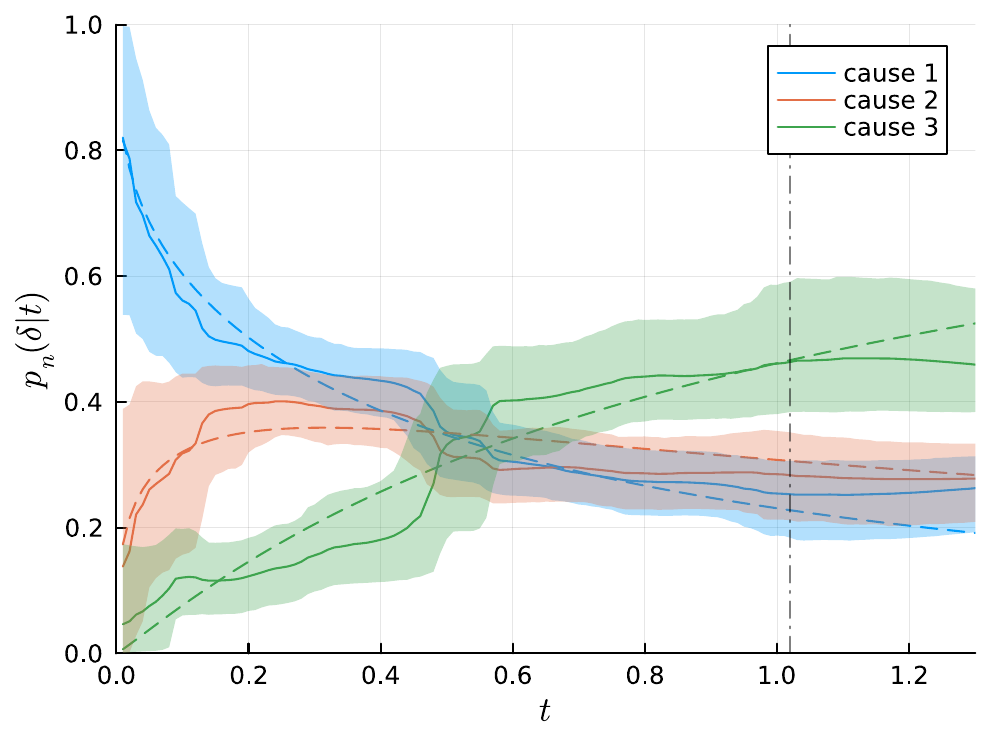}
		\captionsetup{width=0.85\textwidth,font=small}
		\caption{Cause--specific incidence functions (left), subdistributions (middle) and prediction curves (right); true curves (dashed) and posterior estimates (solid), with corresponding pointwise $0.95$ credible bands; estimates of subdistributions are also compared with the Aalen-Johansen estimate (dash-dotted).}
		\label{fig:incidence}
	\end{figure}

	\subsection{Dependent vs.~independent priors}
	\label{sec:simulation_independent}
	
	The proposed model assumes dependence among hazard rates associated to different sources of risk. This is one of its main strengths, as it allows for sharing of information between hazards. We illustrate this phenomenon by means of a simulation study, and contrast the performance of our model with a Bayesian nonparametric alternative having independent cause--specific hazard rates. The latter can be regarded as a special case of hierarchical prior distribution over the vector of random measures in \eqref{model}, namely $\boldsymbol{\tilde \mu} = (\tilde \mu_1, \dots, \tilde \mu_D) \mid \crm_0 \, \overset{\text{\scriptsize iid}}{\sim} \, \text{CRM}(\nu)$, where $\nu(ds,dx) = \rho(ds) \, \Lambda_0(dx)$ is a fixed homogeneous L\'evy intensity. 
	Therefore, $\crm_1,\ldots,\crm_D$ are also unconditionally i.i.d., since $\nu$ does not depend on $\crm_0$. Such prior specification implies that the posterior estimate of the survival function for $T$ coincides with the product of the estimates of $\displaystyle \exp\left(-\int_0^t \int_\X k(s,x)\,\crm_\delta(dx)\, ds\right)$, for each $\delta \in \set{1,\ldots,D}$. These can be interpreted as the survival functions for each event type, in the hypothesis of absence of competing risks, i.e., considering the occurrences of competing events as censored observations.
	
	Consider a data--generating model involving $D = 3$ competing and independent sources of risk. For each event type $\delta \in \set{1,2,3}$, the latent survival time $Y_{i,\delta}$ is sampled from a mixture of a cause--specific and a common distribution, with equal weights; see Section~\ref{supsec:applications} for details. The numerical results are obtained by averaging over $100$ simulated datasets, each consisting of $n = 100$ observations, generated according to the potential survival times approach. The hyperparameters specification coincides with the one adopted in Section~\ref{sec:simulation_three_risks}.
	
	\begin{table}
		\centering
		\begin{tabular}{llcccccc}
			\hline
			&	& \multicolumn{3}{c}{incidence function} & \multicolumn{3}{c}{subdistribution} \\ 
			\cline{3-8} 
			&	& $\delta=1$	& $\delta=2$ & $\delta=3$ & $\delta=1$	& $\delta=2$ & $\delta=3$ \\ 
			\hline
			\multirow{2}{*}{hierarchical} 
			& marginal 		& \textbf{.0952} & \textbf{.0833} & .0963 & \textbf{.1274} & \textbf{.1179} & .1280	\\
			& conditional  	& .0964 & .0836 & \textbf{.0956} & .1293 & .1185 & \textbf{.1277} \\ 
			\hline
			\multirow{2}{*}{independent}
			& marginal 		& .1081 & .0932 & .0974 & .1385 & .1291 & .1339 \\
			& conditional  	& .1085 & .0934 & .0971 & .1392 & .1298 & .1338 \\ 
			\hline
			frequentist	& 	& 	& 	& 	& .1810 & .1976 & .2060 \\
			\hline
		\end{tabular}
		\captionsetup{font=small}
		\caption{Comparison of rescaled total variation distances between estimated and true incidence functions, and rescaled Kolmogorov distances between estimated and true subdistribution functions, for the hierarchical and the independent model, using marginal and conditional estimation approaches. Distances are averaged over $100$ simulated datasets.}
		\label{table:errors}
	\end{table}
	
	Table~\ref{table:errors} provides a comparison of the performance of the proposed model with the \emph{independent} model. Bayesian estimates are computed either via the marginal algorithm in Section~\ref{sec:marginalMCMC} or averaging over the posterior samples generated via the conditional sampling approach in Section~\ref{sec:conditionalMCMC}. In particular, the average errors in estimating the cause--specific incidence function and subdistribution for each event type are reported for both models and both estimation strategies. In addition, the average estimation error for the Aalen-Johansen estimator of the subdistribution functions is also displayed. The comparison is based on total variation and Kolmogorov distances, which are rescaled according to the true probabilities of occurrence of each event type. Specifically, we consider the measures of error
	\begin{equation*}
		e^{\text{TV}}_\delta = \frac{1}{2 \pi_\delta} \int_{\R^+} \vert \hat f_\delta(t) - f_\delta(t) \vert \, dt, \qquad e^{\text{K}}_\delta = \frac{1}{\pi_\delta} \sup_{t \in \R^+} \vert \hat F_\delta(t) - F_\delta(t) \vert, \qquad \delta = 1, \dots, D,
	\end{equation*}
	where $\hat f_\delta$ and $f_\delta$ are the estimated and true incidence functions, $\hat F_\delta$ and $F_\delta$ are the estimated and true subdistribution functions, and $\pi_\delta$ is the true probability of occurrence for event type $\delta$, for any $\delta$. This rescaling of standard distances is motivated by the need to obtain error measurements that are comparable across different event types. Indeed, the rescaled versions of the incidence and subdistribution functions correspond to proper density and distribution functions, respectively.
	Our hierarchical model systematically outperforms the independent model in the estimation of both incidence functions and subdistributions. Furthermore, the average number of common latent locations in the hierarchical model is $17.79$, whereas the alternative model employs on average $35.46$ latent locations among the independent random measures. This significant difference (roughly double) suggests that a substantially larger number of components is required to compensate for the lack of information borrowing in the independent model. Nonetheless, a deeper probabilistic understanding of the role of information borrowing in the competing-risks setting will require further investigation.

	\subsection{Applications to clinical datasets} 
	\label{sec:applications}

	We apply the proposed Bayesian nonparametric model to a clinical dataset from the bone marrow transplant registry of the European Blood and Marrow Transplant (EBMT) Group, publicly available in the \texttt{crrSC} package in R. Data from this registry have previously been used to illustrate methods for competing risks \citep{zhou2012}, and this specific dataset was recently used in \cite{schmitt2023} to compare different approaches to clustered data. Refer to Section~\ref{supsec:melanoma} for an additional application on a dataset of patients affected by stage I melanoma.

	Real--world datasets require a slight extension of our model to accommodate both right--censored observations and categorical predictors. For right--censored data, it is enough to modify the likelihood in \eqref{beginning_likelihood} by removing censored observations from the factor that multiplies the exponential term. 
	The distributional results in Sections~\ref{sec:marginal} and~\ref{sec:posterior} and algorithms in Section~\ref{sec:gibbs} are easily adapted to the presence of right--censored observations once the multiplicative intensity likelihood in \eqref{beginning_likelihood} is replaced by \lesspace
	\begin{equation*}
		\mathcal{L}(\bm{\crm}; \boldsymbol T_{1:n}, \boldsymbol \Delta_{1:n}) 
		= \exp \left( - \sum_{\ell=1}^D \sum_{i=1}^n \int_0^{T_i} \int_\X k(s;x) \, \tilde \mu_\ell(dx)\,ds \right) \: \prod_{i \notin \mathscr{C}} \: \int_\X k(T_i;x) \, \tilde \mu_{\Delta_i}(dx),
	\end{equation*}
	where $\mathscr{C} \subseteq \set{1,\ldots,n}$ denotes the set of censored observations.
	This implies that the additional latent r.v.~$\boldsymbol X_{1:n}$ and $\boldsymbol Z_{1:n}$ are introduced exclusively for non--censored observations, which are those accommodated within the latent partition. In contrast, censored observations contribute to the augmented likelihood only through $\displaystyle K_n(x)=K_n(x; \boldsymbol T_{1:n})=\sum_{i=1}^n \int_0^{T_i}k(s;x)\, ds$.
	
	The inclusion of predictors is based on a Cox regression, or \emph{proportional hazards}, model, which defines a semiparametric prior on the hazard rates. Specifically, cause--specific hazard rates \eqref{mixture_hazard} are extended to \begin{equation*}
		\tilde h_\delta(t; \zeta) = \exp\left(\eta^T \zeta \right) \int_\X k(t;x) \, \crm_\delta(dx), \qquad \delta = 1, \dots, D,
	\end{equation*}
	where $\eta$ are the regression coefficients and $\zeta$ denotes the categorical predictors. The extension of distributional results and sampling algorithms requires minimal efforts, namely the re--definition of the quantities in \eqref{kernel_quantities} as $\displaystyle Q(\boldsymbol T_{1:n}, \boldsymbol X_{1:n}, \boldsymbol \zeta_{1:n}) = \prod_{j=1}^k \, \prod_{i \in C_j^X} \exp\left(\eta^T \zeta_i\right) \, k(T_i;X^*_j)$ and
	\begin{equation*}
		K_n(x) = K_n(x; \boldsymbol T_{1:n}, \boldsymbol \zeta_{1:n}) = \sum_{i=1}^n \exp\left(\eta^T \zeta_i\right) \int_0^{T_i} k(s;x)\,ds,
	\end{equation*}
	where $\boldsymbol \zeta_{1:n} = (\zeta_1, \dots, \zeta_n)$ is the collection of observed predictors. Our application involves a single binary predictor $\zeta_i \in \set{0,1}$, leading to a single regression coefficient $\eta \in \R$, for which we specify a non--informative centered Gaussian prior with large variance. Moreover, since the increasing hazard rates assumption from the Dykstra-Laud kernel appears too restrictive, we resort to the Ornstein-Uhlenbeck kernel (see Section~\ref{sec:mixture_hazard}) with a non-informative hyperprior on the rate parameter.
	
	\begin{figure}[t]
		\centering
		\noindent\makebox[\textwidth]{
			\includegraphics[width=0.5\textwidth]{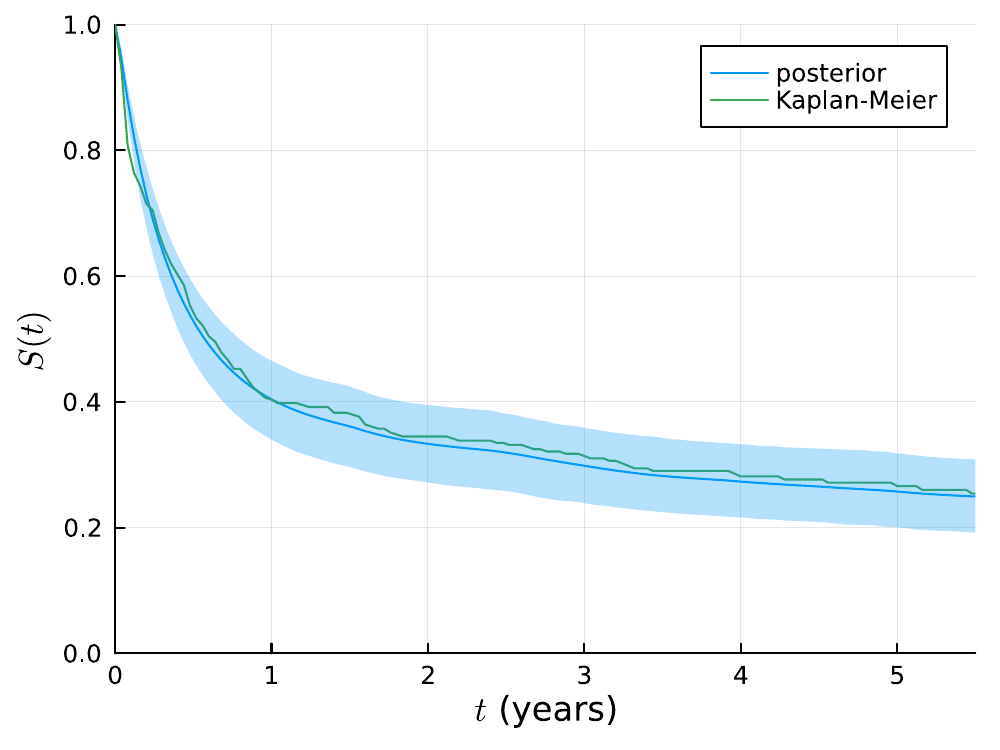}
			\includegraphics[width=0.5\textwidth]{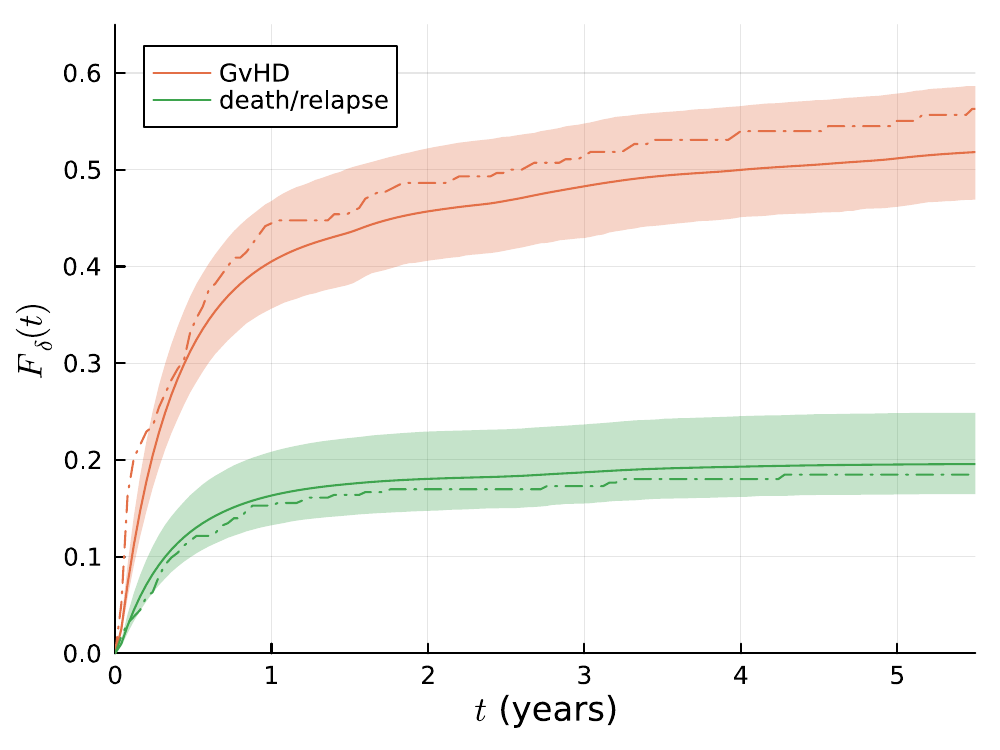}
		}
		\includegraphics[width=0.5\textwidth]{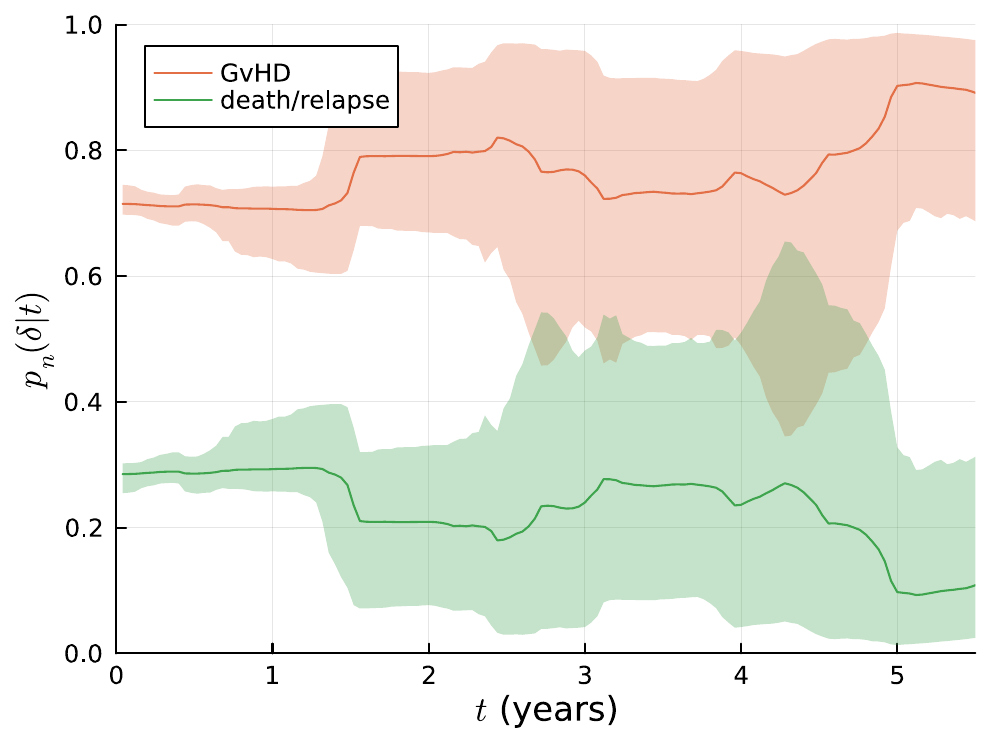}
		\captionsetup{width=0.85\textwidth,font=small}
		\caption{EBMT dataset: posterior estimates of the survival function (left), the subdistributions for the primary (GvHD) and competing (death/relapse) event types (middle), and prediction curves (right); the first two estimates are compared with the frequentist counterparts. Curves are related to patients who experienced graft from bone marrow cells.}
		\label{fig:transplantation}
	\end{figure}

	The EBMT dataset includes data for $400$ patients diagnosed with acute myeloid leukemia, who underwent a bone marrow transplantation in $153$ different hospitals. The primary event of interest is the occurrence of either acute or chronic \emph{Graft-versus-Host-Disease} (GvHD), with death or relapse without GvHD as competing events; survival times are expressed in years from the graft. Of the patients, $194$ ($48.5$\%) observed GvHD, $74$ ($18.5$\%) experienced a competing event and the rest ($33.0$\%) were censored. The median follow-up time is $0.50$ years, and the maximum observed follow-up time for the primary event of interest is $7.30$ years.  The main categorical predictor is the stem cells source for transplantation: bone marrow grafts were performed on $222$ patients ($55.5$\%), while $178$ ($44.5$\%) patients received peripheral blood stem cells. The posterior estimates of the survival function, subdistributions for the primary and competing sources of risk, and prediction curves are displayed in Figure~\ref{fig:transplantation}, along with pointwise $0.95$ credible bands, for patients who experienced graft from bone marrow cells; the corresponding Kaplan-Meier and Aalen-Johansen estimators are also shown for comparison. The posterior estimate of the regression parameter $\eta$ is $-0.079$, with $0.95$ credible interval $[-0.325, 0.169]$, indicating no significant difference between the stem cell sources. These findings align with \cite{zhou2012}, who reported an estimate of $-0.051$ ($p$-value: $0.35$) using frequentist techniques based on the Fine-Gray proportional hazards model \citep{finegray1999} on a larger dataset. Finally, the posterior distribution of the kernel rate parameter $\kappa$ provides strong evidence against the increasing hazard rate assumption. 

	\subsection*{Acknowledgements}
	
	The authors are grateful to the Associate Editor and two anonymous referees for their constructive and insightful suggestions, which led to major improvements to the paper. The research was largely conducted while Claudio Del Sole was a PhD student at Bocconi University (Milan, Italy).

	\subsection*{Funding}
	
	The authors were partially supported by the European Union - Next Generation EU PRIN-PNRR (project 2022CLTYP4).
	
	\medskip
	
	\singlespacing
	\bibliographystyle{chicago}
	\bibliography{arXiv_biblio}
	
	\clearpage
	\onehalfspacing
	
	\appendix
	
	\renewcommand{\thesection}{S\arabic{section}}
	\setcounter{section}{0}
	\renewcommand{\thesubsection}{S\arabic{section}.\arabic{subsection}}
	\setcounter{subsection}{0}
	\renewcommand{\theequation}{S\arabic{equation}}
	\setcounter{equation}{0}
	\renewcommand{\thefigure}{S\arabic{figure}}
	\setcounter{figure}{0}
	\renewcommand{\thetable}{S\arabic{table}}
	\setcounter{table}{0}
	
	\begin{center}
		\emph{\LARGE Supplementary Material to:} \\[\baselineskip]
			{\LARGE Principled Estimation and Prediction with} \\[.5\baselineskip]
			{\LARGE Competing Risks: a Bayesian Nonparametric Approach}
	\end{center}

	\bigskip
	
	\begin{abstract}
		\noindent This supplement contains a discussion of an alternative modeling approach based on latent failure times in Section~\ref{supsec:latent_times}, background material and a key technical result in Section~\ref{supsec:crm}, and the proofs of the main statements in Section~\ref{supsec:proofs}. Section~\ref{supsec:sampling} is devoted to additional technical details on the model and algorithms. Complements to the applications presented in the main manuscript, as well as further illustrations with synthetic and real data, including a comparison with an alternative tree--based method, are provided in Sections~\ref{supsec:applications}--\ref{supsec:sparapani}.
		Numbered elements (i.e., theorems, equations, figures, and so on) specific to this Supplementary Material, and not included in the main manuscript, are labeled with an `S' prefix. References to numbered items without the `S' prefix refer to the main manuscript.
	\end{abstract}
	
	\clearpage
	\section{Latent failure times approach}
	\label{supsec:latent_times}
	
	Section~\ref{sec:competing_risks} of the manuscript mentions an alternative approach to competing risks, which is briefly outlined here. This approach relies on the introduction of latent time--to--event variables $Y_{i,1}, \dots, Y_{i,D}$, for each observation $i = 1, \dots, n$. The observed survival times are then defined as $T_i=\min\{Y_{i,1},\ldots,Y_{i,D}\}$. In a Bayesian nonparametric framework, this setting can be accommodated by assuming partial exchangeability of the array $\set{(Y_{i,\delta})_{i\ge 1} \colon \delta=1,\ldots,D}$. This is a very intuitive distributional invariance property that extends the notion of exchangeability. Specifically, let  $\pi_1,\ldots,\pi_D$ denote finite permutations on the integers, i.e.~for any $\delta$ there exists $I_\delta\in\mathds{N}$ such that $\pi_\delta(i)=i$ for any $i\ge I_\delta$. Hence, partial exchangeability corresponds to the following distributional equivalence:
	\begin{equation*}
		\set{(Y_{i,\delta})_{i\ge 1} \,\colon\, \delta=1,\ldots,D}
		\,\stackrel{\mbox{\scriptsize d}}{=}\,
		\set{(Y_{\pi_\delta(i),\delta})_{i\ge 1} \,\colon\, \delta=1,\ldots,D}.
	\end{equation*}
	This implies that the latent failure times are: (i) exchangeable when they share the same cause $\delta$; (ii) conditionally independent across different causes. In view of a representation theorem for partially exchangeable arrays, there exists a vector of random probability measures $\bm{\tilde p} = (\tilde p_1,\ldots,\tilde p_D)$ on $\R^+$ such that
	\begin{equation*}
		\label{eq:part_exchange}
		(Y_{i_1,\delta_1},\ldots,Y_{i_k,\delta_k}) \mid \bm{\tilde p} \: \sim \:
		\tilde p_{\delta_1}\times\,\cdots\,\times\,\tilde p_{\delta_k},
	\end{equation*} 
	for any $k\ge 1$, $\delta_1,\ldots,\delta_k\in\set{1,\ldots,D}$ and $i_1,\ldots,i_k\in\mathds{N}$. Remarkably, the dependence among components $\tilde p_1,\ldots,\tilde p_D$ implies unconditional dependence of the $Y_{i,\delta}$'s, both within each cause--specific group and across different causes. This differs from standard approaches that assume independent latent time-to-event for the different sources of risk.
	
	\clearpage
	\section{Completely random measures}
	\label{supsec:crm}
	
	The model proposed in \eqref{model} for analyzing competing risks data relies on completely random measures, originally introduced in \cite{kingman1967}, as main building blocks. Therefore, before presenting the proofs of the main results, we outline key properties of this class of random measures.
	
	\begin{definition}
		{\rm A random element $\tilde \mu$, defined on $\X$ and taking values in the space of  boundedly finite measures, is a \emph{completely random measure} (CRM) if, for any collection of pairwise disjoint sets $A_1, \dots, A_n$, the random variables $\crm(A_1), \dots, \crm(A_n)$ are mutually independent. }
	\end{definition}
	
	\noindent We refer the reader to \cite{daley2007} and \cite{kingman1993} for exhaustive accounts. Throughout, we consider CRMs without deterministic component and fixed points of discontinuity. These are almost surely (a.s.) discrete and can be represented as 
	\begin{equation} \label{jumps_atoms}
		\crm = \sum_{h \ge 1} S_h \, \delta_{X^*_h},
	\end{equation}
	where $(S_h)_{h \ge 1}$ is a sequence of positive random jumps and $(X^*_h)_{h \ge 1}$ a sequence of random locations. Remarkably, the pairs $(S_h,X_h^*)_{h\ge 1}$ are the points of a Poisson random measure $\tilde{N}$ on $\R^+ \times \X$ having mean intensity $\nu$, which implies that, for any $A$ such that $\nu(A)<\infty$, the integer--valued r.v.~$\tilde{N}(A)=\mbox{card}(A\cap\set{(S_h,X_h^*) \colon h\ge 1})$ has Poisson distribution with mean $\nu(A)$. 
	In fact, the discreteness of their realizations stems from their characterization as linear functionals of Poisson random measures
	\begin{equation*}
		\crm(dx) = \int_{\R^+} s \, \tilde N(ds,dx),
	\end{equation*}
	where $\tilde N$ is a Poisson random measure on the product space $\R^+ \times \X$. 
	The measure $\nu$ is also called \textit{Lévy intensity} of $\crm$ and it uniquely identifies $\crm$, which motivates the notation $\crm \sim \text{CRM}(\nu)$. This can be also seen from the Laplace functional transform of $\crm$, which coincides, for any non-negative bounded $f \colon \X \mapsto \R^+$, with
	\begin{equation} \label{laplace_transform}
		\E \left[ \exp \left( - \int_\X f(x) \, \crm(dx) \right) \right] = \exp \left( - \int_{\R^+ \times \X} (1 - e^{-s\,f(x)}) \, \nu(ds,dx) \right).
	\end{equation}
	For any bounded set $A \subset \X$, the measure $\nu$ must satisfy the integrability condition
	\begin{equation*}
		\label{eq:integrab_nu}
		\displaystyle \int_{\R^+ \times A} \min(s,1) \, \nu(ds,dx) < \infty.
	\end{equation*}
	Moreover, we assume $\nu(\R^+ \times A) = \infty$, which corresponds to $\tilde \mu$ having an infinite number of random jumps on bounded sets and thus guarantees that $\tilde \mu$ is strictly positive a.s.~\citep{regazzini2003}. Such an assumption is common in Bayesian nonparametric modeling and is known as \emph{infinite activity} property. The L\'evy intensity can be conveniently decomposed as
	\begin{equation*}
		\nu(ds,dx) = \rho(ds \, \vert \, x) \, \alpha(dx),
	\end{equation*}
	where $\rho$ is a transition kernel characterizing the random jumps $S_h$ and $\alpha$ is a $\sigma$--finite measure characterizing the random locations $X_h^*$. If the jumps are independent of the locations, we have a \emph{homogeneous} CRM with $\nu(ds,dx) = \rho(ds) \, \alpha(dx)$. For instance, the popular \emph{gamma} CRM arises when $\rho(ds) = s^{-1} \, e^{-\beta \, s} \, ds$, for $ \beta > 0$; from \eqref{laplace_transform}, it follows that $\crm(A)$ has gamma distribution  with shape parameter $\alpha(A)$ and scale parameter $\beta$. Some recent extensions of CRMs and of their uses in Bayesian nonparametrics can be found in \cite{passeggeri2023,passeggeri2024}.
	
	\subsection{Preliminary technical lemma}
	
	We provide a technical result about completely random measures (CRMs), which constitutes the backbone for the following proofs. The main tool used for its derivation is the Laplace functional of CRMs, in the spirit of the general approach of \cite{james2009}.
	Consider a homogeneous CRM $\crm$ with Lévy intensity $\nu(ds, dx) = \rho(ds) \, \alpha(dx)$, where $\alpha$ is a $\sigma$--finite diffuse measure on $\X$. The Laplace exponent and cumulants of $\crm$ are, respectively, defined as
	\begin{equation*}
		\psi(u) = \int_{\R^+}\left(1-e^{-us}\right) \rho(ds), \qquad 
		\tau(m; u) = \int_{\R^+} s^m e^{-us} \rho(ds).
	\end{equation*}
	
	\begin{lemma}
		For any non-negative measurable function $f \colon \X \mapsto \R^+$, and distinct values $x_1^*, \ldots, x_k^* \in \mathbb{X}$, with multiplicities $n_1, \ldots, n_k \ge 1$, one has
		\begin{multline} \label{key_identity_formal}
			\qquad \lim _{\varepsilon \rightarrow 0} \ \frac{ \displaystyle \E \left[ \exp \left(-\int_{\mathbb{X}} f(x) \tilde{\mu}(d x)\right) \prod_{j=1}^k \tilde{\mu}(B_{\varepsilon}(x_j^*))^{n_j} \right]}{\displaystyle \prod_{j=1}^k \alpha (B_{\varepsilon}(x_j^*))} \\
			= \exp \left( - \int_\X \psi(f(x)) \, \alpha(dx) \right) \, \prod_{j=1}^k \tau(n_j; f(x^*_j)), \qquad
		\end{multline}
		where $B_\varepsilon(\bar x) = \set{x \in \X \colon d(x, \bar x) < \varepsilon}$ denotes the $\varepsilon$-ball centered at $\bar x \in \X$.
	\end{lemma}
	
	\noindent Note that the relationship \eqref{key_identity_formal} can be rewritten less formally but more conveniently as
	\begin{multline} \label{key_identity}
		\qquad \E \left[ \exp \left( - \int_\X f(x) \, \tilde \mu(dx) \right) \prod_{j=1}^k \tilde \mu(dx^*_j)^{n_j}  \right] \\[8pt]
		= \exp \left( - \int_\X \psi(f(x)) \, \alpha(dx) \right) \prod_{j=1}^k \tau(n_j; f(x^*_j)) \, \alpha(dx^*_j). \qquad
	\end{multline}
	
	\noindent \textit{Proof.}
	Consider $\varepsilon > 0$ sufficiently small such that the $\varepsilon$-balls $B_{\varepsilon}(x_1^*), \dots, B_{\varepsilon}(x_k^*)$ are pairwise disjoint, and define the complement $\X^* = \X \setminus \left( \cup_{j=1}^k B_{\varepsilon}(x_j^*) \right)$. By the independence of the evaluations of a CRM on pairwise disjoint sets, one has
	\begin{multline*}
		\E \left[ \exp \left(-\int_{\mathbb{X}} f(x) \tilde{\mu}(d x) \right) \prod_{j=1}^k \tilde{\mu}(B_{\varepsilon}(x_j^*))^{n_j} \right] \\[8pt]
		= \E \left[ \exp \left(-\int_{\X^*} f(x) \tilde{\mu}(d x)\right) \right] \, \prod_{j=1}^k \E \left[ \exp \left(- \int_{B_{\varepsilon}(x_j^*)} f(x) \tilde{\mu}(dx) \right) \tilde{\mu}(B_{\varepsilon}(x_j^*))^{n_j} \right]. 
	\end{multline*}
	The first factor is the Laplace transform of $\tilde\mu$ restricted to $\X^*$ and evaluated at $f$. Therefore
	\begin{align*}
		\E \left[ \exp \left( - \int_{\X^*} f(x) \tilde{\mu}(dx) \right) \right] & = \exp \left( - \int_{\X^*} \left( 1-e^{-f(x)\,s} \right) \rho(ds)\,\alpha(dx) \right) \\[8pt]
		& = \exp \left( - \int_{\X^*} \psi(f(x))\,\alpha(dx) \right).
	\end{align*}
	By the monotone convergence theorem, in the limit as $\varepsilon \to 0$, we obtain
	\begin{align*}
		\lim_{\varepsilon \to 0} \, \E \left[ \exp \left( - \int_{\X^*} f(x) \tilde{\mu}(dx) \right) \right] & = \lim_{\varepsilon \to 0} \, \exp \left( - \int_{\X^*} \psi(f(x))\,\alpha(dx) \right) \\[8pt]
		& = \exp \left( - \int_{\X} \psi(f(x))\,\alpha(dx) \right).
	\end{align*}
	On the other hand, for $j = 1, \dots, k$, by the monotone convergence theorem, we also get
	\begin{align*}
		M(\varepsilon,x_j^*;n_j)
		& =\E \left[ \exp \left(- \int_{B_{\varepsilon}(x_j^*)} f(x) \tilde{\mu}(dx) \right) \tilde{\mu}(B_{\varepsilon}(x_j^*))^{n_j} \right] \\[8pt]
		& = \lim_{u \to 0} \, \E \left[ \exp \left( - \int_{B_{\varepsilon}(x_j^*)} f(x) \tilde{\mu}(dx) - u \, \tilde{\mu}(B_{\varepsilon}(x_j^*)) \right) \tilde{\mu}(B_{\varepsilon}(x_j^*))^{n_j} \right].
	\end{align*}
	Given that, for $n_j \ge 0$,
	\begin{equation*}
		\frac{\mathrm{d}^{n_j}}{\mathrm{d}u^{n_j}} \, \exp \left( - u \, \tilde{\mu}(B_{\varepsilon}(x_j^*)) \right) = (-1)^{n_j} \, \tilde{\mu}(B_{\varepsilon}(x_j^*))^{n_j} \, \exp \left( - u \, \tilde{\mu}(B_{\varepsilon}(x_j^*)) \right),
	\end{equation*}
	the expression above becomes 
	\begin{align*}
		M(\varepsilon,x_j^*;n_j) & = (-1)^{n_j} \, \lim_{u \to 0} \, \E \left[ \frac{\mathrm{d}^{n_j}}{\mathrm{d}u^{n_j}} \, \exp \left( - \int_{B_{\varepsilon}(x_j^*)} f(x) \tilde{\mu}(dx) - u \, \tilde{\mu}(B_{\varepsilon}(x_j^*)) \right) \right] \\[8pt]
		& = (-1)^{n_j} \, \lim_{u \to 0} \,  \frac{\mathrm{d}^{n_j}}{\mathrm{d}u^{n_j}} \, \E \left[ \exp \left( - \int_{B_{\varepsilon}(x_j^*)} (f(x)+u) \tilde{\mu}(dx) \right) \right].
	\end{align*}
	Again, the expectation coincides with the Laplace transform of $\tilde \mu$ restricted to $B_{\varepsilon}(x_j^*)$ and evaluated at $f+u\,1$, where $x\mapsto 1(x)=1$ is the identity function. Hence we have
	\begin{equation*}
		M(\varepsilon,x_j^*;n_j) = (-1)^{n_j} \, \lim_{u \to 0} \,  \frac{\mathrm{d}^{n_j}}{\mathrm{d}u^{n_j}} \, \exp \left( - \int_{B_{\varepsilon}(x_j^*)} \psi(f(x)+u)\, \alpha(dx) \right).
	\end{equation*}
	Since $\alpha$ is assumed to be finite and diffuse, the measure of the $\varepsilon$-ball $B_{\varepsilon}(x_j^*)$ vanishes for $\varepsilon \to 0$, and the expression above can be rewritten as
	\begin{multline*}
		M(\varepsilon,x_j^*;n_j) = (-1)^{n_j+1} \, \lim_{u \to 0} \, \int_{B_{\varepsilon}(x_j^*)} \frac{\mathrm{d}^{n_j}}{\mathrm{d}u^{n_j}} \, \psi(f(x)+u)\, \alpha(dx) \\
		\times  \, \exp \left( - \int_{B_{\varepsilon}(x_j^*)} \psi(f(x)+u)\, \alpha(dx) \right) + \, o \left(\alpha(B_{\varepsilon}(x_j^*))\right).
	\end{multline*}	
	By definition of Laplace exponent of $\tilde \mu$, its derivatives coincide with the cumulants of $\tilde \mu$,
	\begin{align*}
		\frac{\mathrm{d}^m}{\mathrm{d}u^m} \, \psi(u) & = \frac{\mathrm{d}^m}{\mathrm{d}u^m} \int_{\R^+}\left(1-e^{-us}\right) \rho(ds) = \int_{\R^+} \frac{\mathrm{d}^m}{\mathrm{d}u^m} \left(1-e^{-us}\right) \rho(ds) \\[8pt]
		& = (-1)^{m+1} \int_{\R^+} s^m e^{-us} \, \rho(ds) = (-1)^{m+1} \, \tau(m; u).
	\end{align*}
	Therefore, the quantity above becomes
	\begin{align*}
		M(\varepsilon,x_j^*;n_j)	
		& = \lim_{u \to 0} \, \int_{B_{\varepsilon}(x_j^*)} \tau(n_j; f(x)+u)\, \alpha(dx) \\
		& \qquad \qquad \times\, \exp \left( - \int_{B_{\varepsilon}(x_j^*)} \psi(f(x)+u)\, \alpha(dx) \right) + \, o \left(\alpha(B_{\varepsilon}(x_j^*))\right) \\[8pt]
		& = \int_{B_{\varepsilon}(x_j^*)} \tau(n_j; f(x))\, \alpha(dx)  \, \exp \left( - \int_{B_{\varepsilon}(x_j^*)} \psi(f(x))\, \alpha(dx) \right) + \, o \left(\alpha(B_{\varepsilon}(x_j^*))\right).
	\end{align*}
	In conclusion, taking the limit for $\varepsilon \to 0$, by the monotone convergence theorem, we get
	\begin{multline*}
		\lim_{\varepsilon \to 0} \, \frac{1}{\alpha(B_{\varepsilon}(x_j^*))} \, \E \left[ \exp \left(- \int_{B_{\varepsilon}(x_j^*)} f(x) \tilde{\mu}(dx) \right) \tilde{\mu}(B_{\varepsilon}(x_j^*))^{n_j} \right] \\[8pt]
		\begin{aligned}
			& = \lim_{\varepsilon \to 0} \, \, \frac{1}{\alpha(B_{\varepsilon}(x_j^*))} \, \int_{B_{\varepsilon}(x_j^*)} \tau(n_j; f(x))\, \alpha(dx) \, \exp \left( - \int_{B_{\varepsilon}(x_j^*)} \psi(f(x))\, \alpha(dx) \right) \\[8pt]
			& = \lim_{\varepsilon \to 0} \, \, \frac{1}{\alpha(B_{\varepsilon}(x_j^*))} \, \int_{B_{\varepsilon}(x_j^*)} \tau(n_j; f(x))\, \alpha(dx) = \tau(n_j; f(x_j^*)).
		\end{aligned}
	\end{multline*}
	\hfill \qed
	
	\clearpage
	\section{Proofs}
	\label{supsec:proofs}
	
	\paragraph{Proof of Proposition~\ref{prop:survival}}
	
	As a first step, we establish a necessary and sufficient condition that guarantees a proper survival function, a result of independent interest. Let $\psi$ and $\psi_0$ denote the Laplace exponents of $\crm_1, \dots, \crm_D$ and $\crm_0$, respectively:
	\begin{equation*}
		\psi(u) = \int_{\R^+} (1-e^{-us}) \, \rho(ds), \qquad \psi_0(u) = \int_{\R^+} (1-e^{-us}) \, \rho_0(ds).
	\end{equation*}
	
	\begin{proposition}
		The survival function $\tilde S(t)$ in \eqref{survival} is a proper survival function, that is $\lim _{t \to \infty} \tilde S(t) = 0$ a.s., if and only if, for any $\lambda > 0$,
		\begin{equation} \label{kernel_condition}
			\lim _{t \to \infty} \int_\X \psi_0\left( D \psi\left( \lambda \int_0^t k(s; x) \, ds \right) \right) \Lambda_0(dx) = \infty.
		\end{equation}
	\end{proposition}
	
	\noindent \textit{Proof.}
	The survival function is defined as
	\begin{equation*}
		\tilde S(t) = \exp \left( - \sum_{\ell=1}^D \int_0^t \int_\X k(s;x) \, \tilde \mu_\ell(dx)\,ds \right).
	\end{equation*}
	Therefore, the condition $\lim_{t \to \infty} \tilde S(t) = 0$ a.s.~is equivalent to
	\begin{equation*}
		\lim_{t \to \infty} \, \sum_{\ell=1}^D \int_0^t \int_\X k(s;x) \, \tilde \mu_\ell(dx)\,ds = \lim_{t \to \infty} \tilde H(t) = \infty, \qquad \text{a.s.},
	\end{equation*}
	where $\tilde H(t)$ denotes the cumulative hazard for all sources of risk. The Laplace transform of $\tilde H(t)$ can be easily computed using \eqref{laplace_transform} recursively. Indeed, conditionally on the random measure $\crm_0$ at the root of the hierarchical prior, the random measures $\crm_1, \dots, \crm_D$ are independent CRMs on $\X$, characterized by the homogeneous Lévy intensity $\tilde \nu(ds,dx) = \rho(ds) \, \crm_0(dx)$. For any $\lambda \ge 0$,
	\begin{align*}
		\E \left[ \exp \left( - \lambda \tilde H(t) \right) \mid \crm_0 \right] 
		& = \E \left[ \exp \left( - \lambda \sum_{\ell=1}^D \int_0^t \int_\X k(s;x) \, \tilde \mu_\ell(dx)\,ds \right) \bigm\vert \crm_0 \right] \\[8pt]
		& = \prod_{\ell=1}^D \E \left[ \exp \left( - \lambda \int_\X \int_0^t  k(s;x) \,ds \, \tilde \mu_\ell(dx) \right) \bigm\vert \crm_0 \right] \\[8pt]
		& = \exp \left( - D \int_\X \psi\left( \lambda \int_0^t k(s;x) \,ds \right) \crm_0(dx) \right),
	\end{align*}
	where $\psi$ is the Laplace exponent of $\crm_1, \dots, \crm_D$, which are identically distributed.
	The random measure $\crm_0$ at the root of the hierarchical prior is itself a CRM with homogeneous Lévy intensity $\nu_0(ds,dx) = \rho_0(ds) \, \Lambda_0(dx)$. Therefore, the Laplace transform of $\tilde H(t)$ is
	\begin{align*}
		\E \left[ \exp \left( - \lambda \tilde H(t) \right) \right] 
		& = \E \left[ \exp \left( - D\, \int_\X \psi\left( \lambda \int_0^t k(s;x) \,ds \right) \tilde \mu_0(dx) \right) \right] \\[8pt]
		& = \exp \left( - \int_\X \psi_0 \left( D \psi\left( \lambda \int_0^t k(s;x) \,ds \right) \right) \Lambda_0(dx) \right),
	\end{align*}
	where $\psi_0$ is the Laplace exponent of $\crm_0$.
	The Laplace transform of an a.s.~infinite random variable is zero for every $\lambda > 0$. Moreover, the distribution of the limit of a sequence of non-negative random variables is characterized by the pointwise limit of their Laplace transforms, by the Lévy continuity theorem. Consequently, $\lim_{t \to \infty} \tilde H(t) = \infty$ if and only if, for any $\lambda > 0$,
	\begin{equation*}
		\lim_{t \to \infty} \E \left[ \exp \left( - \lambda \tilde H(t) \right) \right] = \exp \left( - \lim_{t \to \infty} \int_\X \psi_0 \left( D \psi\left( \lambda \int_0^t k(s;x) \,ds \right) \right) \Lambda_0(dx) \right) = 0.
	\end{equation*} 
	\hfill \qed
	
	\noindent \textit{Proof of Proposition 3.}
	Remarkably, the Laplace exponents $\psi$ and $\psi_0$ are strictly positive and increasing functions on $(0,\infty)$. Moreover,
	\begin{equation*}
		\lim_{u \to \infty} \psi(u) = \lim_{u \to \infty} \int_{\R^+} (1-e^{-us}) \, \rho(ds) = \int_{\R^+} \rho(ds),
	\end{equation*}
	and similarly $\displaystyle \lim_{u \to \infty} \psi_0(u) = \int_{\R^+} \rho_0(ds)$. Therefore, $\lim_{u \to \infty} \psi(u) = \lim_{u \to \infty} \psi_0(u) = \infty$ if and only if $\crm_1,\ldots,\crm_D$ and $\crm_0$ are infinitely active CRMs; see Section~\ref{supsec:crm} for the definition and implications of infinite activity.
	To conclude the proof, we show that either one of the conditions contained in the statement imply \eqref{kernel_condition}:
	\begin{itemize}
		\item[(i)] Assume $\displaystyle \lim_{t\to\infty} \int_0^t k(s;x)\,ds = \infty$ for every $x \in \X$. Then, for any $\lambda > 0$,
		\begin{multline*}
			\lim _{t \to \infty} \int_\X \psi_0\left( D \psi\left( \lambda \int_0^t k(s; x) \, ds \right) \right) \Lambda_0(dx) \\[8pt]
			= \int_\X \lim _{u \to \infty} \psi_0\left( D \psi\left( \lambda u \right) \right) \Lambda_0(dx) = \lim _{u \to \infty} \psi_0\left( D \psi\left( u \right) \right) \cdot \Lambda_0(\X).
		\end{multline*}
		Since $\Lambda_0(\X) > 0$, the quantity above is infinite if and only if $ \psi_0\left( D \psi\left( u \right) \right) \to \infty$, i.e.~if and only if $\crm_1,\ldots,\crm_D$ and $\crm_0$ are infinitely active, 
		\item[(ii)] Assume $\displaystyle \lim_{t\to\infty} \int_0^t k(s;x)\,ds \ge C$ for every $x \in \X$ and some $C > 0$. Then, for any $\lambda > 0$,
		\begin{multline*}
			\lim _{t \to \infty} \int_\X \psi_0\left( D \psi\left( \lambda \int_0^t k(s; x) \, ds \right) \right) \Lambda_0(dx) \\[8pt]
			\ge \int_\X \psi_0\left( D \psi\left( \lambda C \right) \right) \Lambda_0(dx) = \psi_0\left( D \psi\left( \lambda C \right) \right) \cdot \Lambda_0(\X),
		\end{multline*}
		which is infinite if and only if $\Lambda_0(\X) = \infty$, since $\psi_0\left( D \psi\left( \lambda \, C \right) \right) > 0$.
	\end{itemize}
	\hfill \qed
	
	\paragraph{Proof of Theorem~\ref{prop:support}}
	The proof leverages results from Section 3 in \cite{deblasi2009}, where the authors focus on single--cause mixture hazards
	\begin{equation} \label{kernel_mixture}
		\tilde h(t) = \int_\X k(t;x) \, \crm(dx),
	\end{equation}
	with mixing measure $\crm$ a CRM as defined in Section~\ref{supsec:crm} and characterized by the Lévy intensity $\nu(ds,dx) = \rho(ds \, \vert \, x) \, \alpha(dx)$. They also assume that $\rho$ and $\alpha$ are diffuse measures and that the support of $\alpha$ is the entire space $\X$. Finally, their Proposition 3 contains a crucial condition on the small--time behaviour of $\crm$, namely that there exists $r > 0$ such that
	\begin{equation} \label{condition_crm}
		\liminf_{t \to 0} \:\frac{\crm((0,t])}{t^r} = \infty \qquad \text{a.s.}
	\end{equation}
	The random mixture hazard \eqref{kernel_mixture} induces a prior $\mathcal{P}$ on the space of densities for the survival time, i.e.~densities on $\mathbb{R}^+$.
	\cite{deblasi2009} provide an explicit characterization of the classes of densities at which $\mathcal{P}$ is weakly consistent, for several kernel choices. Their proof technique consists in translating the well--known sufficient condition for weak consistency, expressed in terms of Kullback--Leibler neighborhoods in the space of densities, to a condition on uniform neighborhoods in the space of hazards.
	Hereafter, $\tilde H(t) = \displaystyle \int_0^t \tilde h(s)\,ds$ denotes the cumulative hazard, so that $\tilde f(t)=\tilde h(t) \exp(-\tilde H(t))$ is the associated density function. Moreover, $f_0$ is the ``true'' data--generating density on $\mathbb{R}^+$, and $h_0$ its corresponding hazard function. For convenience, we restate the specific consistency results that are relevant to our purposes.
	
	\begin{theorem}[\citet{deblasi2009}, Theorem 4]
		Let $\tilde h$ be a mixture hazard \eqref{kernel_mixture} with Dykstra-Laud kernel $k(t;x) = \mathds{1}_{\set{t \ge x}}$ and assume $\crm$ satisfies \eqref{condition_crm}.
		Then, $\mathcal{P}$ is weakly consistent at each $f_0$ such that $\displaystyle \int_0^\infty \E[\tilde H(t)] \, f_0(t) \,dt < \infty$, with $h_0(0) = 0$ and $h_0(t)$ strictly positive and non--decreasing for any $t > 0$.
	\end{theorem}
	
	\noindent An analogous result holds for the slight variation of the Dykstra-Laud kernel $k(t;x) = \gamma \,\mathds{1}_{\set{t \ge x}}$ considered in our paper, since the rescaled measure $\gamma\, \crm(dx)$ is again a CRM satisfying \eqref{condition_crm}.
	
	\begin{theorem}[\cite{deblasi2009}, Theorem 5] \label{th:rect}
		Let $\tilde h$ be a mixture hazard \eqref{kernel_mixture} with rectangular kernel $k(t;x) = \mathds{1}_{\set{\vert t-x\vert \le \tau}}$. Assume that the bandwidth $\tau$ is random and its prior assigns positive probability to $[0, L]$ for some $L > 0$. Then, $\mathcal{P}$ is weakly consistent at each $f_0$ such that $\displaystyle \int_0^\infty \max \left(\E[\tilde H(t)], t \right) f_0(t) \,dt < \infty$, with $h_0$ strictly positive, bounded and Lipschitz.
	\end{theorem}
	
	\noindent As with the previous case, an adaptation of the proof of Theorem~\ref{th:rect} can be used to establish an analogous result for the rectangular kernel $k(t;x) = \gamma \,\mathds{1}_{\set{ 0 \le t-x \le \tau}}$, which is the one we consider here. Indeed, the claim rests on the following arguments: (i) the condition on the random bandwidth remains valid if $\tau$ is replaced by $2\tau$; (ii) the rescaled and shifted measure $\gamma \,\crm(dx-\tau)$ is a CRM on $(\tau,\infty)$; (iii) the behavior on $(0,\tau)$ coincides with that of the Dykstra-Laud kernel. Therefore, Theorem~\ref{th:rect} extends to our setting provided $h_0(0) = 0$ and $\crm$ satisfies \eqref{condition_crm}.
	
	\begin{theorem}[\cite{deblasi2009}, Theorem 6]
		Let $\tilde h$ be a mixture hazard \eqref{kernel_mixture} with Ornstein-Uhlenbeck kernel $k(t; x) = \sqrt{2 \kappa} \, \exp(- \kappa(t-x)) \,\mathds{1}_{\set{t \ge x}}$ and assume $\crm$ satisfies \eqref{condition_crm}.
		Then, $\mathcal{P}$ is weakly consistent at each $f_0$ such that $\displaystyle \int_0^\infty \max \left(\E[\tilde H(t)], t \right) f_0(t) \,dt < \infty$, with $h_0(0) = 0$ and $h_0(t)$ satisfying the following properties for any $t > 0$:
		\begin{itemize}
			\item[\rm (a)] $h_0(t)$ is strictly positive and differentiable;
			\item[\rm (b)] for any $t > 0$ for which $h_0'(t) < 0$, the local exponential decay rate $-h_0'(t)/h_0(t)$ is smaller than $\kappa \sqrt{2\kappa}$.
		\end{itemize}
	\end{theorem}
	
	\noindent Finally, as stated in Remark 1 of \cite{deblasi2009}, if $\alpha$ is (proportional to) the Lebesgue measure, the integrability conditions on the cumulative hazard $\tilde H(t)$ in the theorems above become, respectively, $\displaystyle \int_0^\infty t^2 f_0(t)\,dt < \infty$ for the Dykstra-Laud kernel and $\displaystyle \int_0^\infty t f_0(t)\,dt < \infty$ for the rectangular and Ornstein-Uhlenbeck kernels.
	
	To take advantage of these results in our setting, we must show that the hazard rate for the survival time $T$ from all sources of risk is a mixture hazard as in \eqref{kernel_mixture}. Indeed, by construction of our model (see \eqref{survival}), the hazard rate from all sources of risk is
	\begin{equation*}
		\sum_{\ell=1}^D \int_\X k(s;x) \, \tilde \mu_\ell(dx)\,ds = \int_\X k(s;x) \, \crm_{\rm sum}(dx) \,ds,
	\end{equation*}
	where $\displaystyle \crm_{\rm sum} = \sum_{\ell=1}^D \crm_\ell$ is the mixing random measure. Conditionally on $\crm_0$ at the root of the hierarchy, the measure $\crm_{\rm sum}$ is the sum of independent CRMs, and thus is itself a CRM. Moreover, a hierarchical structure of CRMs is a CRM even unconditionally on the hierarchy's root, as stated in Theorem 1 of \cite{catalano2025}. In particular, its Laplace exponent is $\psi_{\rm sum}(u) = \psi_0 \left(D\psi(u)\right)$, where $\psi$ and $\psi_0$ are the Laplace exponents of each $\crm_\ell$ and $\crm_0$, respectively. Note that $\crm_{\rm sum}$ is infinitely active if the random measures at both levels of the hierarchy are infinitely active \cite[][Lemma 2]{catalano2025}. Hence, the hazard rate from all sources of risk is indeed of the form \eqref{kernel_mixture}.
	
	To conclude the proof, let $\boldsymbol{\crm}$ be a generalized gamma hCRM and let $\Lambda_0$ be proportional to the Lebesgue measure. This specification guarantees that the Lévy measure of $\crm_{\rm sum}$ is diffuse, and its base measure $\Lambda_0$ has full support. Moreover, let $\rho_{\rm sum}$ be the jump component of its Lévy measure, and note that
	\begin{equation*}
		\int_0^1 s^\varepsilon \rho_{\rm sum}(ds) = \int_0^1 \int_0^1 s^\varepsilon \rho_{\rm sum}(ds) \,\Lambda_0(dx) = D \int_0^1 s^\varepsilon \rho(ds) \int_0^\infty s \, \rho_0(ds) \int_0^1 \Lambda_0(dx),
	\end{equation*}
	which is infinite for every $0 < \varepsilon < \sigma$.
	Therefore, by Proposition 47.18 in \cite{sato1999}, the small--time behaviour of $\crm_{\rm sum}$ is such that $\liminf_{t \to 0} \:\crm_{\rm sum}((0,t]) / G(t) = C$ a.s.~for some constant $0 < C < \infty$, where
	\begin{equation*}
		G(t) = \frac{\log\log(t^{-1})}{\psi_{\rm sum}^{-1}(t^{-1}\log\log(t^{-1}))}.
	\end{equation*}
	For the generalized gamma hCRM, when $t \to 0$, one obtains
	\begin{equation*}
		G(t) \approx \frac{D^{1/\sigma}}{(\sigma \sigma_0^{1/\sigma_0})^{1/\sigma}} \, t^{1/\sigma\sigma_0} \left(\log\log(t^{-1})\right)^{1-1/\sigma\sigma_0},
	\end{equation*}
	which implies that condition \eqref{condition_crm} is satisfied by taking $r > 1/\sigma\sigma_0$. 
	\hfill \qed
	
	\paragraph{Detailed derivation of \eqref{likelihood} from \eqref{beginning_likelihood}}
	
	The multiplicative intensity likelihood for the competing risks model \eqref{model} is given in \eqref{beginning_likelihood},
	\begin{equation*}
		\mathcal{L}(\bm{\crm}; \boldsymbol T_{1:n}, \boldsymbol \Delta_{1:n}) = \prod_{i=1}^n \left[ \int_\X k(T_i;x_i) \, \tilde \mu_{\Delta_i}(dx_i) \: \exp \left( - \sum_{\ell=1}^D \int_0^{T_i} \int_\X k(s;y) \, \tilde \mu_\ell (dy)\,ds \right) \right].
	\end{equation*}
	In Section~\ref{sec:marginal}, we rewrite this likelihood in terms of the extended random measures $\bm{\crm^e}=(\crm^e_1, \dots, \crm^e_D)$. Moreover, we further disintegrate with respect to such random measures by introducing two sequences of latent variables, namely $\boldsymbol X_{1:n} = (X_1, \dots, X_n)$ and $\boldsymbol Z_{1:n} = (Z_1, \dots, Z_n)$. This leads to the doubly--augmented likelihood \cut
	\begin{multline*}
		\mathcal{L}(\bm{\crm^e}; \boldsymbol T_{1:n}, \boldsymbol \Delta_{1:n}, \boldsymbol X_{1:n}, \boldsymbol Z_{1:n}) \\[8pt]
		= \prod_{i=1}^n \left[ k(T_i;X_i) \, \tilde \mu^e_{\Delta_i}(dZ_i,dX_i) \:
		\exp \left( - \sum_{\ell=1}^D \int_0^{T_i} \int_{[0,1] \times \X} k(s;x) \, \tilde \mu^e_\ell (dz,dx)\,ds \right) \right].
	\end{multline*}
	The derivation of \eqref{likelihood} from the expression above reflects the nested partition structure induced by the tied values in the sequences $\boldsymbol X_{1:n}$ and $\boldsymbol Z_{1:n}$. Specifically, the coarser level of the partition is denoted by $\Pi_{\boldsymbol X_{1:n}} = \set{C_{j}^X \colon j=1,\ldots,k}$, for some $k\in\set{1,\ldots,n}$, where each $C_{j}^X=\set{i \colon X_{i}=X^*_{j}}$ contains the observations associated to location $X^*_j$.
	Therefore, the augmented likelihood becomes
	\begin{align*}
		\mathcal{L}(\bm{\crm^e}; \boldsymbol T_{1:n}, \boldsymbol \Delta_{1:n}, \boldsymbol X_{1:n}, \boldsymbol Z_{1:n}) 
		& = \prod_{j=1}^k \prod_{i \in C_j^X} \left[ k(T_i;X^*_j) \, \tilde \mu^e_{\Delta_i}(dZ_i,dX^*_j) \right] \\
		& \qquad \times \exp \left( - \sum_{i=1}^n \sum_{\ell=1}^D \int_0^{T_i} \int_{[0,1] \times \X} k(s;x) \, \tilde \mu^e_\ell (dz,dx)\,ds \right),
	\end{align*}
	where the product over $i$ has been decomposed as: (i) a product over $j$, i.e.~over the elements of the partition, and (ii) a nested product over the observations in each group $C_j^X$. Remarkably, the observations in $C_j^X$ do not necessarily share the same value of $\Delta$: indeed, $i,\ell\in C_j^X$ does not imply $\Delta_i = \Delta_\ell$. The observed partition encoded in $\boldsymbol \Delta_{1:n}$ is taken into account by considering, for each $j \in \set{1, \dots, k}$, the nested partition $\set{C_{\delta j}^X \colon \delta = 1, \dots, D}$, where $C_{\delta j}^X = \set{i \in C_j^X \colon \Delta_i = \delta}$. The product above is then rewritten as
	\begin{equation*}
		\prod_{j=1}^k \prod_{i \in C_j^X} \left[ k(T_i;X^*_j) \, \tilde \mu^e_{\Delta_i}(dZ_i,dX^*_j) \right] 
		= \left[ \prod_{j=1}^k \prod_{i \in C_j^X} k(T_i;X^*_j) \right] \left[ \prod_{j=1}^k \prod_{\delta=1}^D \prod_{i \in C_{\delta j}^X} \mu^e_\delta(dZ_i,dX^*_j) \right].
	\end{equation*}
	Finally, at the finer level of the partition, for each $j \in \set{1, \dots, k}$ and $\delta \in \set{1, \dots, D}$, the observations in $C_{\delta j}^X$ are further partitioned as $\set{C_{\delta jh}^Z \colon h=1,\ldots,r_{\delta j} }$, for some $r_{\delta j} > 0$, where each $C_{\delta jh}^Z=\set{i\in C_{\delta j}^X \colon Z_{i}=Z^*_{\delta jh}}$ contains the observations associated to latent mark $Z^*_{\delta jh}$ from the random measure $\crm^e_\delta$ paired with location $X^*_j$. Indeed, $i \in C_{\delta jh}^Z$ implies $X_i = X^*_j$ and $\Delta_i = \delta$. Therefore, the nested product becomes
	\begin{align*}
		\prod_{j=1}^k \prod_{\delta=1}^D \prod_{i \in C_{\delta j}^X} \mu^e_\delta(dZ_i,dX^*_j)
		& = \prod_{j=1}^k \prod_{\delta=1}^D \prod_{h=1}^{r_{\delta j}} \prod_{i \in C_{\delta j h}^Z} \mu^e_\delta(dZ^*_{\delta jh},dX^*_j) \\[8pt]
		& = \prod_{j=1}^k \prod_{\delta=1}^D \prod_{h=1}^{r_{\delta j}}\mu^e_\delta(dZ^*_{\delta jh},dX^*_j)^{q_{\delta jh}},
	\end{align*}
	where $q_{\delta jh}=\mbox{card}(C_{\delta jh}^Z)$ is the number of observations in $C_{\delta jh}^Z$. As a result, the augmented likelihood is expressed as
	\begin{align*}
		\mathcal{L}(\bm{\crm^e}; \boldsymbol T_{1:n}, \boldsymbol \Delta_{1:n}, \boldsymbol X_{1:n}, \boldsymbol Z_{1:n}) 
		& = \left[ \prod_{j=1}^k \prod_{i \in C_j^X} k(T_i;X^*_j) \right] \left[ \prod_{j=1}^k \prod_{\delta=1}^D \prod_{h=1}^{r_{\delta j}}\mu^e_\delta(dZ^*_{\delta jh},dX^*_j)^{q_{\delta jh}} \right] \\
		& \qquad \times \exp \left( - \sum_{\ell=1}^D \int_{[0,1] \times \X} \sum_{i=1}^n \int_0^{T_i} k(s;x) \,ds \: \tilde \mu^e_\ell (dz,dx) \right),
	\end{align*}
	which coincides with \eqref{likelihood} by substituting the quantities in \eqref{kernel_quantities}. \hfill \qed
	
	\paragraph{Proof of Theorem~\ref{prop:marginal}}
	
	Consider the augmented likelihood function in \eqref{likelihood}, \vspace{-6pt}
	\begin{multline*}
		\mathcal{L}(\boldsymbol{\crm}^e; \mathcal{F}_n^*) \\
		= Q_n(\boldsymbol T_{1:n}, \boldsymbol X_{1:n}) \, \exp \left( - \sum_{\ell=1}^D \int_{[0,1] \times \X} K_n(x) \, \tilde \mu^e_\ell(dz,dx) \right) \prod_{\ell=1}^D \prod_{j=1}^k \prod_{h=1}^{r_{\ell j}} \tilde \mu^e_\ell(dZ^*_{\ell jh},dX^*_j)^{q_{\ell jh}},
	\end{multline*}
	where the quantities $Q_n(\boldsymbol T_{1:n}, \boldsymbol X_{1:n})$ and $K_n(x)$ are defined as in \eqref{kernel_quantities}. Conditionally on $\crm_0$, the extended random measures in $\boldsymbol{\crm}^e$ are independent CRMs on $[0,1] \times \X$, characterized by the homogeneous Lévy intensity $\tilde \nu^e(ds,dz,dx) = \rho(ds) \, H(dz) \, \crm_0(dx)$. The conditional expectation of the augmented likelihood, given $\crm_0$, is
	\begin{multline*}
		\mathbb{E} \left[ \mathcal{L}(\boldsymbol{\crm}^e; \mathcal{F}_n^*) \mid \crm_0 \right] = Q_n(\boldsymbol T_{1:n}, \boldsymbol X_{1:n}) \\ 
		\times \mathbb{E} \left[ \exp \left( - \sum_{\ell=1}^D \int_{[0,1] \times \X} K_n(x) \, \tilde \mu^e_\ell(dz,dx) \right) \prod_{\ell=1}^D \prod_{j=1}^k \prod_{h=1}^{r_{\ell j}} \tilde \mu^e_\ell(dZ^*_{\ell jh},dX^*_j)^{q_{\ell jh}} \bigm\vert \crm_0 \right].
	\end{multline*}
	By conditional independence of the extended random measures $\crm^e_1, \dots, \crm^e_D$, we have
	\begin{multline*}
		\mathbb{E} \left[ \mathcal{L}(\boldsymbol{\crm}^e; \mathcal{F}_n^*) \mid \crm_0 \right] = Q_n(\boldsymbol T_{1:n}, \boldsymbol X_{1:n}) \\
		\times \prod_{\ell=1}^D \mathbb{E} \left[ \exp \left( - \int_{[0,1] \times \X} K_n(x) \, \tilde \mu^e_\ell(dz,dx) \right) \prod_{j=1}^k \prod_{h=1}^{r_{\ell j}} \tilde \mu^e_\ell (dZ^*_{\ell jh},dX^*_j)^{q_{\ell jh}} \bigm\vert \crm_0 \right],
	\end{multline*}
	and exploiting the identity in \eqref{key_identity} with $\alpha(dz,dx) = H(dz) \, \crm_0(dx)$ and $f(z,x) = K_n(x)$, we obtain
	\begin{align*}
		\mathbb{E} \left[ \mathcal{L}(\boldsymbol{\crm}^e; \mathcal{F}_n^*) \mid \crm_0 \right] 
		& = Q_n(\boldsymbol T_{1:n}, \boldsymbol X_{1:n}) \, \prod_{\ell =1}^D \exp \left( - \int_\X \psi(K_n(x)) \, \crm_0(dx) \right) \\
		& \qquad \qquad \times \prod_{\ell =1}^D \prod_{j=1}^k \prod_{h=1}^{r_{\ell j}} \tau(q_{\ell jh}; K_n(X_j^*)) \, H(dZ_{\ell jh}^*) \, \crm_0(dX^*_j) \displaybreak \\
		& = Q_n(\boldsymbol T_{1:n}, \boldsymbol X_{1:n}) \, \exp \left( - \int_\X D \psi(K_n(x)) \, \crm_0(dx) \right) \\
		& \qquad \qquad \times \prod_{j=1}^k \crm_0(dX^*_j)^{r_j} \prod_{\ell=1}^D \prod_{h=1}^{r_{\ell j}} \tau(q_{\ell jh}; K_n(X_j^*)) \, H(dZ_{\ell jh}^*).
	\end{align*}
	The random measure $\crm_0$ at the root of the hierarchical prior is a CRM with homogeneous Lévy intensity $\nu_0(ds,dx) = \rho_0(ds) \, \Lambda_0(dx)$. Therefore, the joint marginal distribution of the observations $(\boldsymbol T_{1:n}, \boldsymbol \Delta_{1:n})$ and the latent variables $(\boldsymbol X_{1:n}, \boldsymbol Z_{1:n})$, i.e.~the expectation of the augmented likelihood, is
	\begin{align*}
		\Prob (\mathcal{F}_n^*) 
		& = \mathbb{E} \left[ \mathcal{L}(\boldsymbol{\crm}^e; \mathcal{F}_n^*) \right] 
		= \mathbb{E} \left[ \, \mathbb{E} \left[ \mathcal{L}(\boldsymbol{\crm}^e; \mathcal{F}_n^*) \mid \crm_0 \right] \, \right] \\[8pt]
		& = Q_n(\boldsymbol T_{1:n}, \boldsymbol X_{1:n}) \, \mathbb{E} \left[ \exp \left( - \int_\X D \psi(K_n(x)) \, \crm_0(dx) \right) \prod_{j=1}^k \crm_0(dX^*_j)^{r_j} \right] \\[8pt]
		& \qquad \qquad \times \prod_{j=1}^k \prod_{\ell=1}^D \prod_{h=1}^{r_{\ell j}} \tau(q_{\ell jh}; K_n(X_j^*)) \, H(dZ_{\ell jh}^*).
	\end{align*}
	Leveraging again \eqref{key_identity} with $\alpha(dx) = \Lambda_0(dx)$ and $f(x) = D\psi(K_n(x))$, we conclude that
	\begin{multline*}
		\Prob (\mathcal{F}_n^*)
		= Q_n(\boldsymbol T_{1:n}, \boldsymbol X_{1:n}) \, \exp \left( - \int_\X \psi_0 \big( D \psi(K_n(x)) \big) \, \Lambda_0(dx) \right) \qquad \qquad \\[8pt]
		\times \prod_{j=1}^k \tau_0 \big( r_j ; D \psi(K_n(X_j^*)) \big) \, \Lambda_0(dX_j^*) \, \prod_{\ell =1}^D \prod_{h=1}^{r_{\ell j}} \tau(q_{\ell jh}; K_n(X_j^*)) \, H(dZ_{\ell jh}^*).
	\end{multline*}
	\hfill \qed
	
	\paragraph{Proof of Theorem~\ref{prop:predict_competing}}
	
	The predictive distribution for the observation $(T_{n+1}, \Delta_{n+1})$ and the corresponding latent location $X_{n+1}$ and latent mark $Z_{n+1}$ is obtained from the joint marginal distribution in Theorem~\ref{prop:marginal}, and given by
	\begin{equation} \label{eq:pred_supp}
		\Prob(T_{n+1}, \Delta_{n+1}, X_{n+1}, Z_{n+1} \mid \mathcal{F}_n^*) = \frac{\displaystyle \Prob(T_{n+1}, \Delta_{n+1}, X_{n+1}, Z_{n+1}, \mathcal{F}_n^*)}{\displaystyle \Prob(\mathcal{F}_n^*)}.
	\end{equation}
	where the numerator is the joint marginal distribution of the $n+1$ observations and corresponding latent variables. 
	
	The analytical form of the predictive distribution \eqref{eq:pred_supp} depends on the allocation of the observation $n+1$ within the latent partition structure, i.e.~on the specific values assumed by $X_{n+1}$ and $Z_{n+1}$. Three alternative cases can be identified:
	\begin{enumerate}
		\item[(1)] both $X_{n+1}$ and $Z_{n+1}$ display ties with values in $\boldsymbol X^*$ and $\boldsymbol Z^*$, respectively, say $X_{n+1} = X^*_j$ and $Z_{n+1} = Z^*_{\delta jh}$, where necessarily $\Delta_{n+1} = \delta$;
		\item[(2)] $X_{n+1}$ displays a tie with values in $\boldsymbol X^*$, say $X_{n+1} = X^*_j$, while $Z_{n+1}$ assumes a value not included in $\boldsymbol Z^*$;
		\item[(3)] both $X_{n+1}$ and $Z_{n+1}$ assume values not included in $\boldsymbol X^*$ and $\boldsymbol Z^*$, respectively.
	\end{enumerate}
	These three cases are highlighted in the following expression for the numerator:
	\begin{multline*}
		\Prob(T_{n+1} \in dt, \Delta_{n+1} = \delta, X_{n+1} \in dx, Z_{n+1} \in dz, \mathcal{F}_n^*) \\[8pt]
		= Q_n(\boldsymbol T_{1:n}, \boldsymbol X_{1:n}) \, k(t; x) \, \exp \left( - \int_\X \psi_0 \big( D \psi(K_{n+1}(y)) \big) \, \Lambda_0(dy) \right) \\[8pt]
		\times \prod_{j=1}^k \left( \prod_{\ell=1}^D \prod_{h=1}^{r_{\ell j}} \tau(q_{\ell jh}; K_{n+1}(X_j^*)) \, H(dZ_{\ell jh}^*) \right) \tau_0 \big( r_j ; D \psi(K_{n+1}(X_j^*)) \big) \, \Lambda_0(dX_j^*) \\[8pt]
		\times \Bigg\{ \sum_{j=1}^k \sum_{h=1}^{r_{\delta j}} \frac{\tau(q_{\delta jh}+1; K_{n+1}(X_j^*))}{\tau(q_{\delta jh}; K_{n+1}(X_j^*))} \, \delta_{X^*_j}(dx) \, \delta_{Z^*_{\delta jh}}(dz) \qquad \leftarrow \ \text{case (1)} \\[8pt]
		\text{case (2)} \ \rightarrow \qquad + \sum_{j=1}^k \tau(1; K_{n+1}(X_j^*)) \, \frac{\tau_0 ( r_j + 1 ; D \psi(K_{n+1}(X_j^*)))}{\tau_0 ( r_j ; D \psi(K_{n+1}(X_j^*)) )} \, \delta_{X^*_j}(dx) \, H(dz) \\[8pt]
		\text{case (3)} \ \rightarrow \qquad + \, \tau(1; K_{n+1}(x)) \, \tau_0 ( 1 ; D \psi(K_{n+1}(x))) \, \Lambda_0(dx) \, H(dz) \Bigg\} \,dt,
	\end{multline*}
	where $K_{n+1}(x) = K_n(x) + \displaystyle \int_0^t k(s,x)\,ds$ is the updated kernel term, since $T_{n+1} = t$. Remarkably, these three cases coincide with the three cases appearing in the description of the marginal sampling algorithm (Section~\ref{sec:marginalMCMC}). The joint predictive distribution is thus obtained by dividing the expression above by the joint marginal distribution in Theorem~\ref{prop:marginal}: 
	\begin{multline*} 
		\Prob(T_{n+1} \in dt, \Delta_{n+1} = \delta, X_{n+1} \in dx, Z_{n+1} \in dz \mid \mathcal{F}_n^*) \\[8pt]
		= k(t; x) \, \exp \left( - \int_\X \left( \psi_0 \big( D \psi(K_{n+1}(y)) \big) - \psi_0 \big( D \psi(K_n(y)) \big) \right) \Lambda_0(dy) \right) \\[8pt]
		\times \prod_{j=1}^k \left( \prod_{\ell=1}^D \prod_{h=1}^{r_{\ell j}} \frac{\tau(q_{\ell jh}; K_{n+1}(X_j^*))}{\tau(q_{\ell jh}; K_n(X_j^*))} \right) \frac{\tau_0 \big( r_j ; D \psi(K_{n+1}(X_j^*)) \big)}{\tau_0 \big( r_j ; D \psi(K_n(X_j^*)) \big)} \displaybreak \\[8pt]
		\times \Bigg\{ \sum_{j=1}^k \sum_{h=1}^{r_{\delta j}} \frac{\tau(q_{\delta jh}+1; K_{n+1}(X_j^*))}{\tau(q_{\delta jh}; K_{n+1}(X_j^*))} \, \delta_{X^*_j}(dx) \, \delta_{Z^*_{\delta jh}}(dz) \qquad \qquad \qquad \\[8pt]
		\qquad \qquad \qquad \qquad + \sum_{j=1}^k \tau(1; K_{n+1}(X_j^*)) \, \frac{\tau_0 ( r_j +1 ; D \psi(K_{n+1}(X_j^*)))}{\tau_0 ( r_j ; D \psi(K_{n+1}(X_j^*)) )} \, \delta_{X^*_j}(dx) \, H(dz) \\[8pt]
		+ \, \tau(1; K_{n+1}(x)) \, \tau_0 ( 1 ; D \psi(K_{n+1}(x))) \, \Lambda_0(dx) \, H(dz) \Bigg\} \,dt.
	\end{multline*}
	The predictive distribution for observation $(T_{n+1}, \Delta_{n+1})$ is obtained by marginalizing the latent variables
	\begin{multline} \label{survival_predictive}
		\Prob(T_{n+1} \in dt, \Delta_{n+1} = \delta \mid \mathcal{F}_n) = \exp \left( - \int_\X \left( \psi_0 \big( D \psi(K_{n+1}(y)) \big) - \psi_0 \big( D \psi(K_n(y)) \big) \right) \Lambda_0(dy) \right) \\[8pt] 
		\qquad \qquad \times \prod_{j=1}^k \left( \prod_{\ell=1}^D \prod_{h=1}^{r_{\ell j}} \frac{\tau(q_{\ell jh}; K_{n+1}(X_j^*))}{\tau(q_{\ell jh}; K_n(X_j^*))} \right) \frac{\tau_0 \big( r_j ; D \psi(K_{n+1}(X_j^*)) \big)}{\tau_0 \big( r_j ; D \psi(K_n(X_j^*)) \big)} \\[8pt]
		\times \Bigg\{ \sum_{j=1}^k k(t; X^*_j) \sum_{h=1}^{r_{\delta j}} \frac{\tau(q_{\delta jh}+1; K_{n+1}(X_j^*))}{\tau(q_{\delta jh}; K_{n+1}(X_j^*))} \qquad \qquad \qquad \qquad \\[8pt]
		\qquad \qquad + \sum_{j=1}^k k(t; X^*_j) \, \tau(1; K_{n+1}(X_j^*)) \, \frac{\tau_0 ( r_j +1 ; D \psi(K_{n+1}(X_j^*)))}{\tau_0 ( r_j ; D \psi(K_{n+1}(X_j^*)) )} \\[8pt]
		+ \int_\X k(t; x) \, \tau(1; K_{n+1}(x)) \, \tau_0 ( 1 ; D \psi(K_{n+1}(x))) \, \Lambda_0(dx) \Bigg\} \,dt.
	\end{multline}
	Finally, the predictive distribution of $\Delta_{n+1}$, given the survival time $T_{n+1}$ is proportional to the expression above, in which factors not depending on $\Delta_{n+1} = \delta$ can be ignored:
	\begin{multline*}
		\Prob(\Delta_{n+1} = \delta \mid T_{n+1} = t, \, \mathcal{F}_n)
		\, \propto \, \sum_{j=1}^k k(t; X^*_j) \sum_{h=1}^{r_{\delta j}} \frac{\tau(q_{\delta jh}+1; K_{n+1}(X_j^*))}{\tau(q_{\delta jh}; K_{n+1}(X_j^*))} \\[8pt]
		+ \sum_{j=1}^k k(t; X^*_j) \, \tau(1; K_{n+1}(X_j^*)) \, \frac{\tau_0 ( r_j +1 ; D \psi(K_{n+1}(X_j^*)))}{\tau_0 ( r_j ; D \psi(K_{n+1}(X_j^*)) )} \\[8pt]
		+ \int_\X k(t; x) \,\tau(1; K_{n+1}(x)) \, \tau_0 ( 1 ; D \, \psi(K_{n+1}(x))) \, \Lambda_0(dx).
	\end{multline*}
	\hfill \qed
	
	\paragraph{Proof of Theorem~\ref{prop:posterior}}
	
	The posterior distribution of the random measure $\crm_0$ at the root of the hierarchy, given $\mathcal{F}_n^* = (\bm{T}_{1:n},\bm{\Delta}_{1:n}, \bm{X}_{1:n}, \bm{Z}_{1:n})$, can be characterized through its conditional Laplace transform,
	\begin{equation*}
		\mathbb{E} \left[ \exp \left( - \int_\X h_0(x) \, \crm_0(dx) \right) \bigm\vert \mathcal{F}_n^* \,\right]
		= \frac{\displaystyle \mathbb{E} \left[ \exp \left( - \int_\X h_0(x) \, \crm_0(dx) \right) \mathcal{L}(\bm{\crm}^e; \mathcal{F}_n^*) \right]}{\mathbb{E} \left[ \mathcal{L}(\bm{\crm}^e; \mathcal{F}_n^*) \right]},
	\end{equation*}
	where $h_0 \colon \X \mapsto \R^+$ is any non-negative measurable function. The quantity at the denominator is the joint marginal distribution in Theorem~\ref{prop:marginal}, while the numerator can be computed similarly. Specifically, the expectation at the numerator, conditionally on $\crm_0$, is
	\begin{multline*}
		\mathbb{E} \left[ \exp \left( - \int_\X h_0(x) \, \crm_0(dx) \right) \mathcal{L}(\bm{\crm}^e; \mathcal{F}_n^*) \bigm\vert \crm_0 \right] \\[8pt]
		\begin{aligned}
			& = \exp \left( - \int_\X h_0(x) \, \crm_0(dx) \right) \mathbb{E} \left[ \mathcal{L}(\bm{\crm}^e; \mathcal{F}_n^*) \mid \crm_0 \right] \\[8pt]
			& = Q_n(\boldsymbol T_{1:n}, \boldsymbol X_{1:n}) \, \exp \left( - \int_\X \left(h_0(x) + D \psi(K_n(x))\right) \crm_0(dx) \right) 
		\end{aligned} \\
		\times \prod_{j=1}^k \crm_0(dX^*_j)^{r_j} \prod_{\ell =1}^D \prod_{h=1}^{r_{\ell j}} \tau(q_{\ell jh}; K_n(X_j^*)) \, H(dZ_{\ell jh}^*).
	\end{multline*}
	By marginalization of the conditional expectation with respect to $\crm_0$, we get
	\begin{multline*}
		\mathbb{E} \left[ \exp \left( - \int_\X h_0(x) \, \crm_0(dx) \right) \mathcal{L}(\bm{\crm}^e; \mathcal{F}_n^*) \right] \\[8pt]
		\begin{aligned}
			& = Q_n(\boldsymbol T_{1:n}, \boldsymbol X_{1:n}) \ \mathbb{E} \left[ \exp \left( - \int_\X \left(h_0(x) + D \psi(K_n(x))\right) \crm_0(dx) \right) \prod_{j=1}^k \crm_0(dX^*_j)^{r_j} \right] \\
			& \qquad \qquad \times \prod_{j=1}^k \prod_{\ell =1}^D \prod_{h=1}^{r_{\ell j}} \tau(q_{\ell jh}; K_n(X_j^*)) \, H(dZ_{\ell jh}^*) \\[8pt]
			& = Q_n(\boldsymbol T_{1:n}, \boldsymbol X_{1:n}) \, \exp \left( - \int_\X \psi_0 \big( h_0(x) + D \psi(K_n(x)) \big) \, \Lambda_0(dx) \right) \prod_{j=1}^k \Lambda_0(dX_j^*) \\
			& \qquad \qquad \times \prod_{j=1}^k \tau_0 \big( r_j ; h_0(X^*_j) + D \psi(K_n(X_j^*)) \big) \prod_{\ell =1}^D \prod_{h=1}^{r_{\ell j}} \tau(q_{\ell jh}; K_n(X_j^*)) \, H(dZ_{\ell jh}^*),
		\end{aligned}
	\end{multline*}
	where the last identity is obtained resorting to \eqref{key_identity} with $f(x) = h_0(x) + D \psi(K_n(x))$.
	Therefore, the conditional Laplace transform is \cut
	\begin{multline*}
		\mathbb{E} \left[ \exp \left( - \int_\X h_0(x) \, \crm_0(dx) \right) \bigm\vert \mathcal{F}_n^*\, \right] \\[8pt]
		= \exp \left( - \int_\X \Big( \psi_0 \big( h_0(x) + D  \psi(K(x)) \big) - \psi_0 \big(D \psi(K(x)) \big) \Big) \, \Lambda_0(dx) \right) \\
		\times \prod_{j=1}^k \frac{\tau_0 \big( r_j ; h_0(X^*_j) + D \psi(K(X_j^*)) \big)}{\tau_0 \big( r_j; D \psi(K(X_j^*)) \big)}.
	\end{multline*}
	Thanks to simple algebra with Laplace exponents, the exponential term can be rewritten as
	\begin{multline*}
		\exp \left( - \int_\X \Big( \psi_0 \big( h_0(x) + D \psi(K_n(x)) \big) - \psi_0 \big(D \psi(K_n(x)) \big) \Big) \, \Lambda_0(dx) \right) \\[8pt]
		= \exp \left( - \int_\X \psi_{0,x}^+(h_0(x)) \, \Lambda_0(dx) \right) = \mathbb{E} \left[ \exp \left( - \int_\X h_0(x) \, \crm^+_0(dx) \right) \right],
	\end{multline*}
	where $\psi_{0,x}^+(\cdot)$ is the Laplace exponent corresponding to $\rho_0^+(ds \,\vert\, x) = \exp \left( - D \psi(K_n(x))\,s \right) \, \rho_0(ds)$, and $\crm^+_0$ is a CRM with Lévy intensity $\nu^+_0(ds,dx) = \rho_0^+(ds \,\vert\, x) \, \Lambda_0(dx)$. Moreover, for each $j = 1, \dots, k$,
	\begin{align*}
		\tau_0 \big( r_j ; h_0(X^*_j) + D \psi(K_n(X_j^*)) \big) & = \int_{\R^+} s^{r_j} e^{-h_0(X^*_j)\,s} \, \exp \left( - D \psi(K_n(X_j^*)\,s \right) \,\rho_0(ds) \\[8pt]
		& = \int_{\R^+} e^{-h_0(X^*_j)\,s} \, s^{r_j} \rho_{0}^+(ds \mid X^*_j).
	\end{align*}
	Therefore, the ratios in the expression of the conditional Laplace transform can be regarded as expectations,
	\begin{equation*}
		\frac{\tau_0 \big( r_j ; h_0(X^*_j) + D \psi(K_n(X_j^*)) \big)}{\tau_0 \big( r_j; D \psi(K_n(X_j^*)) \big)} = \frac{\displaystyle \int_{\R^+} e^{-h_0(X^*_j)\,s} \, s^{r_j} \rho_0^+(ds \mid X^*_j) }{\displaystyle \int_{\R^+} s^{r_j} \rho_0^+(ds \mid X^*_j)} = \mathbb{E} \left[ \exp \left( -h_0(X^*_j)\,V_j \right) \right],
	\end{equation*}
	where each $V_j$ is a random variable with density function
	\begin{equation*}
		f_{V_j}(s) \, ds \propto s^{r_j} \exp \left( - D \psi(K_n(X_j^*))\,s \right) \, \rho_0(ds) = s^{r_j} \rho_0^+(ds \mid X^*_j).
	\end{equation*}
	In conclusion, the conditional Laplace transform can be written as
	\begin{align*}
		\mathbb{E} \left[ \exp \left( - \int_\X h_0(x) \, \crm_0(dx) \right) \bigm\vert \mathcal{F}_n^*\, \right]
		& = \mathbb{E} \left[ \exp \left( - \int_\X h_0(x) \, \crm^+_0(dx) \right) \right] \prod_{j=1}^k \mathbb{E} \left[ \exp \left( -h_0(X^*_j)\,V_j \right) \right] \\[8pt]
		& = \mathbb{E} \left[ \exp \left( - \int_\X h_0(x) \, \crm^+_0(dx) - \sum_{j=1}^k h_0(X^*_j)\,V_j \right) \right],
	\end{align*}
	which highlights the structure of the posterior distribution of $\crm_0$ as the sum of the non-homogeneous CRM $\crm_0^*$ and the random jumps $V_1, \dots, V_k$ at the fixed locations $X^*_1, \dots, X^*_k$; moreover, these components are mutually independent (second line). 
	
	A similar characterization holds for the posterior distribution of the extended random measures $\crm^e_1, \dots, \crm^e_D$ at the bottom of the hierarchy. Indeed, their joint conditional Laplace transform, given $\mathcal{F}_n^* = (\bm{T}_{1:n},\bm{\Delta}_{1:n}, \bm{X}_{1:n}, \bm{Z}_{1:n})$ and the root random measure $\crm_0$, is
	\begin{multline*}
		\mathbb{E} \left[ \exp \left( - \sum_{\ell=1}^D \int_{[0,1] \times \X} h_\ell(z,x) \, \crm^e_\ell(dz,dx) \right) \bigm\vert \mathcal{F}_n^*, \crm_0 \right] \\
		= \frac{\displaystyle \mathbb{E} \left[ \exp \left( - \sum_{\ell=1}^D \int_{[0,1] \times \X} h_\ell(z,x) \, \crm^e_\ell(dz,dx) \right) \mathcal{L}(\bm{\crm}^e; \mathcal{F}_n^*) \bigm\vert \crm_0 \right]}{\mathbb{E} \left[ \mathcal{L}(\bm{\crm}^e; \mathcal{F}_n^*) \mid \crm_0 \right]},
	\end{multline*}
	where $h_\delta \colon [0,1] \times \X \mapsto \R^+$, for $\delta \in \set{1, \dots, D}$, are non-negative measurable functions. The quantity at the denominator is the conditional expectation of the augmented likelihood, obtained as an intermediate step in the proof of Theorem~\ref{prop:marginal}, while the expectation at the numerator can be computed similarly. In particular, we have
	\begin{multline*}
		\mathbb{E} \left[ \exp \left( - \sum_{\ell=1}^D \int_{[0,1] \times \X} h_\ell(z,x) \, \crm^e_\ell(dz,dx) \right) \mathcal{L}(\bm{\crm}^e; \mathcal{F}_n^*) \bigm\vert \crm_0 \right] \\
		= Q(\boldsymbol T, \boldsymbol X) \, \mathbb{E} \Bigg[ \exp \left( - \sum_{\ell=1}^D \int_{[0,1] \times \X} \big( h_\ell(z,x) + K_n(x) \big) \, \tilde \mu^e_\ell(dz,dx) \right) \\ 
		\times \prod_{\ell=1}^D \prod_{j=1}^k \prod_{h=1}^{r_{\ell j}} \tilde \mu^e_\ell(dZ^*_{\ell jh},dX^*_j)^{q_{\ell jh}} \Bigg].
	\end{multline*}
	By conditional independence of the extended random measures $\crm^e_1, \dots, \crm^e_D$, and leveraging again \eqref{key_identity} with $f(z,x) = h_\delta(z,x) + K_n(x)$, for $\delta = 1, \dots, D$, we obtain
	\begin{align*}
		& \mathbb{E} \left[ \exp \left( - \sum_{\ell=1}^D \int_{[0,1] \times \X} h_\ell(z,x) \, \crm^e_\ell(dz,dx) \right) \mathcal{L}(\bm{\crm}^e; \mathcal{F}_n^*) \bigm\vert \crm_0 \right] \\
		& \hspace{6em} = Q(\boldsymbol T, \boldsymbol X) \, \prod_{\ell=1}^D \mathbb{E} \Bigg[ \exp \left( - \int_{[0,1] \times \X} \big( h_\ell(z,x) + K_n(x) \big) \, \tilde \mu^e_\ell(dz,dx) \right) \\
		& \hspace{6em} \qquad \qquad \times \prod_{j=1}^k \prod_{h=1}^{r_{\ell j}} \tilde \mu^e_\ell(dZ^*_{\ell jh},dX^*_j)^{q_{\ell jh}} \Bigg] \displaybreak \\
		& \hspace{6em} = Q(\boldsymbol T, \boldsymbol X) \, \prod_{\ell =1}^D \exp \left( - \int_{[0,1] \times \X} \psi(h_\ell(z,x) + K_n(x)) \, H(dz)\, \crm_0(dx) \right) \\
		& \hspace{6em} \qquad \qquad \times \prod_{\ell=1}^D \prod_{j=1}^k \prod_{h=1}^{r_{\ell j}} \tau(q_{\ell jh}; h_\ell(Z^*_{\ell jh}, X^*_j) + K_n(X_j^*)) \, H(dZ_{\ell jh}^*) \, \crm_0(dX^*_j).
	\end{align*}
	Therefore, the conditional Laplace transform is
	\begin{multline*}
		\mathbb{E} \left[ \exp \left( - \sum_{\ell=1}^D \int_{[0,1] \times \X} h_\ell(z,x) \, \crm^e_\ell(dz,dx) \right) \bigm\vert \mathcal{F}_n^*, \crm_0 \right] \\[8pt]
		= \prod_{\ell=1}^D \exp \left( - \int_{[0,1] \times \X} \Big( \psi(h_\ell(z,x) + K_n(x)) - \psi(K_n(x)) \Big) H(dz)\, \crm_0(dx) \right) \\
		\times \prod_{\ell=1}^D \prod_{j=1}^k \prod_{h=1}^{r_{\ell j}} \frac{\tau(q_{\ell jh}; h_\ell(Z^*_{\ell jh}, X^*_j) + K_n(X_j^*))}{\tau(q_{\ell jh}; K_n(X_j^*))}.
	\end{multline*}
	For each $\delta = 1, \dots, D$, the exponential term can be rewritten as
	\begin{align*}
		\exp & \left( - \int_{[0,1] \times \X} \Big( \psi(h_\delta(z,x) + K_n(x)) - \psi(K_n(x)) \Big) \, H(dz)\, \crm_0(dx) \right) \\[8pt]
		& \qquad \qquad \qquad = \mathbb{E} \left[ \exp \left( - \int_{[0,1] \times \X} h_\delta(z,x) \, \crm^+_\delta(dz,dx) \right) \right],
	\end{align*}
	where $\crm^+_\delta$ is a CRM with non-homogeneous Lévy intensity $\nu^+(ds,dz,dx) = \rho^+(ds \,\vert\, x) \, H(dz)\, \crm_0(dx)$, with jump component $\rho^+(ds \,\vert\, x) = e^{ - K_n(x)\,s} \rho(ds)$. 
	
	Moreover, for each $\delta = 1, \dots, D$, $j = 1, \dots, k$, and $h = 1, \dots, r_{\delta j}$,
	\begin{equation*}
		\tau(q_{\delta jh}; h_\delta(Z^*_{\delta jh}, X^*_j) + K_n(X_j^*)) = \int_{\R^+} e^{-h_\delta(Z^*_{\delta jh}, X^*_j)\,s} \, s^{q_{\delta jh}} \rho^+(ds \mid X^*_j).
	\end{equation*}
	Therefore, the ratios in the expression of the conditional Laplace transform can be regarded as expectations
	\begin{equation*}
		\frac{\tau(q_{\delta jh}; h_\delta(Z^*_{\delta jh}, X^*_j) + K_n(X_j^*))}{\tau(q_{\delta jh}; K_n(X_j^*))} = \mathbb{E} \left[ \exp \left( - h_\delta(Z^*_{\delta jh}, X^*_j)\,S_{\delta jh} \right) \right],
	\end{equation*}
	where each $S_{\delta jh}$ is a non-negative random variable having density function proportional to
	\begin{equation*}
		s^{q_{\delta jh}} \, e^{ - K_n(X^*_j)\,s} \, \rho(ds) = s^{q_{\delta jh}} \rho^+(ds \mid X^*_j).
	\end{equation*}
	In conclusion, the conditional Laplace transform can be written as
	\begin{align*}
		& \mathbb{E} \left[ \exp \left( - \sum_{\ell=1}^D \int_{[0,1] \times \X} h_\ell(z,x) \, \crm^e_\ell(dz,dx) \right) \bigm\vert \mathcal{F}^*_n, \crm_0 \right] \\[8pt]
		& \qquad \qquad = \prod_{\ell=1}^D \mathbb{E} \left[ \exp \left( - \int_{[0,1] \times \X} h_\ell(z,x) \, \crm^+_\ell(dz,dx) \right) \right] \prod_{j=1}^k \prod_{h=1}^{r_{\ell j}} \mathbb{E} \left[ \exp \left( - h_\ell(Z^*_{\ell jh}, X^*_j)\,S_{\ell jh} \right) \right] \\[8pt]
		& \qquad \qquad = \prod_{\ell=1}^D \mathbb{E} \left[ \exp \left( - \int_{[0,1] \times \X} h_\ell(z,x) \, \crm^+_\ell(dz,dx) - \sum_{j=1}^k \sum_{h=1}^{r_{\ell j}} h_\ell(Z^*_{\ell jh}, X^*_j)\,S_{\ell jh} \right) \right],
	\end{align*}
	which highlights the structure of the posterior distribution of each $\crm_\delta$, given the root random measure $\crm_0$, as the sum of the non-homogeneous CRM $\crm_\delta^*$ and the random jumps $(S_{\delta jh})_{jh}$ at fixed locations. These components are mutually independent, and the random measures $\crm_1^e, \dots, \crm_D^e$ are conditionally independent, given $\crm_0$. The posterior characterization in Theorem~\ref{prop:posterior} is obtained by marginalization with respect to the first component of the random measure, that is $\crm_\delta(dx) = \crm^e_\delta([0,1],dx)$, for $\delta = 1, \dots, D$.
	\hfill \qed 
	
	\paragraph{Proof of Proposition~\ref{prop:posterior_survival}}
	
	The posterior estimate of the overall survival function $\tilde S(t)$, given $\mathcal{F}_n = (\bm{T}_{1:n},\bm{\Delta}_{1:n}, \Pi_{\bm{X}_{1:n}}, \Pi_{\bm{Z}_{1:n}}, \bm{X}^*)$, at $t \ge 0$, is obtained as
	\begin{align*}
		\E \left[ \tilde S(t) \mid \mathcal{F}_n \right]
		& = \E \left[ \exp \left( - \sum_{\ell=1}^D \int_0^t \int_\X k(s;x) \, \tilde \mu_\ell(dx)\,ds \right) \bigm\vert \mathcal{F}_n \right] \\[8pt]
		& = \E \left[ \exp \left( - \sum_{\ell=1}^D \int_\X K_1(x) \, \tilde \mu_\ell(dx) \right) \bigm\vert \mathcal{F}_n \right],
	\end{align*}
	where $K_1(x) = K_1(x; t) = \displaystyle \int_0^t k(s;x)\,ds$ is the integrated kernel up to time $t \ge 0$. Conditionally on the root random measure $\crm_0$, the expressions for the posterior distributions of the random measures $\crm_1, \dots, \crm_D$ derived in Theorem~\ref{prop:posterior} are substituted above, leading to
	\begin{align*}
		\E \left[ \tilde S(t) \mid \mathcal{F}_n, \crm_0 \right] & = \E \left[ \exp \left( - \sum_{\ell=1}^D \int_\X K_1(x) \, \tilde \mu^+_\ell(dx) - \sum_{\ell=1}^D \sum_{j=1}^k \sum_{h=1}^{r_{\ell j}} K_1(X^*_j) \, S_{\ell jh} \right) \Bigm\vert \crm_0 \right] \\[8pt]
		& = \E \left[ \prod_{\ell=1}^D \exp \left( - \int_\X K_1(x) \, \tilde \mu^+_\ell(dx) \right) \prod_{j=1}^k \prod_{h=1}^{r_{\ell j}} \exp \left( - K_1(X^*_j) \, S_{\ell jh} \right) \Bigm\vert \crm_0 \right];
	\end{align*}
	By the conditional independence \emph{a posteriori} of the random measures  $\crm_1, \dots, \crm_D$, and of the their components, we have \cut
	\begin{equation*}
		\E \left[ \tilde S(t) \mid \mathcal{F}_n, \crm_0 \right] = \prod_{\ell=1}^D \E \left[ \exp \left( - \int_\X K_1(x) \, \tilde \mu^+_\ell(dx) \right) \bigm\vert \crm_0 \right] \prod_{j=1}^k \prod_{h=1}^{r_{\ell j}} \E \left[ \exp \left( - K_1(X^*_j) \, S_{\ell jh} \right) \right].
	\end{equation*}
	The posterior characterization of the random measures $\crm_1, \dots, \crm_D$ implies that
	\begin{align*}
		\E \left[ \tilde S(t) \mid \mathcal{F}_n, \crm_0 \right] & = \prod_{\ell=1}^D \exp \left( - \int_\X \psi^+_{x}(K_1(x)) \, \tilde \mu_0(dx) \right) \prod_{j=1}^k \prod_{h=1}^{r_{\ell j}} \frac{\tau^+_{X^*_j}(q_{\ell jh}; K_1(X^*_j))}{\tau^+_{X^*_j}(q_{\ell jh}; 0)} \\[8pt]
		& = \exp \left( - \int_\X D \psi^+_{x}(K_1(x)) \, \tilde \mu_0(dx) \right) \prod_{\ell=1}^D \prod_{j=1}^k \prod_{h=1}^{r_{\ell j}} \frac{\tau^+_{X^*_j}(q_{\ell jh}; K_1(X^*_j))}{\tau^+_{X^*_j}(q_{\ell jh}; 0)},
	\end{align*}
	where $\psi^+_x$ and $\tau^+_x$ are the Laplace exponent and cumulants, respectively, corresponding to $\rho^+$ and defined in \eqref{eq:cumulants_posterior}. The posterior estimate is obtained by substituting the specification of the posterior distribution of the root random measure $\crm_0$, and exploiting the mutual independence of its components,
	\begin{multline*}
		\E \left[ \tilde S(t) \mid \mathcal{F}_n \right] = \E \left[ \exp \left( - \int_\X D \psi^+_x(K_1(x)) \, \tilde \mu^+_0(dx) \right) \right] \\
		\times \prod_{j=1}^k \E \left[ \exp \left( - D \psi^+_{X^*_j}(K_1(X^*_j)) \, V_j \right) \right] \prod_{\ell=1}^D \prod_{j=1}^k \prod_{h=1}^{r_{\ell j}} \frac{\tau^+_{X^*_j}(q_{\ell jh}; K_1(X^*_j))}{\tau^+_{X^*_j}(q_{\ell jh}; 0)}.
	\end{multline*}
	In conclusion, the posterior characterization of the random measure $\crm_0$ implies that
	\begin{multline*}
		\E \left[ \tilde S(t) \mid \mathcal{F}_n \right] = \exp \left( - \int_\X \psi_{0,x}^+ \big( D \psi^+_x(K_1(x))\big) \, \Lambda_0(dx) \right) \\
		\times \prod_{j=1}^k \frac{\tau_{0,X^*_j}^+\big(r_j; D \psi^+_{X^*_j}(K_1(X_j^*)) \big)}{\tau_{0,X^*_j}^+(r_j; 0)} \, \prod_{\ell=1}^D \prod_{j=1}^k \prod_{h=1}^{r_{\ell j}} \frac{\tau^+_{X^*_j}(q_{\ell jh}; K_1(X^*_j))}{\tau^+_{X^*_j}(q_{\ell jh}; 0)},
	\end{multline*}
	where $\psi_{0,x}^+$ and $\tau_{0,x}^+$ are the Laplace exponent and cumulants corresponding to $\rho_0^+$. 
	\hfill \qed
	
	\paragraph{Proof of Proposition~\ref{prop:posterior_incidence}}
	
	The proof follows a similar line of reasoning to that in the proof of Proposition~\ref{prop:posterior_survival}. Indeed, the posterior estimate of the cause-specific incidence function $\tilde p(dt, \delta) = \tilde f_\delta(t)\,dt$ for $\delta \in \set{1, \dots, D}$, given $\mathcal{F}_n = (\bm{T}_{1:n},\bm{\Delta}_{1:n}, \Pi_{\bm{X}_{1:n}}, \Pi_{\bm{Z}_{1:n}}, \bm{X}^*)$, at $t \ge 0$, is obtained as
	\begin{equation*}
		\E \big[ \tilde p(dt, \delta) \mid \mathcal{F}_n \big] = \E \left[ \int_\X k(t;x) \, \crm_\delta(dx) \, \exp \left( - \sum_{\ell=1}^D \int_\X K_1(y) \, \tilde \mu_\ell(dy) \right) dt \bigm\vert \mathcal{F}_n \right],
	\end{equation*}
	where $K_1(x) = K_1(x; t) = \displaystyle \int_0^t k(s;x)\,ds$ is the integrated kernel up to time $t \ge 0$. 
	Conditionally on the root random measure $\crm_0$, the expressions for the posterior distributions of the random measures $\crm_1, \dots, \crm_D$ are inserted into the previous formula yielding
	\begin{multline*}
		\E \big[ \tilde p(dt, \delta) \mid \mathcal{F}_n, \crm_0 \big] = \E \Bigg[ \prod_{\ell \neq \delta} \exp \left( - \int_\X K_1(y) \, \tilde \mu^+_\ell(dy) \right) \prod_{j=1}^k \prod_{h=1}^{r_{\ell j}} \exp \left( - K_1(X^*_j) \, S_{\ell jh} \right) \\[8pt]
		\times \Bigg( \int_\X k(t;x) \, \crm^+_\delta(dx) \, \exp \left( - \int_\X K_1(y) \, \tilde \mu^+_\delta(dy) \right) \prod_{j=1}^k \prod_{h=1}^{r_{\delta j}} \exp \left( - K_1(X^*_j) \, S_{\delta jh} \right) \\[8pt]
		+ \sum_{\ell=1}^k k(t;X^*_\ell) \sum_{\varepsilon=1}^{r_{\delta \ell}} \exp \left( - \int_\X K_1(y) \, \tilde \mu^+_\delta(dy) \right) S_{\delta \ell \varepsilon} \prod_{j=1}^k \prod_{h=1}^{r_{\delta j}} \exp \left( - K_1(X^*_j) \, S_{\delta jh} \right) \Bigg) dt \bigm\vert \crm_0 \Bigg].
	\end{multline*}
	By the conditional independence \emph{a posteriori} of $\crm_1, \dots, \crm_D$, and of the their components,
	\begin{multline*}
		\E \big[ \tilde p(dt, \delta) \mid \mathcal{F}_n, \crm_0 \big] = \prod_{\ell \neq \delta} \E \left[ \exp \left( - \int_\X K_1(y) \, \tilde \mu^+_\ell(dy) \right) \bigm\vert \crm_0 \right] \prod_{j=1}^k \prod_{h=1}^{r_{\ell j}} \E \left[ \exp \left( - K_1(X^*_j) \, S_{\ell jh} \right) \right] \\[8pt]
		\times \Bigg( \int_\X k(t;x) \, \E \left[ \crm^+_\delta(dx) \, \exp \left( - \int_\X K_1(y) \, \tilde \mu^+_\delta(dy) \right) \bigm\vert \crm_0 \right] \prod_{j=1}^k \prod_{h=1}^{r_{\delta j}} \E \left[ \exp \left( - K_1(X^*_j) \, S_{\delta jh} \right) \right] \\[8pt]
		+ \sum_{\ell=1}^k k(t;X^*_\ell) \sum_{\varepsilon=1}^{r_{\delta \ell}} \E \left[ \exp \left( - \int_\X K_1(y) \, \tilde \mu^+_\delta(dy) \right) \bigm\vert \crm_0 \right] \E \left[ S_{\delta \ell \varepsilon} \exp \left( - K_1(X^*_\ell) \, S_{\delta \ell \varepsilon} \right) \right] \\[8pt]
		\times \prod_{j \neq \ell} \prod_{h \neq \varepsilon} \E \left[ \exp \left( - K_1(X^*_j) \, S_{\delta j h} \right) \right] \Bigg) dt.
	\end{multline*}
	The posterior characterization of the random measures $\crm_1, \dots, \crm_D$ in Theorem~\ref{prop:posterior}, implies that
	\begin{multline*}
		\E \big[ \tilde p(dt, \delta) \mid \mathcal{F}_n, \crm_0 \big] = \exp \left( - (D-1) \int_\X \psi^+_y(K_1(y)) \, \tilde \mu_0(dy) \right) \prod_{\ell \neq \delta} \prod_{j=1}^k \prod_{h=1}^{r_{\ell j}} \frac{\tau^+_{X^*_j}(q_{\ell jh}; K_1(X^*_j))}{\tau_{X^*_j}^+(q_{\ell jh}; 0)} \\[8pt]
		\times \Bigg( \int_\X k(t;x) \, \tau^+_x(1; K_1(x)) \, \crm_0(dx) \, \exp \left( - \int_\X \psi^+_y(K_1(y)) \, \tilde \mu_0(dy) \right) \prod_{j=1}^k \prod_{h=1}^{r_{\delta j}} \frac{\tau^+_{X^*_j}(q_{\delta jh}; K_1(X^*_j))}{\tau^+_{X^*_j}(q_{\delta jh}; 0)} \\[8pt]
		+ \sum_{\ell=1}^k k(t;X^*_\ell) \sum_{\varepsilon=1}^{r_{\delta \ell}} \exp \left( - \int_\X \psi^+_y(K_1(y)) \, \tilde \mu_0(dy) \right) \frac{\tau^+_{X^*_\ell}(q_{\delta \ell \varepsilon} + 1; K_1(X^*_\ell))}{\tau^+_{X^*_\ell}(q_{\delta \ell \varepsilon}; 0)} \\[8pt]
		\times \prod_{j \neq \ell} \prod_{h \neq \varepsilon} \frac{\tau^+_{X^*_j}(q_{\delta jh}; K_1(X^*_j))}{\tau^+_{X^*_j}(q_{\delta jh}; 0)} \Bigg) dt.
	\end{multline*}
	Rearranging the terms in the expression above, we obtain
	\begin{multline*}
		\E \big[ \tilde p(dt, \delta) \mid \mathcal{F}_n, \crm_0 \big] = \exp \left( - \int_\X D \psi^+_y(K_1(y)) \, \tilde \mu_0(dy) \right) \prod_{\ell=1}^D \prod_{j=1}^k \prod_{h=1}^{r_{\ell j}} \frac{\tau^+_{X^*_j}(q_{\ell jh}; K_1(X^*_j))}{\tau^+_{X^*_j}(q_{\ell jh}; 0)} \\[8pt]
		\times \left( \int_\X k(t;x) \, \tau^+_x(1; K_1(x)) \, \crm_0(dx) + \sum_{j=1}^k k(t;X^*_j) \sum_{h=1}^{r_{\delta j}} \frac{\tau^+_{X^*_j}(q_{\delta j h} + 1; K_1(X^*_j))}{\tau^+_{X^*_j}(q_{\delta j h}; K_1(X^*_j))} \right) dt.
	\end{multline*}
	The posterior estimator is then derived by inserting the expression of the posterior distribution of the root random measure $\crm_0$,
	\begin{multline*}
		\E \big[ \tilde p(dt, \delta) \mid \mathcal{F}_n \big] = \E \Bigg[ \exp \left( - \int_\X D \psi^+_y(K_1(y)) \, \tilde \mu^+_0(dy) \right) \prod_{j=1}^k \exp \left( - D \psi^+_{X^*_j}(K_1(X^*_j)) \, V_j \right) \\[8pt]
		\times \prod_{\ell=1}^D \prod_{j=1}^k \prod_{h=1}^{r_{\ell j}} \frac{\tau^+_{X^*_j}(q_{\ell jh}; K_1(X^*_j))}{\tau^+_{X^*_j}(q_{\ell jh}; 0)} \ \Bigg( \int_\X k(t;x) \, \tau^+_x(1; K_1(x)) \, \crm^+_0(dx) \\[8pt]
		+ \sum_{j=1}^k k(t;X^*_j) \, \tau^+_{X^*_j}(1; K_1(X^*_j)) \, V_j + \sum_{j=1}^k k(t;X^*_j) \sum_{h=1}^{r_{\delta j}} \frac{\tau^+_{X^*_j}(q_{\delta j h} + 1; K_1(X^*_j))}{\tau^+_{X^*_j}(q_{\delta j h}; K_1(X^*_j))} \Bigg) \Bigg] dt.
	\end{multline*}
	By rearranging the terms and invoking the mutual independence of $\crm^+_0$ and $V_1, \dots, V_k$, one obtains a sum of three components:
	\begin{equation*}
		\E \big[ \tilde p(dt, \delta) \mid \mathcal{F}_n \big] = \prod_{\ell=1}^D \prod_{j=1}^k \prod_{h=1}^{r_{\ell j}} \frac{\tau^+_{X^*_j}(q_{\ell jh}; K_1(X^*_j))}{\tau^+_{X^*_j}(q_{\ell jh}; 0)} \ \big( I_1(t) + I_2(t) + I_3(t,\delta) \big) \,dt,
	\end{equation*}
	where the quantities $I_1(t), I_2(t)$ and $I_3(t,\delta)$ are given by
	\begin{align*}
		I_1(t) & = \int_\X k(t;x) \, \tau^+_x(1; K_1(x)) \, \E \left[ \crm^+_0(dx) \, \exp \left( - \int_\X D \psi^+_y(K_1(y)) \, \tilde \mu^+_0(dy) \right) \right] \\ 
		& \qquad \qquad \times \prod_{j=1}^k \E \left[ \exp \left( - D \psi^+_{X^*_j}(K_1(X^*_j)) \, V_j \right) \right],\\[8pt]
		I_2(t) & = \sum_{\ell=1}^k k(t;X^*_\ell) \, \tau^+_{X^*_\ell}(1; K_1(X^*_\ell)) \, \E \left[ \exp \left( - \int_\X D \psi^+_y(K_1(y)) \, \tilde \mu^+_0(dy) \right) \right] \\
		& \qquad \qquad \times \E \left[ V_\ell \, \exp \left( - D \psi^+_{X^*_\ell}(K_1(X^*_\ell)) \, V_\ell \right) \right] \ \prod_{j \neq \ell} \E \left[ \exp \left( - D \psi^+_{X^*_j}(K_1(X^*_j)) \, V_j \right) \right], \displaybreak \\[8pt]
		I_3(t,\delta) & = \sum_{\ell=1}^k k(t;X^*_\ell) \, \sum_{h=1}^{r_{\delta \ell}} \frac{\tau^+_{X^*_\ell}(q_{\delta \ell h} + 1; K_1(X^*_\ell))}{\tau^+_{X^*_\ell}(q_{\delta \ell h}; K_1(X^*_\ell))} \ \E \left[ \exp \left( - \int_\X D \psi^+_y(K_1(y)) \, \tilde \mu^+_0(dy) \right) \right] \\
		& \qquad \qquad \times \prod_{j=1}^k \E \left[ \exp \left( - D \psi^+_{X^*_j}(K_1(X^*_j)) \, V_j \right) \right].
	\end{align*}
	As for the first component, by the posterior characterization of $\crm_0$, and resorting to the identity \eqref{key_identity} with $f(y) = D \psi^+_y(K_1(y))$, we get
	\begin{multline*}
		I_1(t) = \int_\X k(t;x) \, \tau^+_x(1; K_1(x)) \, \tau_{0,x}^+\big(1; D \psi^+_x(K_1(x)) \big) \, \Lambda_0(dx) \\[8pt]
		\times \exp \left( - \int_\X \psi^+_{0,y} \big( D\psi^+_y(K_1(y)) \big) \, \Lambda_0(dy) \right) \prod_{j=1}^k \frac{\tau_{0,X^*_j}^+\big(r_j; D \psi^+_{X^*_j}(K_1(X_j^*)) \big)}{\tau_{0,X^*_j}^+(r_j; 0)},
	\end{multline*}
	while for the second and third component, the expectation with respect to $\crm_0$ and $V_1, \dots, V_k$ reads
	\begin{align*}
		I_2(t) & = \sum_{\ell=1}^k k(t;X^*_\ell) \, \tau^+_{X^*_\ell}(1; K_1(X^*_\ell)) \, \exp \left( - \int_\X \psi_{0,y}^+ \big( D \psi^+_y(K_1(y)) \big) \, \Lambda_0(dy) \right) \\[8pt]
		& \qquad \qquad \times \frac{\tau_{0,X^*_\ell}^+\big(r_\ell + 1; D \psi^+_{X^*_\ell}(K_1(X_\ell^*)) \big)}{\tau_{0,X^*_\ell}^+(r_\ell; 0)} \, \prod_{j \neq \ell} \frac{\tau_{0,X^*_j}^+\big(r_j; D \psi^+_{X^*_j}(K_1(X_j^*))\big)}{\tau_{0,X^*_j}^+(r_j; 0)}, \\[10pt]
		I_3(t,\delta) & = \sum_{\ell=1}^k k(t;X^*_\ell) \, \sum_{h=1}^{r_{\delta \ell}} \, \frac{\tau^+_{X^*_\ell}(q_{\delta \ell h} + 1; K_1(X^*_\ell))}{\tau^+_{X^*_\ell}(q_{\delta \ell h}; K_1(X^*_\ell))} \, \exp \left( - \int_\X \psi_{0,y}^+ \big( D \psi^+_y(K_1(y))\big) \, \Lambda_0(dy) \right) \\[8pt]
		& \qquad \qquad \times \prod_{j=1}^k \frac{\tau_{0,X^*_j}^+\big(r_j; D \psi^+_{X^*_j}(K_1(X_j^*)) \big)}{\tau_{0,X^*_j}^+(r_j; 0)}.
	\end{align*}
	In conclusion, by collecting and rearranging the above terms, we finally obtain
	\begin{multline} \label{survival_posterior}
		\E \big[ \tilde p(dt, \delta) \mid \mathcal{F}_n \big] = \exp \left( - \int_\X \psi_{0,y}^+ \big( D \psi^+_y(K_1(y)) \big) \, \Lambda_0(dy) \right) \\[8pt] 
		\qquad \qquad \times \prod_{j=1}^k \frac{\tau_{0,X^*_j}^+\big(r_j; D \psi^+_{X^*_j}(K_1(X_j^*))\big)}{\tau_{0,X^*_j}^+(r_j; 0)} \, \prod_{\ell=1}^D \prod_{h=1}^{r_{\ell j}} \frac{\tau^+_{X^*_j}(q_{\ell jh}; K_1(X^*_j))}{\tau^+_{X^*_j}(q_{\ell jh}; 0)} \\[8pt]
		\times \Bigg( \int_\X k(t;x) \, \tau^+_x(1; K_1(x)) \,\tau_{0,x}^+\big(1; D\psi^+_x(K_1(x))\big) \, \Lambda_0(dx) \qquad \qquad \qquad \qquad \displaybreak \\[8pt]
		+ \sum_{j=1}^k k(t;X_j^*) \, \tau^+_{X^*_j}(1; K_1(X_j^*)) \, \frac{\tau_{0,X^*_j}^+\big( r_j+1; D \psi^+_{X^*_j}(K_1(X_j^*)) \big)}{\tau_{0,X^*_j}^+ \big( r_j; D \psi^+_{X^*_j}(K_1(X_j^*)) \big)} \\[8pt]
		+ \sum_{j=1}^k k(t;X_j^*) \sum_{h=1}^{r_{\delta j}} \frac{\tau^+_{X_j^*}(q_{\delta jh}+1; K_1(X_j^*))}{\tau^+_{X_j^*}(q_{\delta jh}; K_1(X_j^*))} \Bigg) \, dt.
	\end{multline}
	Note that the factors in the first two lines of this expression coincide with the posterior estimate of the survival function obtained in Proposition~\ref{prop:posterior_survival}.
	\hfill \qed
	
	\paragraph{Alternative proof of Theorem~\ref{prop:predict_competing} based on Proposition~\ref{prop:posterior_incidence} (Remark~\ref{remark})}
	
	The posterior estimate of $\tilde p(dt,\delta)$ detailed in Proposition~\ref{prop:posterior_incidence}, and more explicitly in \eqref{survival_posterior}, coincides with the predictive distribution for observation $(T_{n+1}, \Delta_{n+1})$ described in \eqref{survival_predictive}. This fact can be shown by clarifying the relationships of the Laplace exponents and cumulants for the random measures $\crm^+_1, \dots, \crm^+_D$, defined in \eqref{eq:cumulants_posterior} with the corresponding quantities for $\crm_1, \dots, \crm_D$, defined in \eqref{cumulants}. At the bottom of the hierarchy, the Laplace exponents are related as
	\begin{equation*}
		\psi_x^+(u) = \psi(u+K_n(x)) - \psi(K_n(x)),
	\end{equation*}
	or equivalently, $\psi_x^+(u) + \psi(K_n(x)) = \psi(u+K_n(x))$, while at the root of the hierarchy
	\begin{equation*}
		\psi_{0,x}^+(u) = \psi_0(u+D \psi(K_n(x))) - \psi_0(D \psi(K_n(x))).
	\end{equation*}
	Therefore, exploiting the expressions above and the fact that the updated kernel term is $\displaystyle K_{n+1}(x) = K_n(x) + \int_0^t k(s; x) \,ds = K_n(x) + K_1(x)$, the exponential term in \eqref{survival_posterior} becomes
	\begin{multline*}
		\exp \left( - \int_\X \psi_{0,y}^+ \big( D \psi_y^+(K_1(y))\big) \, \Lambda_0(dy) \right) \\
		= \exp \left( - \int_\X \left( \psi_0(D \psi(K_{n+1}(y))) - \psi_0(D \psi(K_n(y))) \right) \Lambda_0(dy) \right).
	\end{multline*}
	On the other hand, for the cumulants we have
	\begin{equation*}
		\tau_x^+(m; u) = \tau(m; u + K_n(x)), \qquad
		\tau_{0,x}^+(m; u) = \tau_0(m; u + D \psi(K_n(x))).
	\end{equation*}
	The quantities appearing in \eqref{survival_posterior} can be thus rewritten as
	\begin{align*}
		\tau_{0,x}^+\big(m; D\psi_x^+(K_1(x)) \big) = \tau_0\big(m; D \psi_x^+(K_1(x)) + D\psi(K_n(x)) \big) = \tau_0\big(m; D \psi(K_{n+1}(x)) \big),
	\end{align*}
	while $ \tau_{0,x}^+(m; 0) = \tau_0(m; D \psi(K_n(x)))$, and moreover 
	\begin{equation*}
		\tau_x^+(m; K_1(x)) = \tau(m; K_{n+1}(x)), \qquad \tau_x^+(m; 0) = \tau(m; K_n(x))
	\end{equation*}
	In conclusion, the expression in \eqref{survival_predictive} can be obtained from \eqref{survival_posterior} using the identities above. The result in Theorem~\ref{prop:predict_competing} is then retrieved conditioning on the survival time $T_{n+1} = t$.
	\hfill \qed
	
	\clearpage
	\section{Algorithms for generalized gamma mixtures}
	\label{supsec:sampling}
	
	In the following, we specialize the sampling algorithms described in Section~\ref{sec:gibbs} to the subclass of \emph{generalized gamma} hCRMs \eqref{generalized_gamma}, with the kernel choices considered in Section~\ref{sec:mixture_hazard}.
	
	\subsection{Marginal sampling algorithm (Section~\ref{sec:marginalMCMC})} 
	\label{ex:gibbs_marginal}
	
	The marginal Gibbs sampling scheme for the reallocation of observation $i \in \set{1,\ldots,n}$ to the latent partition structure simplifies to: 
	\begin{itemize}[noitemsep]
		\item[(1)] for each latent location $X_j^*$ such that $X_j^* \le T_i$,
		\begin{equation*}
			\Prob \big( X_{i}= X_j^*, \, Z_{i} = Z_{\delta jh}^* \mid \mathcal{F}_n^{-i} \big) \propto k(T_i; X_j^*) \, \frac{q_{\delta jh}^{-i} - \sigma}{\beta + K_n(X_j^*)}, \qquad h = 1, \dots, r_{\delta j};
		\end{equation*}
		\item[(2)] for each latent location $X_j^*$ such that $X_j^* \le T_i$,
		\begin{multline*}
			\Prob \big( X_{i} = X_j^*, \, Z_{i} = \text{`new'} \mid \mathcal{F}_n^{-i} \big) \\
			\propto k(T_{i}; X_j^*) \, (\beta + K_n(X_j^*))^{\sigma-1} \ \frac{r_j^{-i} - \sigma_0}{\beta_0 + \frac{D}{\sigma} \big( (\beta + K_n(X^*_j))^\sigma - \beta^\sigma \big)};
		\end{multline*}
		\item[(3)] both $X_i$ and $Z_i$ take on new values with probability
		\begin{multline*}
			\Prob \big( X_{i} = \text{`new'}, \, Z_{i} = \text{`new'} \mid \mathcal{F}_n^{-i} \big) \\
			\propto \int_0^{T_i} k(T_i;x)\,(\beta + K_n(x))^{\sigma-1} \, \left( \beta_0 + \frac{D}{\sigma} \big( (\beta + K_n(x))^\sigma - \beta^\sigma \big) \right)^{\sigma_0-1} \Lambda_0(dx),
		\end{multline*}
		which depends on $\mathcal{F}_n^{-i}$ only through the observed survival times $\boldsymbol T_{1:n}$; in this case, the new value assigned to $X_i$ is sampled between $0$ and $T_i$, proportionally to
		\begin{equation*}
			k(T_i;x)\, (\beta + K_n(x))^{\sigma-1} \, \left( \beta_0 + \frac{D}{\sigma} \big( (\beta + K_n(x))^\sigma - \beta^\sigma \big) \right)^{\sigma_0-1} \Lambda_0(dx).
		\end{equation*}
	\end{itemize}
	
	\noindent Notably, cases (1) and (2) are in closed form, while the integral in case (3) can be evaluated via quadrature formulae over bounded intervals. Moreover, the potential new value of $X_i$ can be sampled via rejection sampling. As for the acceleration step, each latent location $X_j^*$ for $j \in \set{1,\ldots,k}$ is independently resampled between $0$ and the minimum observed survival time associated with it, $\min_{i \in C_j^X} T_i$, proportionally to \cut
	\begin{equation*}
		\bigg( \prod_{i \in C_j^X} k(T_i; x) \bigg) (\beta + K_n(x))^{r_j \sigma-n_j} \left(\beta + \frac{D}{\sigma} \big( (\beta + K_n(x))^\sigma - \beta^\sigma \big) \right)^{\sigma_0-r_j} \Lambda_0(dx);
	\end{equation*}
	this task is routinely performed by means of Metropolis-Hastings steps.
	Finally, a straightforward initialization strategy consists in populating an empty partition structure using the same sequential allocation scheme outlined above.
	
	\subsection{Posterior sampling of random measures (Section~\ref{sec:conditionalMCMC})} \label{ex:gibbs_conditional}
	
	For the implementation of this posterior sampling algorithm, a pivotal role in generating the jumps' sequences is played by the upper incomplete gamma function $\displaystyle \Gamma(a, x) := \int_x^\infty s^{a-1} \, e^{-s}\, ds$,
	for any $x \ge 0$ and $a \in (-1,0]$. At the root of the hierarchy, an approximate realization of $\crm_0^+$ in \eqref{approx_root} is obtained as follows:
	\begin{itemize}[noitemsep]
		\item[1$_0$)] the jumps $V_1, \dots, V_k$ at fixed locations $X^*_1, \dots, X^*_k$ are sampled from \begin{equation*}
			V_j \stackrel{\mbox{\scriptsize ind}}{\sim} \text{Gamma} \left( r_j - \sigma_0, \, \beta_0 + \frac{D}{\sigma} \Big( (\beta + K_n(X^*_j))^\sigma - \beta^\sigma \Big) \right), \qquad j = 1, \dots, k;
		\end{equation*}
		\item[2$_0$)] the random jumps $V^{(0)}_h$, for $h \ge 1$, are sampled by solving the equations
		\begin{multline}
			\label{eq:jumps_mu0}
			\Gamma \left( - \sigma_0, \, V^{(0)}_h \left( \beta_0 + \frac{D}{\sigma} \Big( \big(\beta + K_n\big(X^{(0)}_h\big)\big)^\sigma - \beta^\sigma \Big) \right) \right) \\
			= \frac{N^{(0)}_h \, \Gamma(1-\sigma_0)}{\Lambda_0(\X) } \left( \beta_0 + \frac{D}{\sigma} \Big( \big(\beta + K_n\big(X^{(0)}_h\big)\big)^\sigma - \beta^\sigma \Big) \right)^{-\sigma_0}.
		\end{multline}
	\end{itemize}
	
	\noindent The truncation level $H_0 \ge 1$ is the first value of $h$ for which
	\begin{equation*}
		N^{(0)}_h > \Lambda_0(\X) \, \frac{\Gamma( - \sigma_0, \, \beta_0 \varepsilon)}{\Gamma(1-\sigma_0)} \, \beta_0^{\sigma_0},
	\end{equation*}		
	if $\beta_0>0$, while if $\beta_0=0$, the inequality simplifies to $N^{(0)}_h > \Lambda_0(\X) \, \varepsilon^{-\sigma_0} / (\sigma_0 \, \Gamma(1-\sigma_0))$.
	Similarly, at the bottom of the hierarchy, the approximate realization of $\crm_\delta^+$ in \eqref{approx_lower} with $\delta \in \set{1,\ldots,D}$ specializes to:
	\begin{itemize}[noitemsep]
		\item[1$_\delta$)] the jumps $(S_{\delta jh})_{jh}$ at fixed locations $X^*_1, \dots, X^*_k$ are sampled in groups as
		\begin{equation*}
			\sum_{h=1}^{r_{\delta j}} S_{	\delta jh} \sim \text{Gamma} \big(n_{\delta j} - r_{\delta j} \sigma, \, \beta + K_n(X_j^*) \big), \qquad j = 1, \dots, k;
		\end{equation*}
		\item[2$_\delta$)] the random jumps $S^{(\delta)}_h$, for $h \ge 1$, are sampled by solving the equations
		\begin{equation}
			\label{eq:jumps_mud}
			\Gamma \left( - \sigma, \, S^{(\delta)}_h \left( \beta + K_n\big(X^{(\delta)}_h\big) \right) \right) = \frac{N^{(\delta)}_h \, \Gamma(1-\sigma)}{
				\crm_0^+(\X) \left( \beta + K_n\big(X^{(\delta)}_h\big) \right)^\sigma}, 
		\end{equation}
		where $\displaystyle \crm_0^+ = \sum_{j=1}^k V_j + \sum_{h=1}^{H_0} V^{(0)}_h$ is the approximated total mass of $\crm_0^+$.
	\end{itemize}
	Also in this case, the truncation level $H_\delta$ is the first value of $h$ for which
	\begin{equation*}
		N^{(\delta)}_h > \frac{\Gamma ( - \sigma, \, \beta \varepsilon)}{\Gamma(1-\sigma)} \, \beta^\sigma \crm_0^+(\X), \quad \mbox{if }\beta>0, \qquad
		N^{(\delta)}_h > \frac{\crm_0^+(\X) \, \varepsilon^{-\sigma}}{\sigma \, \Gamma(1-\sigma)}, \quad \mbox{if } \beta = 0.
	\end{equation*}
	A potential bottleneck for the concrete implementation of this algorithm is the solution of the equations \eqref{eq:jumps_mu0}-\eqref{eq:jumps_mud}. A simple and effective approach consists in solving these with respect to the logarithm of the jumps, through the Newton-Raphson method. Indeed, the involved functions are strictly convex and their derivatives can be computed in closed form. 
	
	\clearpage
	\section{Additional material on the applications in Section~\ref{sec:simulation_study}}
	\label{supsec:applications}
	
	We provide further details, diagnostic plots, and additional figures on the illustrations with synthetic and real datasets presented in Section~\ref{sec:simulation_study}, which were omitted from the main manuscript due to space constraints.
	
	\subsection{A simulation study with three competing risks}
	\label{supsec:illustration}
	
	The simulation scenario in Section~\ref{sec:simulation_three_risks} considers $D = 3$ competing sources of risk: synthetic data are generated according to the latent failure times approach, whose distributions are detailed in \eqref{data_generation}. The sample size is $n = 300$ and the number of subjects in the simulated dataset experiencing each competing event is $120$, $99$ and $81$, respectively. Figure~\ref{fig:histogram} displays the histogram of the observed survival times, overlapped with the true data--generating density function. The model hyperparameters are set at $\beta = \beta_0 = 1.0$ and $\sigma = \sigma_0 = 0.25$, while $\Lambda_0 = \theta\,dx$ for $\theta > 0$. Moreover, non-informative exponential hyperpriors are specified for the concentration parameter $\theta$ and kernel parameter $\gamma$, namely $\theta \sim \text{Exp}(0.1)$ and $\gamma \sim \text{Exp}(0.1)$.
	
	The MCMC procedure is run for $25,000$ iterations, with a burn-in of $5,000$. Posterior samples are summarized for the number $k$ of shared latent locations in the hierarchical model (Figure~\ref{fig:dishes}), the concentration parameter $\theta$ (Figure~\ref{fig:theta}), and the kernel parameter $\gamma$ (Figure~\ref{fig:gamma}). The burn-in samples are highlighted in gray in traceplots, and discarded when constructing histograms. After thinning, $2,000$ samples from the latent partition structure are collected and employed to compute posterior estimates with the marginal method. Figure~\ref{supfig:survival} displays the traceplot (after burn-in and thinning) and corresponding autocorrelation function for the posterior estimates of the survival function, given the latent variables (Proposition~\ref{prop:posterior_survival}), at time $t = 0.4$. The same plots are reported in Figure~\ref{supfig:incidence} for the posterior estimates of incidence functions (Proposition~\ref{prop:posterior_incidence}), again at time $t = 0.4$.
	
	\begin{figure}
		\centering
		\includegraphics[width=0.6\textwidth]{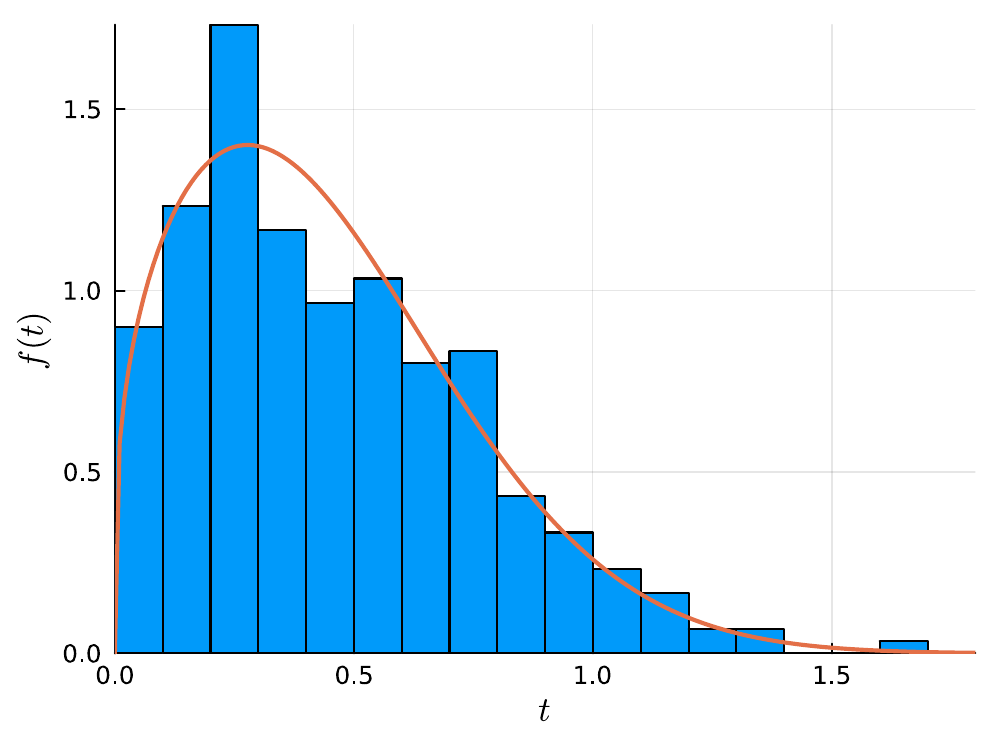}
		\captionsetup{width=0.85\textwidth,font=small}
		\caption{Histogram of the observed survival times for the synthetic dataset with $n = 300$ data points in Section~\ref{sec:simulation_three_risks}, and corresponding data-generating density function; the maximum observed survival time is $1.61$.}
		\label{fig:histogram}
	\end{figure}
	
	\begin{figure}
		\centering
		\noindent\makebox[\textwidth]{\includegraphics[width=0.5\textwidth]{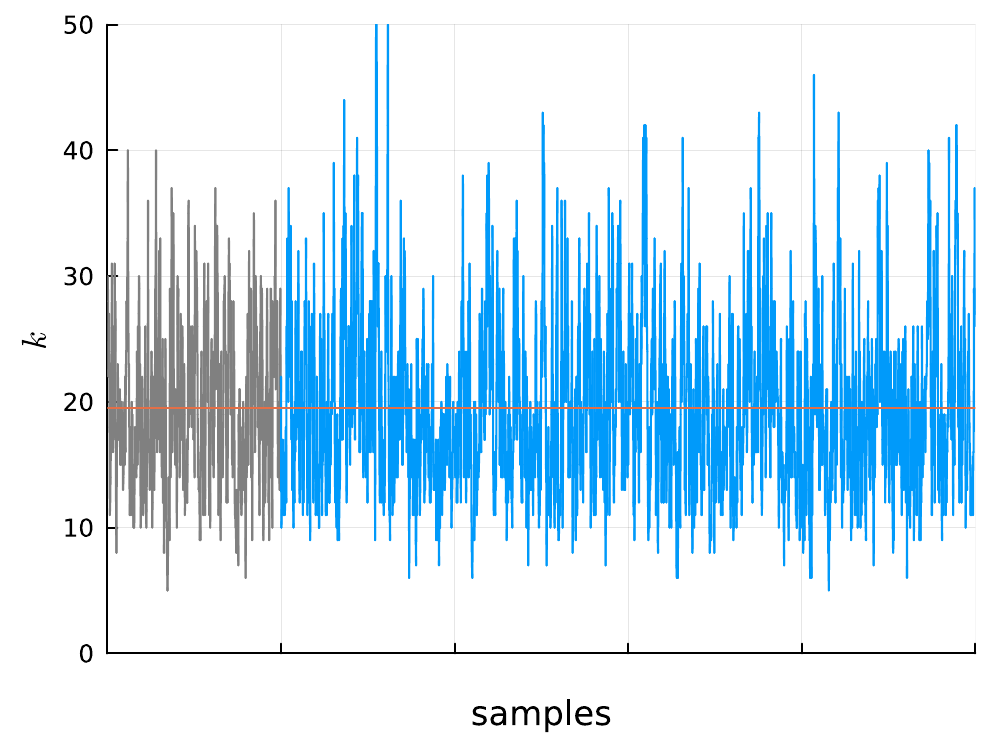}
			\includegraphics[width=0.5\textwidth]{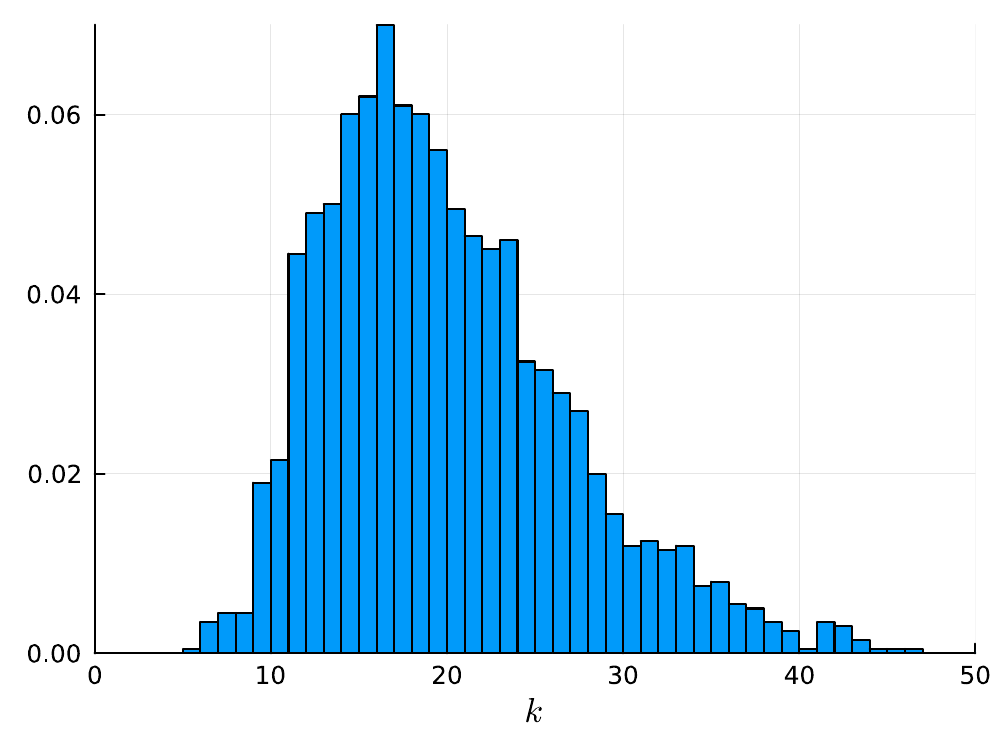}
		}
		\captionsetup{font=small}
		\caption{Posterior samples of the number $k$ of shared latent locations for the synthetic dataset in Section~\ref{sec:simulation_three_risks}; the posterior mean is $19.51$.}
		\label{fig:dishes}
	\end{figure}
	
	\begin{figure}
		\centering
		\noindent\makebox[\textwidth]{\includegraphics[width=0.5\textwidth]{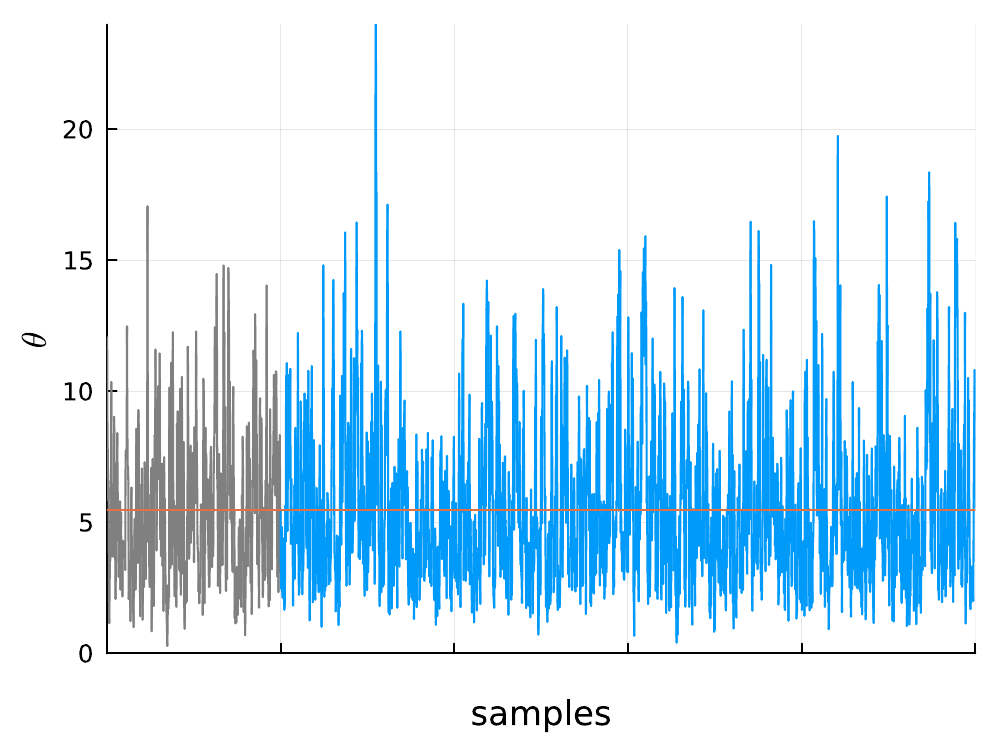}
			\includegraphics[width=0.5\textwidth]{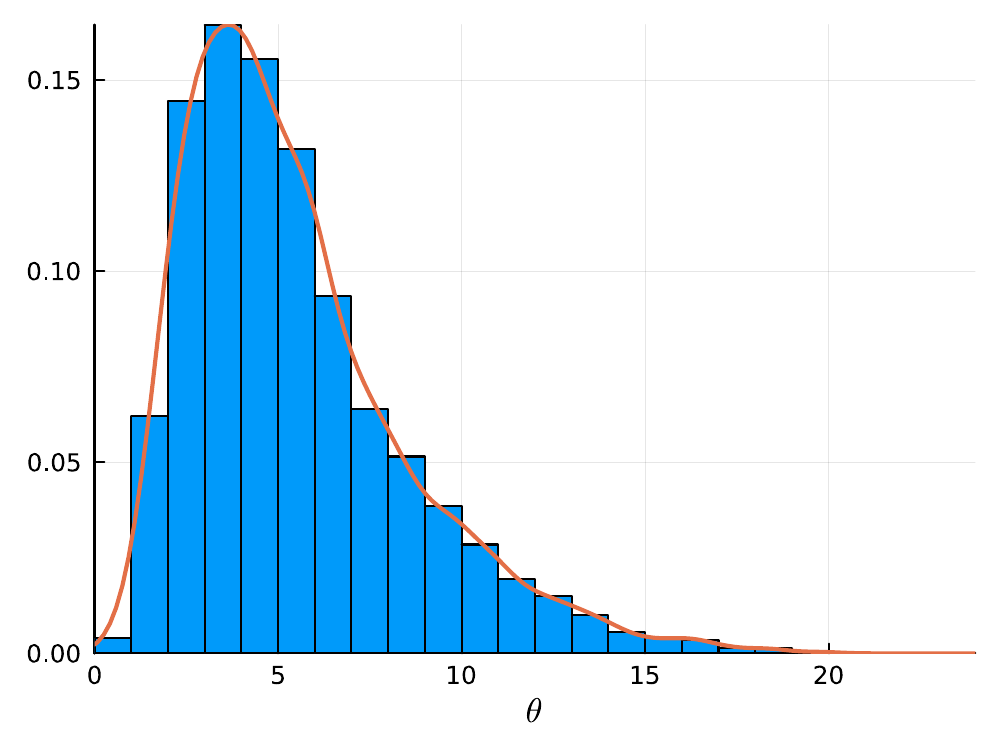}
		}
		\captionsetup{font=small}
		\caption{Posterior samples of the concentration parameter $\theta$ for the synthetic dataset in Section~\ref{sec:simulation_three_risks}; the posterior mean is $5.47$.}
		\label{fig:theta}
	\end{figure}
	
	\begin{figure}
		\centering
		\noindent\makebox[\textwidth]{\includegraphics[width=0.5\textwidth]{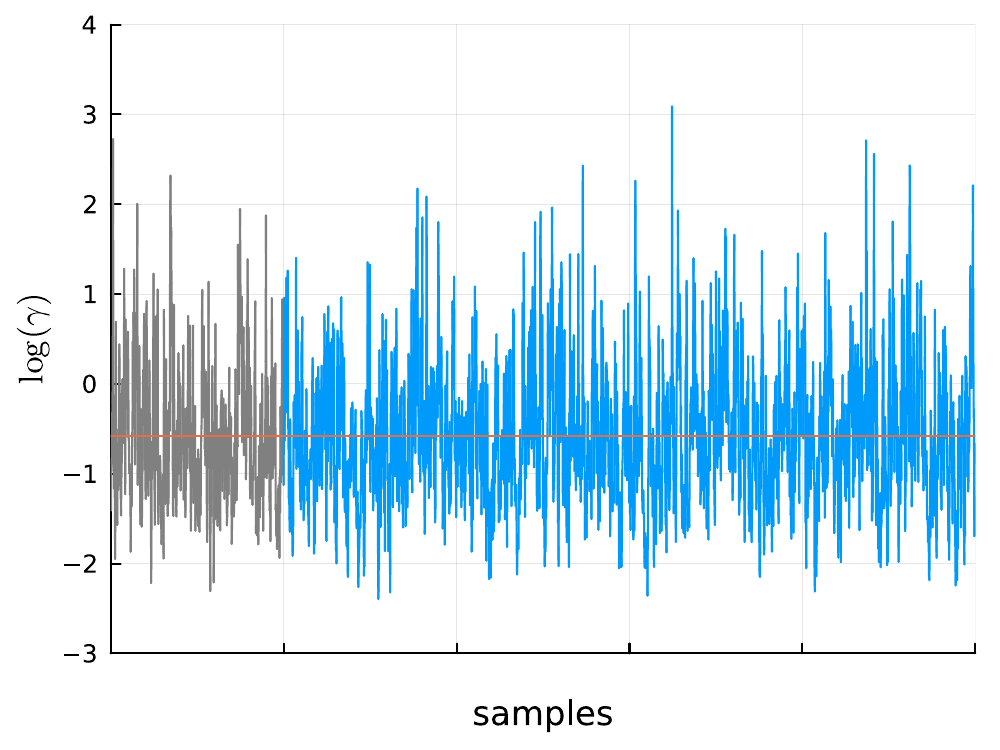}
			\includegraphics[width=0.5\textwidth]{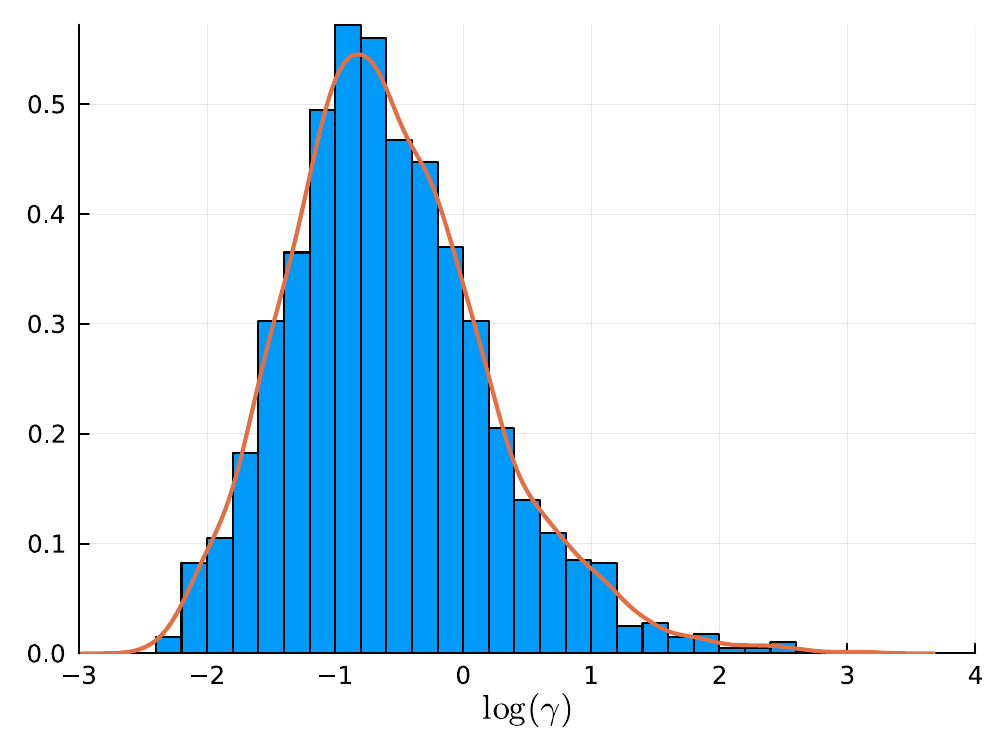}
		}
		\captionsetup{font=small}
		\caption{Posterior samples of the Dykstra-Laud kernel parameter $\gamma$ for the synthetic dataset in Section~\ref{sec:simulation_three_risks}, displayed on the log-scale; the posterior mean is $0.815$, the acceptance rate of the Metropolis-Hastings step is $0.414$.}
		\label{fig:gamma}
	\end{figure}
	
	\begin{figure}
		\centering
		\noindent\makebox[\textwidth]{\includegraphics[width=0.5\textwidth]{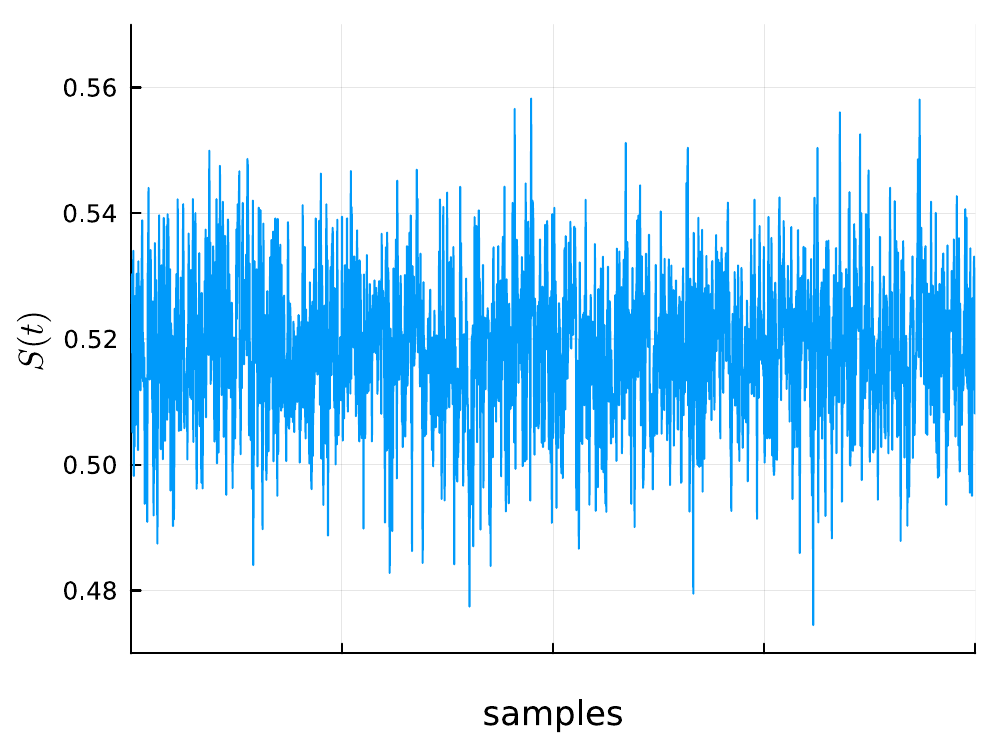}
			\includegraphics[width=0.5\textwidth]{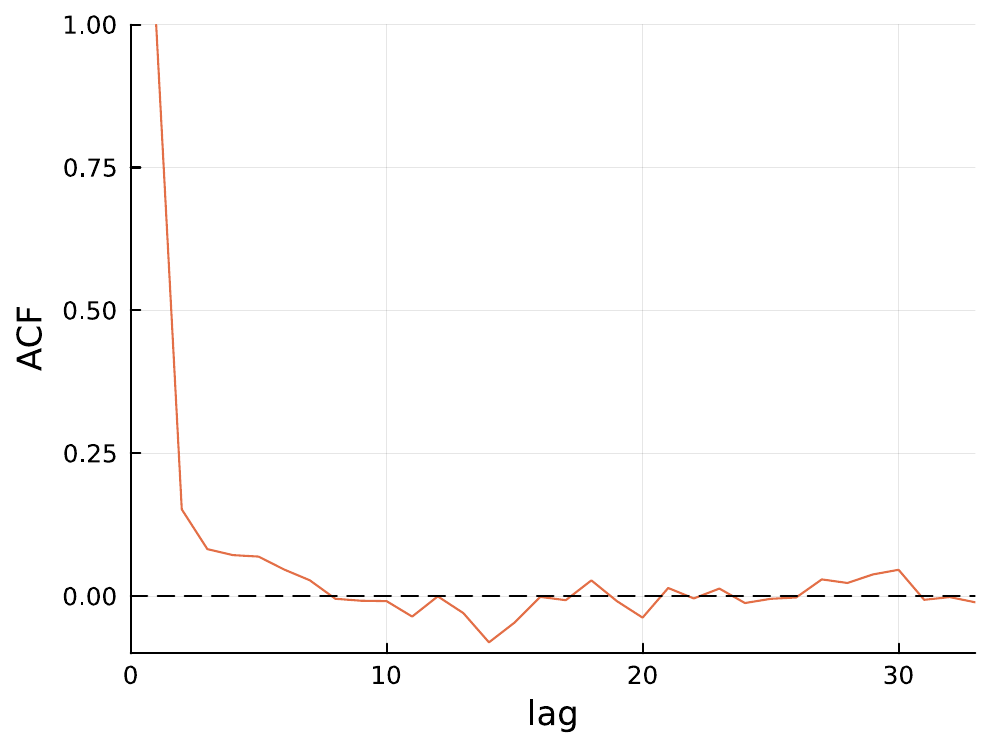}
		}
		\captionsetup{font=small}
		\caption{Posterior estimates of the survival function at time $t = 0.4$, given the latent variables, for the synthetic dataset in Section~\ref{sec:simulation_three_risks}; the effective sample size is $1069$.}
		\label{supfig:survival}
	\end{figure}
	
	\begin{figure}
		\centering
		\noindent\makebox[\textwidth]{\includegraphics[width=0.5\textwidth]{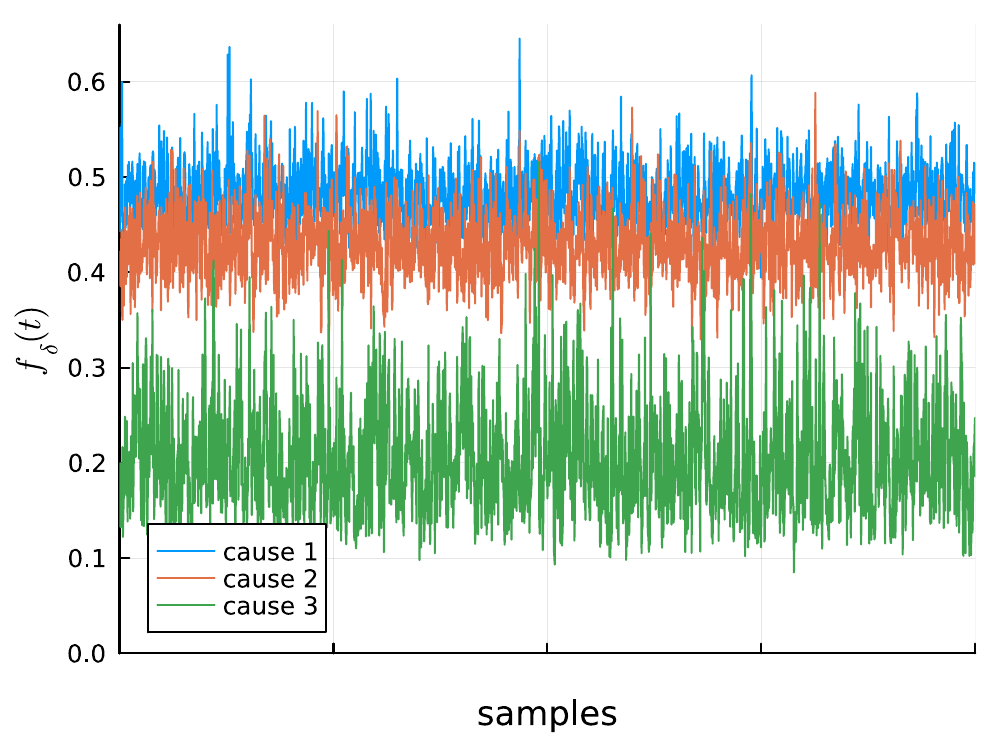}
			\includegraphics[width=0.5\textwidth]{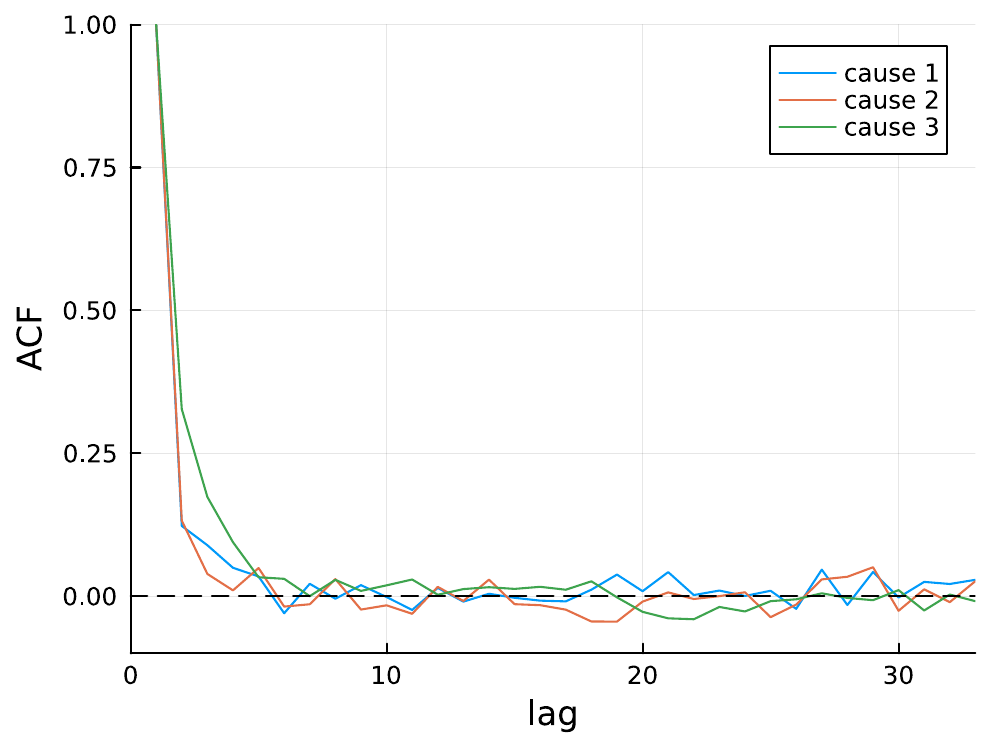}
		}
		\captionsetup{font=small}
		\caption{Posterior estimates of the incidence functions at time $t = 0.4$, given the latent variables, for the synthetic dataset in Section~\ref{sec:simulation_three_risks}; the effective sample sizes are $1267$, $1376$ and $722$, respectively.}
		\label{supfig:incidence}
	\end{figure}
	
	\subsection{Illustration of posterior consistency}
	\label{supsec:consistency}
	
	Theorem~\ref{prop:support} in Section~\ref{sec:consistency} investigates the Kullback-Leibler support and weak consistency of the survival density induced by the mixture hazard model \eqref{mixture_hazard} with hierarchical generalized gamma mixing measures, for the different kernel choices. In particular, the Dykstra-Laud kernel is weakly consistent at each density with non-decreasing hazard rate. This section provides an empirical confirmation of this consistency result by means of a simulation study. Consider the simulation setup in Section~\ref{sec:simulation_three_risks}, with $D = 3$ competing sources of risk, and generate synthetic data according to the latent failure times approach. Notably, the distributions of the latent failure times in \eqref{data_generation} are Weibull with shape parameters larger than $1$, implying a strictly increasing hazard rate for the survival density. The model hyperparameters and hyperpriors are specified as in Section~\ref{supsec:illustration} above. The sample sizes are $n = 30, 60, 120, 240, 480$, and results are averaged over $50$ simulated datasets for each value of $n$. The burn-in period and thinning parameter are scaled linearly with $n$ to ensure proper mixing of the MCMC procedure for each sample size. Figure~\ref{fig:consistency} displays the average errors in estimating the survival function for the proposed hierarchical model and the Kaplan--Meier estimate, in terms of Kolmogorov distances from the ground truth; the posterior estimate is obtained with the marginal approach. For both methods, the estimation errors exhibit a nearly linear relationship with the sample size on the log-log scale. Although far from a comprehensive simulation study, this illustration supports the theoretical results and further confirms the effectiveness of our proposal.
	
	\begin{figure}
		\centering
		\includegraphics[width=0.6\textwidth]{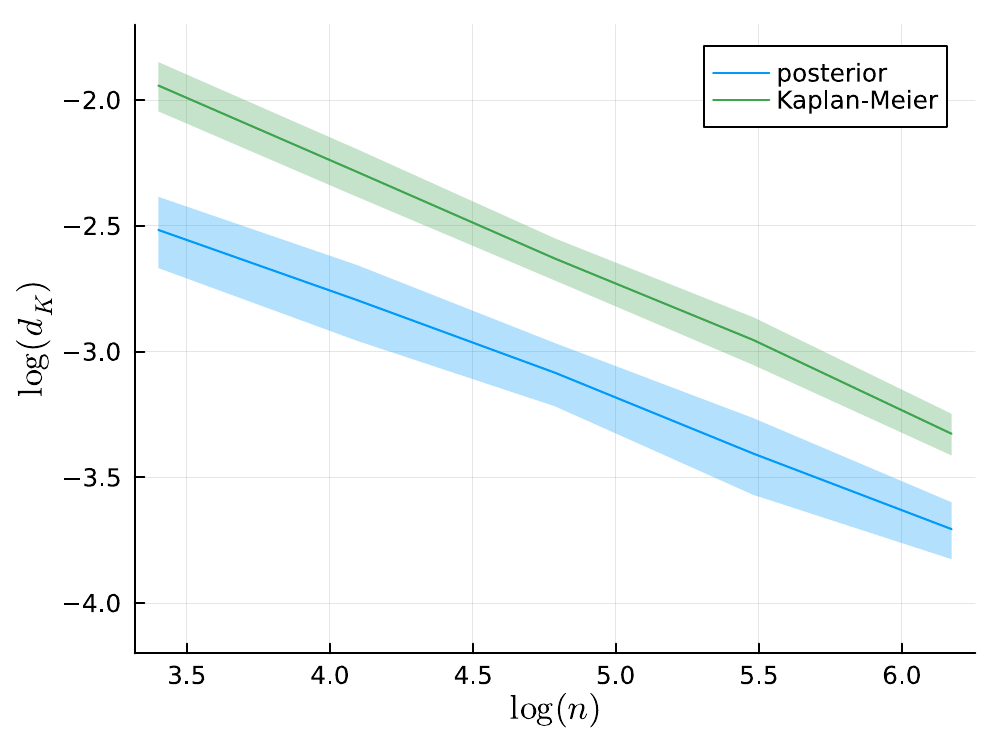}
		\captionsetup{width=0.85\textwidth,font=small}
		\caption{Average errors in estimating the survival function for the hierarchical model (posterior) and Kaplan--Meier estimate (green), in terms of Kolmogorov distances $d_K$ from the ground truth, for different sample sizes; pointwise asymptotic $0.95$ confidence bands for the mean errors are also displayed. Data are generated according to the latent failure times approach detailed in \eqref{data_generation}; results are averaged over $50$ simulated datasets for each sample size.}
		\label{fig:consistency}
	\end{figure}
	
	\subsection{Dependent vs.~independent priors}
	
	The simulation study in Section~\ref{sec:simulation_independent} compares the performances of the proposed model, with hierarchically dependent hazard rates, and of a Bayesian nonparametric model with independent hazard rates. The data--generating model involves $D = 3$ competing and independent sources of risk.
	For each of the $100$ simulated dataset, observations are generated according to the latent failure time approach. Specifically, for each event type $\delta \in \set{1,2,3}$, the potential survival time $Y_{i,\delta}$ is sampled from a mixture of a cause--specific and a common distribution, with equal weights. Denote by $t \mapsto \text{Weibull}(t\,; \xi)$ the density of a Weibull distribution with shape $\xi$ and unit scale parameters.
	The density of $Y_{i,\delta}$ is $g_\delta(t) = 0.5\, \phi_\delta(t) + 0.5 \, \phi_0(t)$,
	where $\phi_0(t)= 0.5 \, \text{Weibull}(t\,; 1.2) + 0.5 \, \text{Weibull}(t - 1\, ; 3.0)$,
	and $\phi_\delta(t) = \text{Weibull}(t; \xi_\delta)$, with shape parameters $\xi_1 = 1.5$, $\xi_2 = 2.0$ and $\xi_3 = 2.5$, respectively.
	The cause--specific hazard rates and incidence functions resulting from such data--generating model are displayed in Figure \ref{fig:model}. The hazard rate functions are essentially increasing, making the Dykstra--Laud kernel an appropriate choice in this scenario. Moreover, the data--generating distributions entails a balanced allocation of the observations to the event types: the average number of subjects experiencing each competing event is $34.11$, $33.21$ and $32.68$, respectively.
	
	\begin{figure}
		\centering
		\noindent\makebox[\textwidth]{\includegraphics[width=0.5\textwidth]{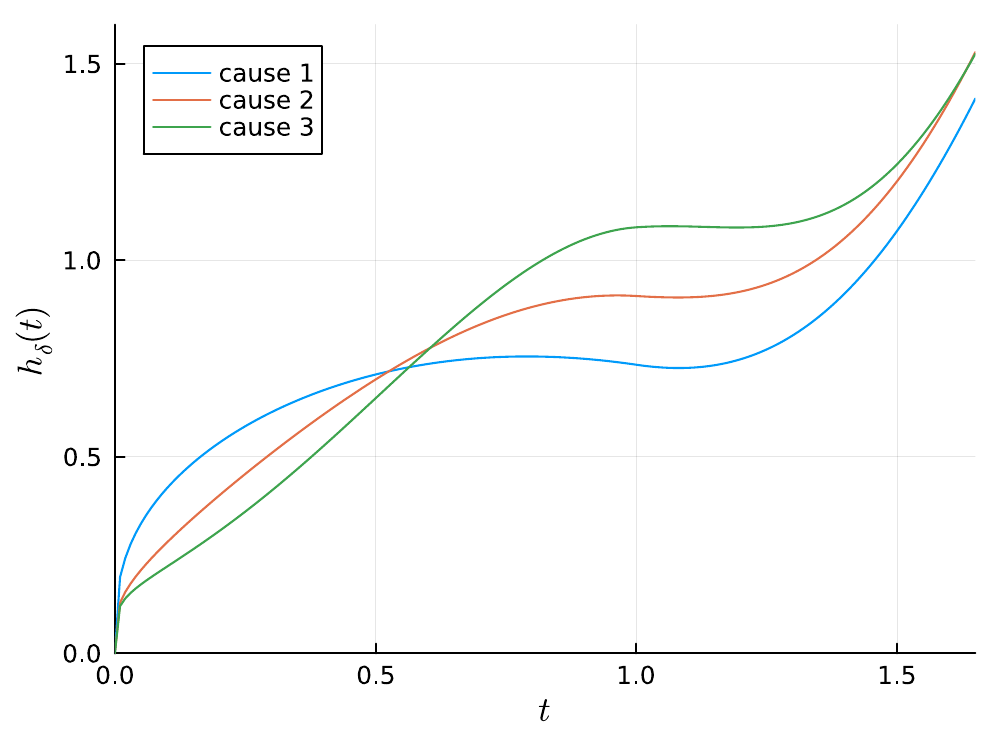}
			\includegraphics[width=0.5\textwidth]{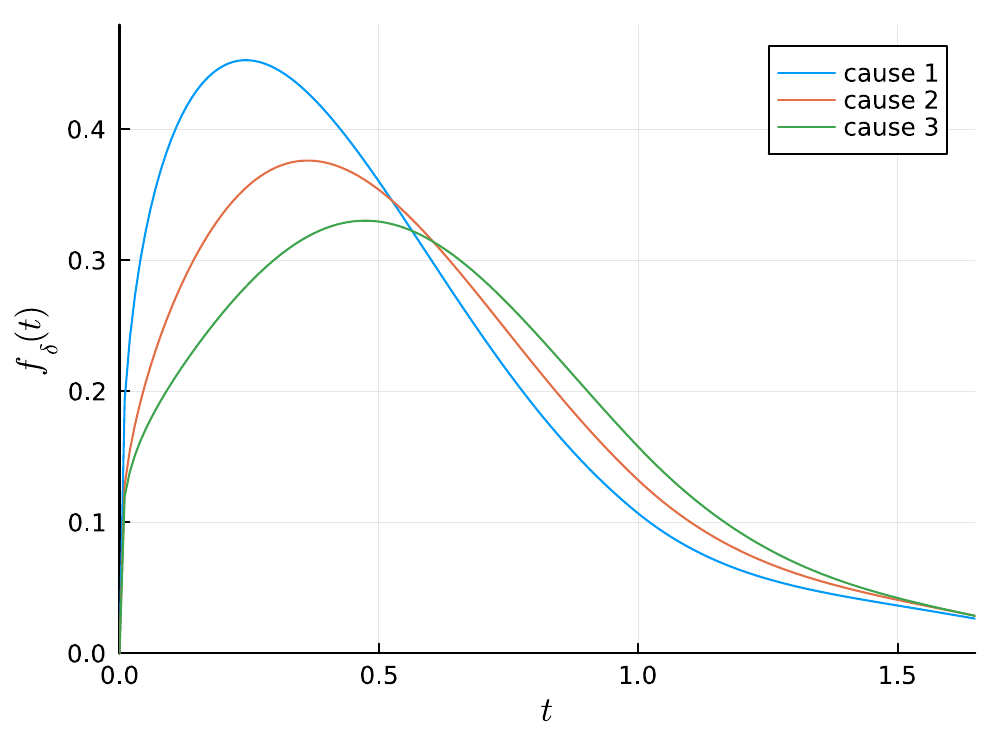}
		}
		\captionsetup{font=small}
		\caption{Cause-specific hazard rates (left) and incidence functions (right) for the $D = 3$ competing 
		independent sources of risks in the data--generating model for the simulation study in Section~\ref{sec:simulation_independent}.}
		\label{fig:model}
	\end{figure}
	
	\subsection{Application to the bone marrow transplantation dataset}
	
	The EBMT dataset in Section~\ref{sec:applications} includes data for $n = 400$ patients diagnosed with acute myeloid leukemia, who underwent a bone marrow transplantation.
	Our analysis involves a single binary predictor (the stem cells source), and therefore a single regression coefficient $\eta$, for which a non-informative centered Gaussian hyperprior $\eta \sim \mathcal{N}(0, 100)$ is specified. Moreover, the increasing hazard rates assumption implied by the Dykstra-Laud kernel appears too restrictive, and the Ornstein-Uhlenbeck kernel is considered instead,
	\begin{equation*}
		k(t; x) = \sqrt{2 \kappa} \, \exp(- \kappa(t-x)) \,\mathds{1}_{\{t \ge x\}},
	\end{equation*}
	where a non--informative exponential hyperprior $\kappa \sim \text{Exp}(0.1)$ is placed on its rate parameter $\kappa > 0$; with this kernel choice, the function $K_n$ in \eqref{kernel_quantities} takes the form
	\begin{equation*}
		K_n(x) = \sqrt{2/\kappa} \, \sum_{i=1}^n \max \left(1 - e^{- \kappa(t-x)}, \, 0 \right).
	\end{equation*}
	Figure~\ref{supfig:transplantation} reports histograms of posterior samples of the number $k$ of shared latent locations, the concentration parameter $\theta$, the kernel parameter $\kappa$ and the regression coefficient $\eta$.
	
	\begin{figure}
		\centering
		\noindent\makebox[\textwidth]{\includegraphics[width=0.5\textwidth]{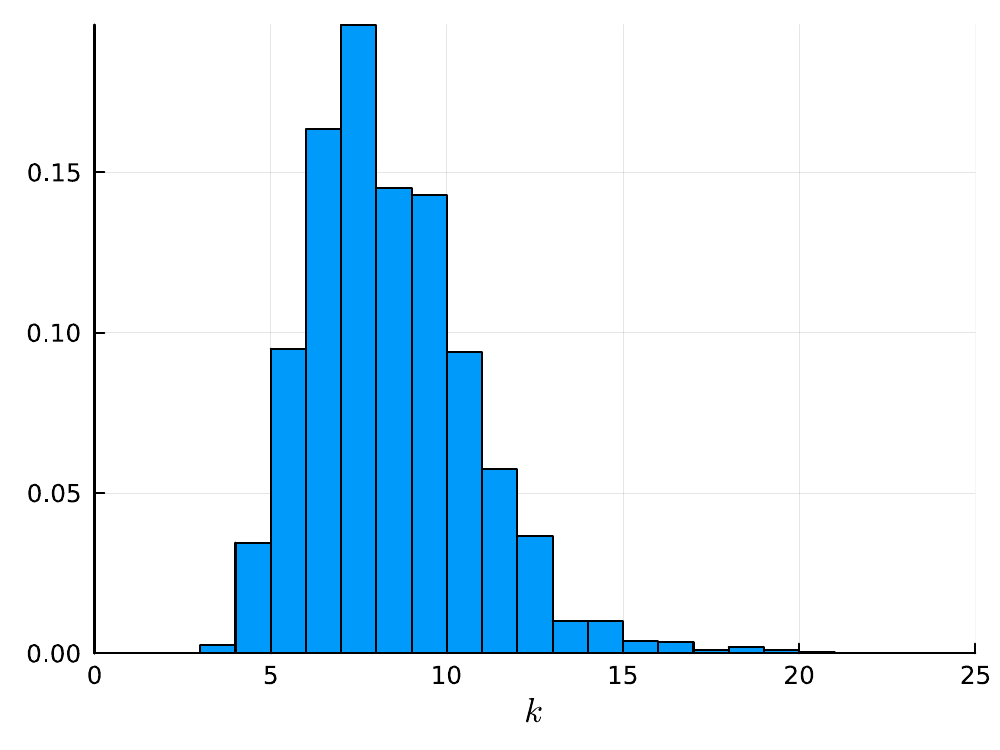}
			\includegraphics[width=0.5\textwidth]{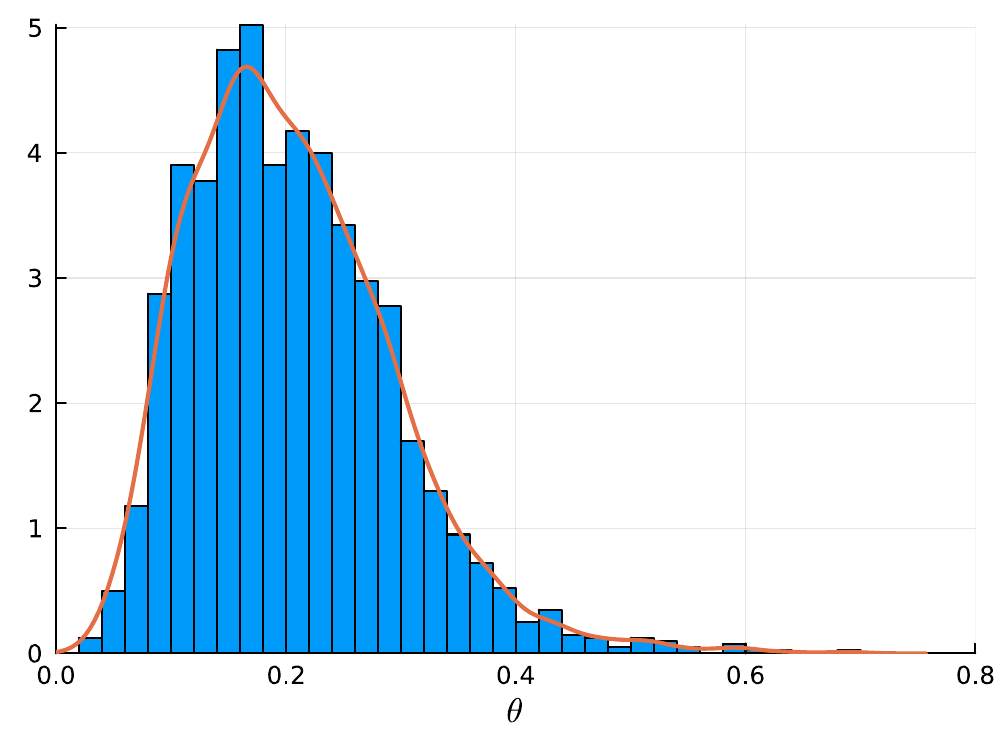}
		}
		\noindent\makebox[\textwidth]{\includegraphics[width=0.5\textwidth]{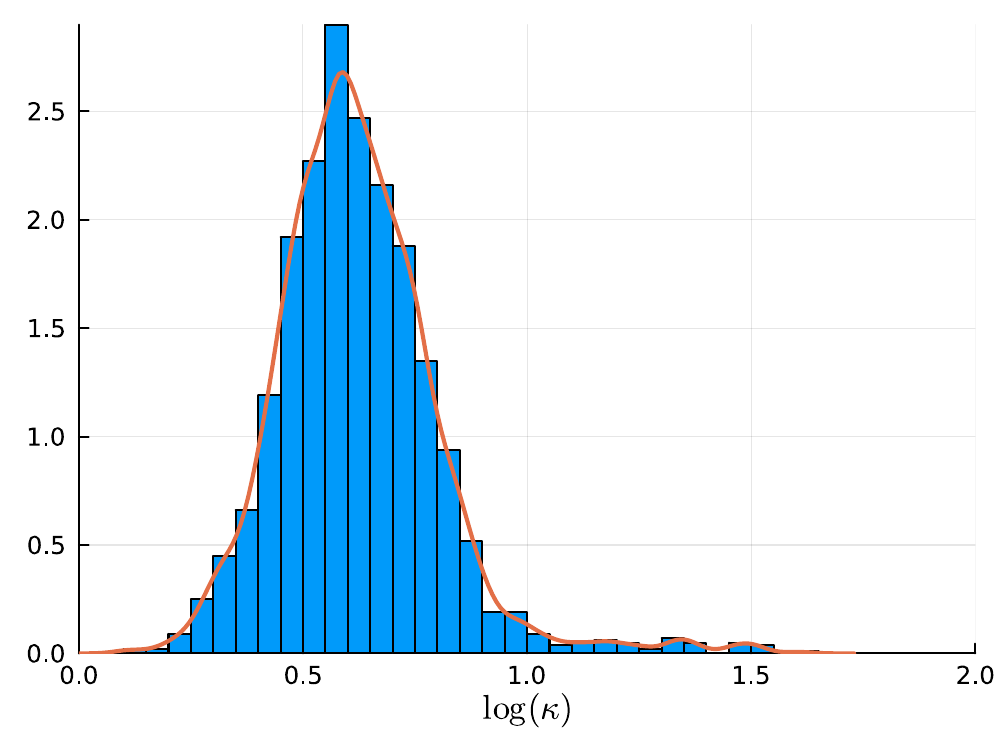}
			\includegraphics[width=0.5\textwidth]{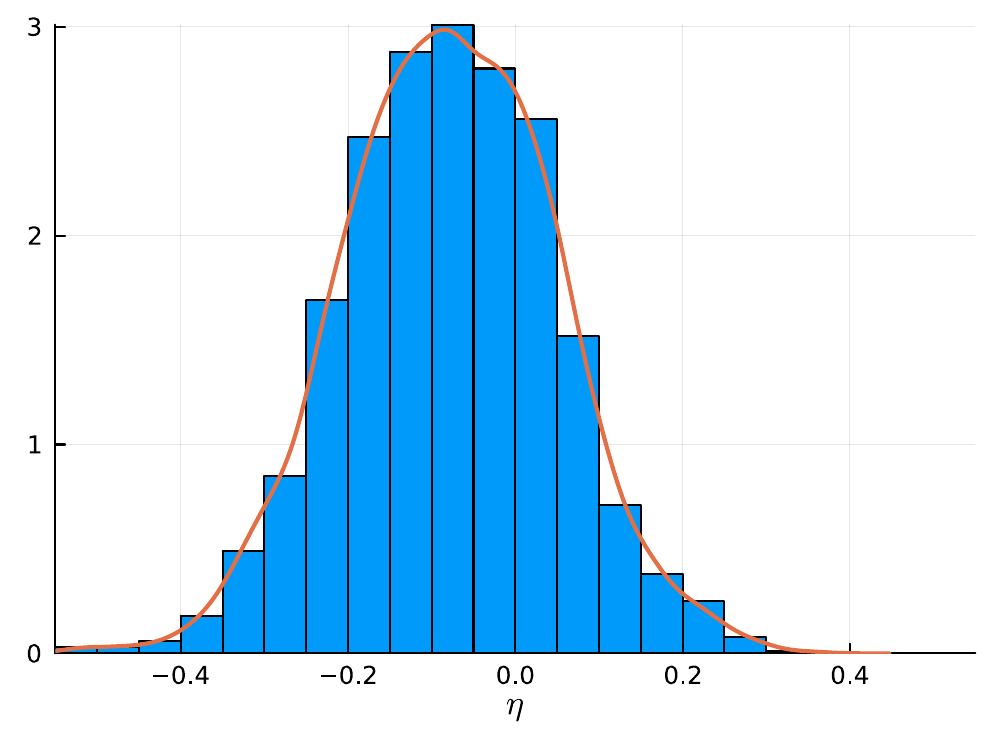}
		}
		\captionsetup{font=small}
		\caption{\emph{EBMT dataset}. Histograms of posterior samples of relevant model parameters: number $k$ of shared latent locations (top left), concentration parameter $\theta$ (top right), kernel parameter $\kappa$ on log--scale (bottom left) and regression coefficient $\eta$ (bottom right).}
		\label{supfig:transplantation}
	\end{figure}
	
	\clearpage
	\section{Analysis of the melanoma dataset}
	\label{supsec:melanoma}
	
	The melanoma dataset was collected by \cite{drzewiecki1980} at the Odense University Hospital, Denmark, and includes $n = 205$ patients affected by stage I melanoma, who underwent surgery between $1962$ and $1977$. Data are publicly available as part of the \texttt{timereg} package in R. \cite{andersen1993} illustrate several survival analysis methods on this dataset, while \cite{arfe2019} analyzed it with their model for competing risks.
	The primary event is death due to melanoma ($57$ patients, $27.8$\%), with death from other causes as competing event ($14$ patients, $6.8$\%); the remaining observations are censored ($65.4$\%). The median survival time is $5.89$ years, with the maximum observed survival time for the primary event of interest being $9.14$ years. The categorical predictor is gender, with $126$ women ($61.5$\%) and $79$  men ($38.5$\%). 
	As for the EBMT dataset, we consider the Ornstein-Uhlenbeck kernel with exponential hyperprior $\kappa \sim \text{Exp}(0.1)$ on the rate parameter.
	
	Figure~\ref{supfig:melanoma} shows the posterior estimates of the survival function and subdistribution functions for the primary and competing event types, for both female and male patients, compared with frequentist estimators. The posterior estimate of the hazard rate ratio $\exp(\eta)$ is $1.93$ with a $0.95$ credible interval of $[1.19,2.99]$, indicating a significant difference between male and female patients. This aligns with clinical evidence showing that male patients with melanoma typically have a higher risk \citep{scoggins2006}.
	The prediction curves for both the primary event (death due to melanoma) and the competing event (death from other causes) are showcased in Figure~\ref{fig:prediction_melanoma} of the main manuscript.
	Finally, Figure~\ref{supfig:melanoma_parameters} reports histograms of posterior samples of the number $k$ of shared latent locations, the concentration parameter $\theta$, the kernel parameter $\kappa$ and the regression coefficient $\eta$, for which a centered Gaussian hyperprior $\eta \sim \mathcal{N}(0, 100)$ is specified.
	
	\begin{figure}
		\centering
		\noindent\makebox[\textwidth]{\includegraphics[width=0.5\textwidth]{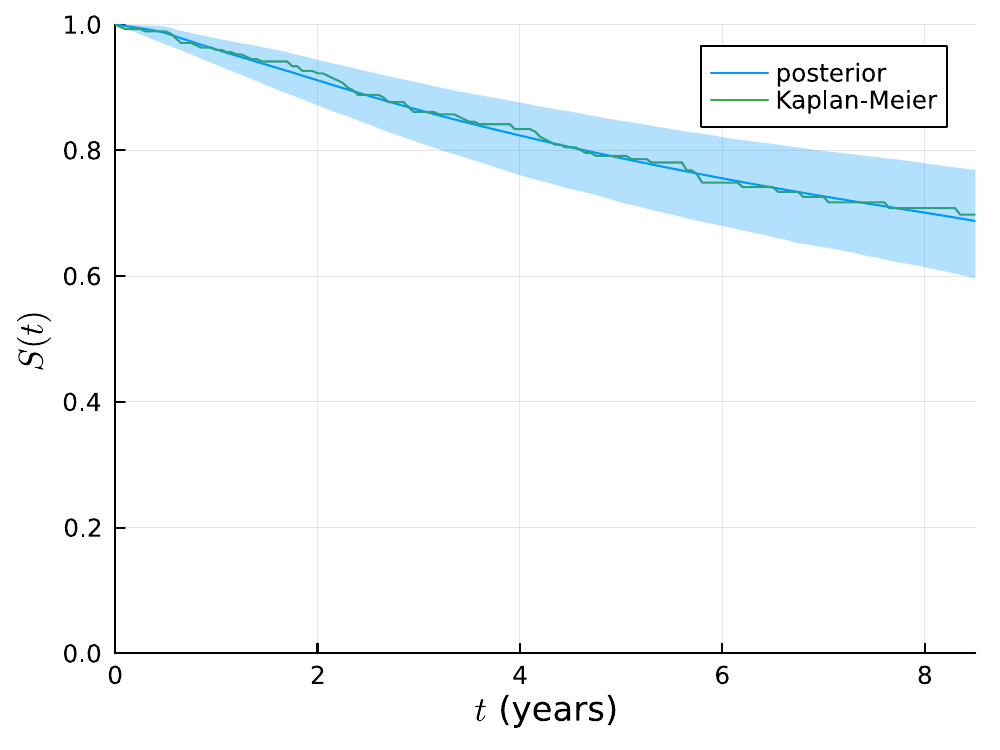}
			\includegraphics[width=0.5\textwidth]{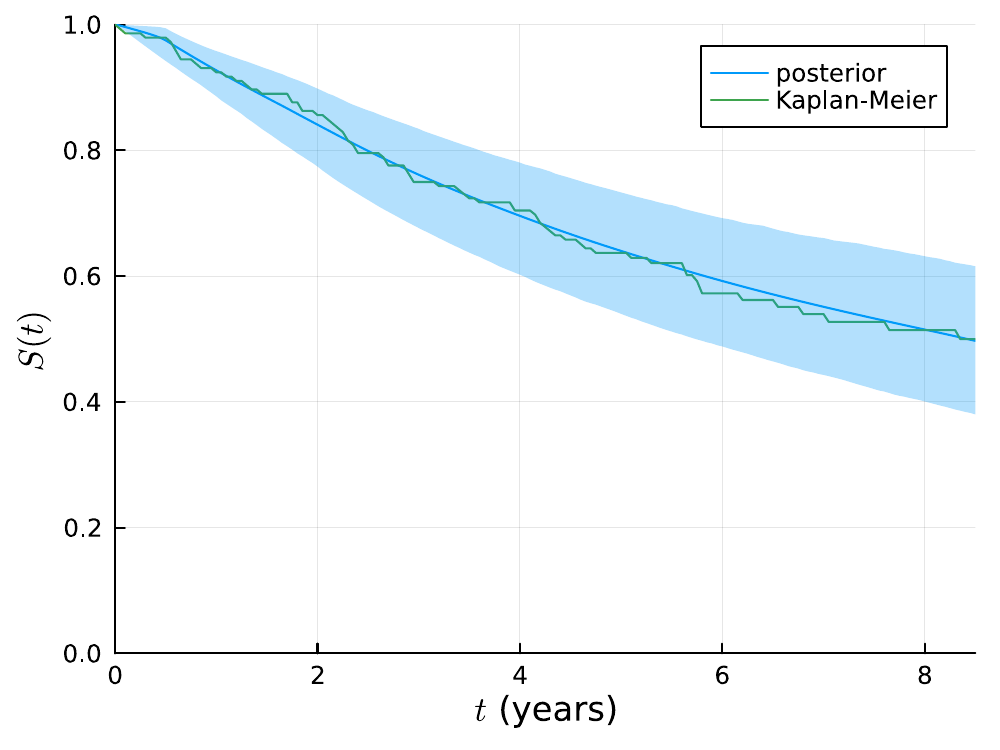}
		}
		\noindent\makebox[\textwidth]{\includegraphics[width=0.5\textwidth]{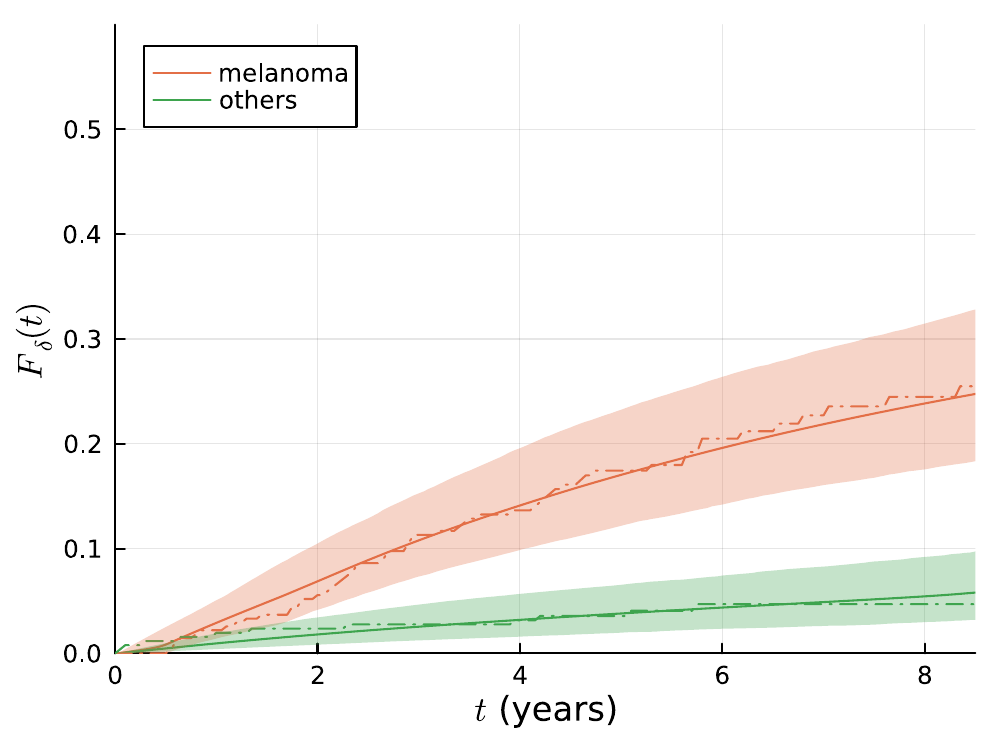}
			\includegraphics[width=0.5\textwidth]{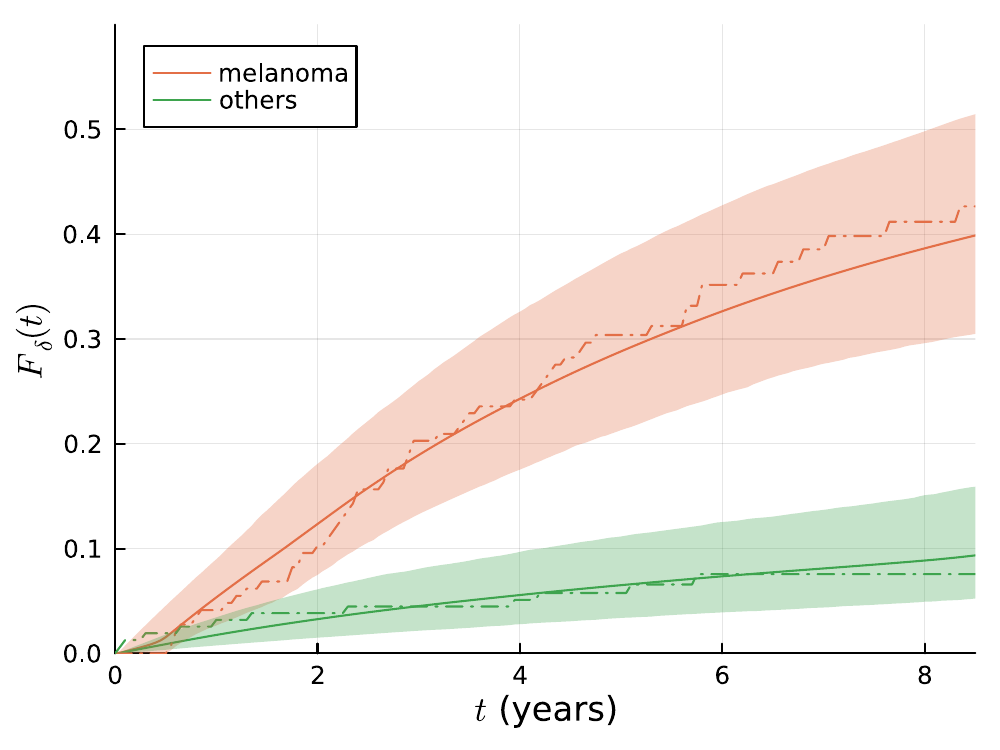}
		}
		\captionsetup{font=small}
		\caption{\emph{Melanoma dataset}. Posterior estimates of survival function (top) and subdistribution functions (bottom) for the primary (death due to melanoma) and competing (death from other causes) event types, compared with the frequentist estimators (top: green; bottom: dash-dotted); plots refer to female (left) and male (right) patients; pointwise $0.95$ credible bands are also displayed.} \label{supfig:melanoma}
	\end{figure}
	
	\begin{figure}
		\centering
		\noindent\makebox[\textwidth]{\includegraphics[width=0.5\textwidth]{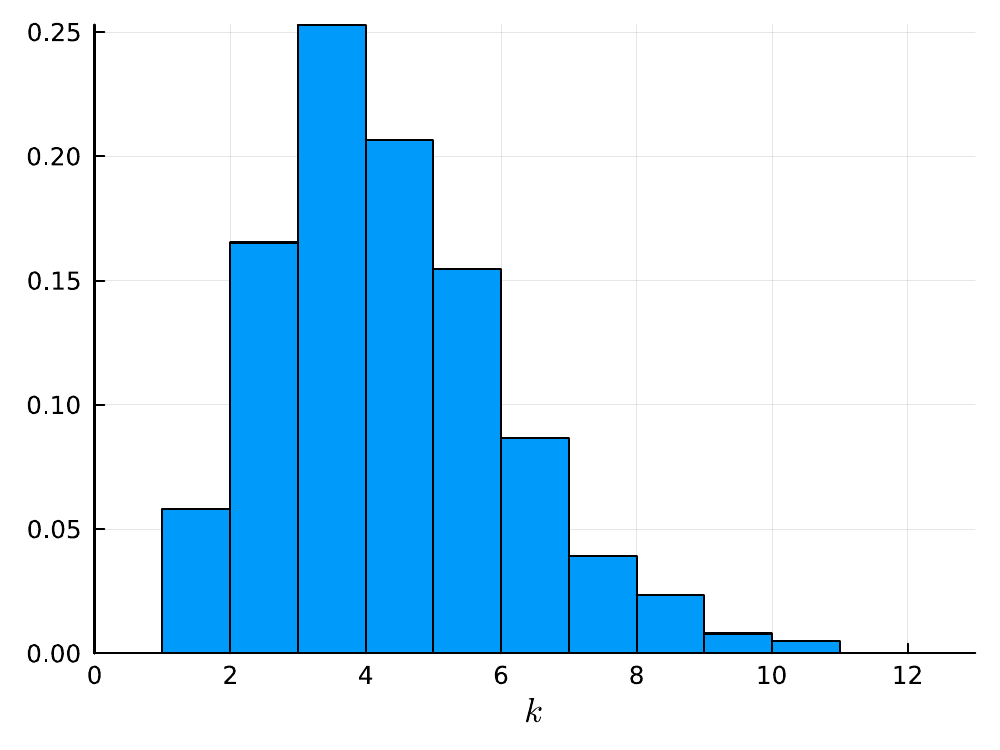}
			\includegraphics[width=0.5\textwidth]{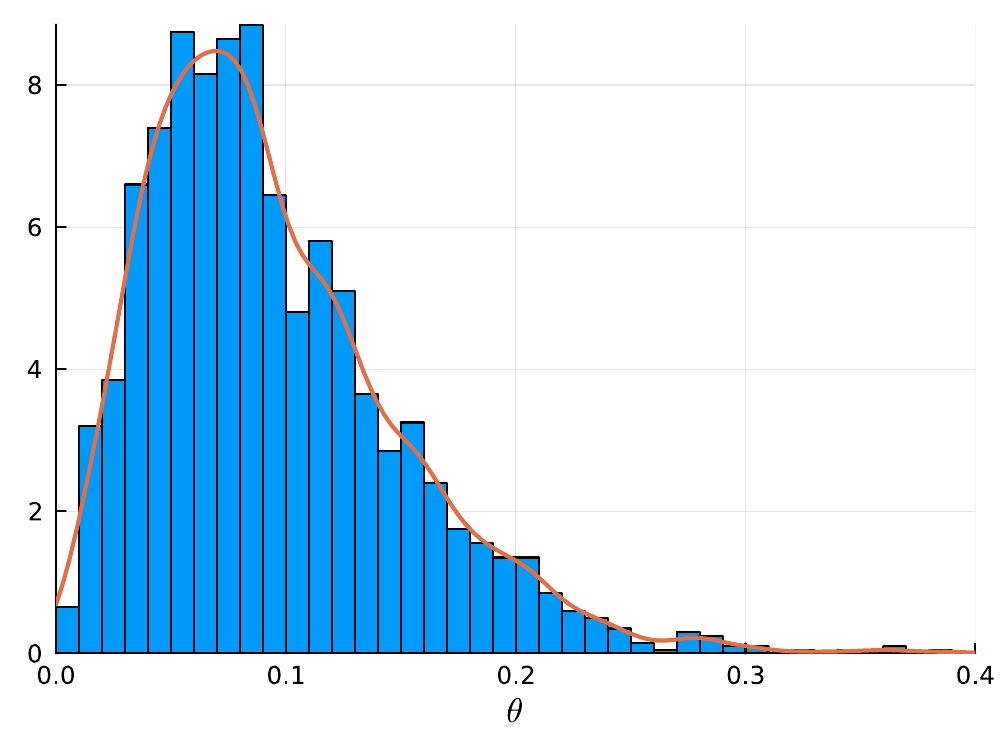}
		}
		\noindent\makebox[\textwidth]{\includegraphics[width=0.5\textwidth]{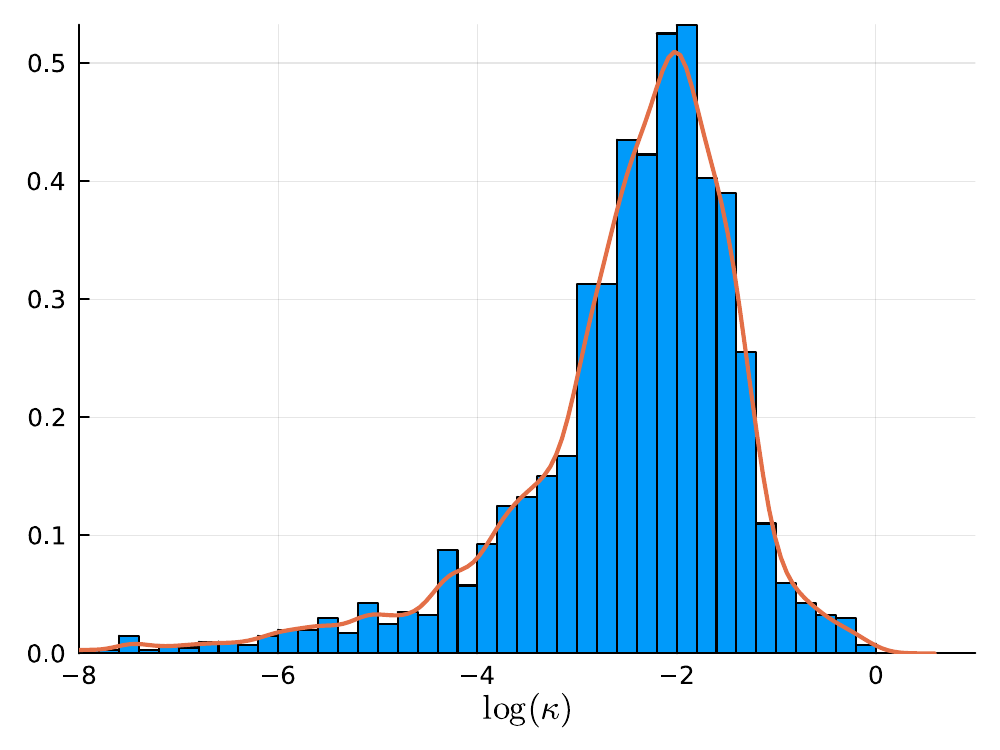}
			\includegraphics[width=0.5\textwidth]{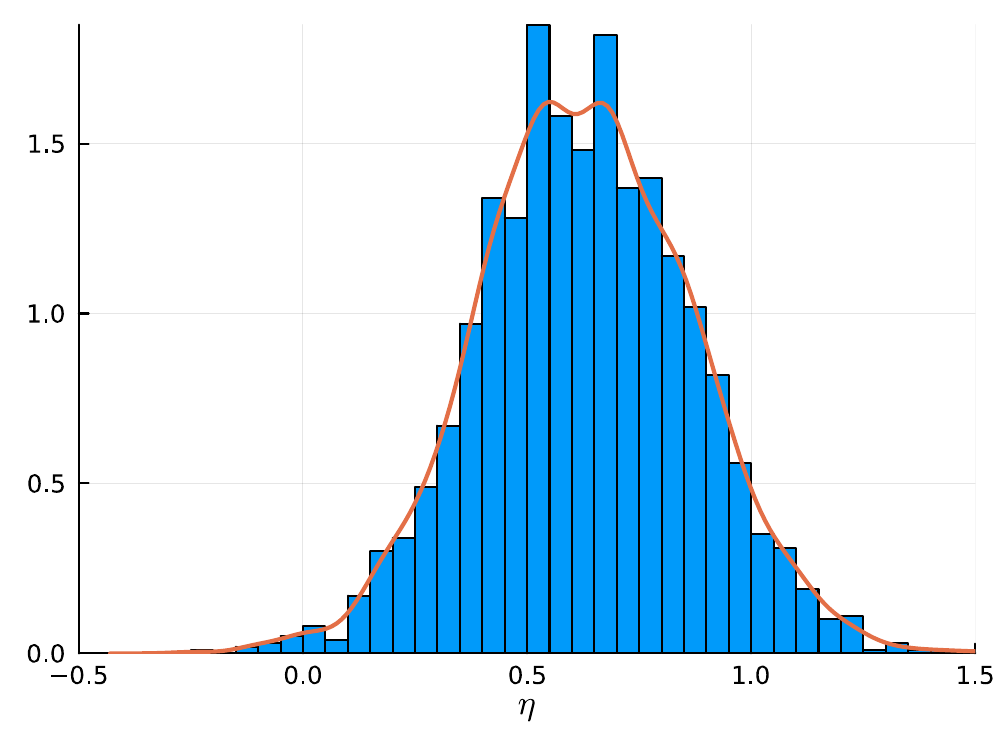}
		}
		\captionsetup{font=small}
		\caption{\emph{Melanoma dataset}. Histograms of posterior samples of relevant model parameters: number $k$ of shared latent locations (top left), concentration parameter $\theta$ (top right), kernel parameter $\kappa$ on log--scale (bottom left) and regression coefficient $\eta$ (bottom right).}
		\label{supfig:melanoma_parameters}
	\end{figure}
	
	\clearpage
	\section{Comparison with a tree-based approach} \label{supsec:sparapani}
	
	This section presents a comparison of the method proposed in our paper, based on hCRMs,
	with the nonparametric approach to competing risks of \citet{sparapani2020}. Their method relies on Bayesian Additive Regression Trees (BART) \citep{chipman2010}, a tree--based ensemble technique widely used for both classification and regression.
	Specifically, the authors adapted the binary probit BART developed in \cite{sparapani2016} for discrete--time survival analysis to the competing risks scenario.  A software implementation of their method is publicly available as part of the \texttt{BART} package in \texttt{R}.
	
	Before presenting the main numerical results on synthetic and real data, we highlight several remarkable features that set the two modeling approaches apart. First, our model is based on kernel mixtures, and thus the prior is supported on the space of continuous cause--specific hazards. As a result, our method is naturally suited to model continuous survival data. In contrast, the tree--based approach in \cite{sparapani2020} assumes discrete--time survival data, since these can be easily accommodated within the BART framework. Additionally, while our approach naturally yields estimates of cause--specific incidence functions and prediction curves, the tree--based method is limited to estimating survival and subdistribution functions. Indeed, these are the functionals included in the output of the \texttt{crisk.bart} function. Moreover, introducing the BART method, the authors consider only two competing causes of risks, namely the main cause of interest and all the others together. Accordingly, the comparison focuses on scenarios involving two competing risks. Finally, categorical predictors are incorporated via a Cox proportional hazards model (Section~\ref{sec:applications}), which preserves interpretability at the cost of reduced flexibility. In contrast, their proposal allows for a richer dependence structure. 
	
	The comparison presented below is necessarily limited to the estimation of survival and subdistribution functions, since, as noted, the tree--based method does not cover cause--specific incidence functions and prediction curves, which are a distinctive capability of our approach. The results should be viewed as preliminary, and we refrain from drawing definitive conclusions, as a more extensive analysis is required and will be the subject of future work.
	
	\subsection{Illustration with synthetic data}
	\label{supsec:illustration_bart}
	
	Consider a scenario with $D = 2$ competing sources of risk. Observations are sampled according to the latent failure times approach (Section~\ref{supsec:latent_times}), and the sample size is $n = 100$. Specifically, for $i = 1, \ldots, n$, 
	\begin{equation*}
		Y_{i,1} \stackrel{\mbox{\scriptsize iid}}{\sim} \text{Weibull} \big(\xi_{1} = 1.2\big), \qquad
		Y_{i,2} \stackrel{\mbox{\scriptsize iid}}{\sim} \text{Weibull} \big(\xi_{2} = 2.4\big),
	\end{equation*}
	where $\xi_1$ and $\xi_2$ are shape parameters. Following the illustration in Section~\ref{sec:simulation_three_risks}, our method relies on the generalized gamma hCRM prior with Dykstra-Laud kernel; the hyperparameters specification, as well as the MCMC settings, coincide with those adopted in Section~\ref{supsec:applications}. For the BART method, we consider the default parameters in the \texttt{crisk.bart} function; the survival times are rounded at the second decimal digit. Figure~\ref{supfig:comparison} displays the posterior estimates of the survival and subdistribution functions, along with their pointwise $0.95$ credible bands, obtained with the two alternative methods; the frequentist estimates are also reported. Notably, the estimated curves produced with the hCRM approach are smoother than those obtained with the BART method, which instead nearly overlap with the frequentist estimates. 
	A smooth estimate is clearly preferable in this scenario, and more generally in any context where the underlying data--generating mechanism is known to be continuous.
	On the other hand, the uncertainty around the estimates is comparable across the two methods, which yield very similar credible bands. Table~\ref{suptable:errors} reports the estimation errors  for  the survival and subdistribution functions under the hCRM method (using the marginal algorithm), the BART method, and the frequentist approach. The errors are computed as described in Section~\ref{sec:simulation_independent}. In this simulation scenario, our method clearly outperforms the tree--based approach  due to its smoother estimates.
	
	\begin{table}
		\centering
		\begin{tabular}{lccc}
			\hline
			& survival & \multicolumn{2}{c}{subdistribution} \\ 
			\cline{3-4} 
			&	& $\delta=1$ & $\delta=2$ \\ 
			\hline
			hCRM & \textbf{.0339} & \textbf{.0592} & \textbf{.0818} \\ 
			BART & .0686 & .1048 & .1252 \\ 
			frequentist & .0744 & .0938 & .1354 \\ 
			\hline
		\end{tabular}
		\captionsetup{width=0.85\textwidth,font=small}
		\caption{Comparison of rescaled total variation distances between estimated and true survival and subdistribution functions, for the synthetic dataset in Section~\ref{supsec:illustration_bart}, using the hCRM method, the BART method and the frequentist approach.}
		\label{suptable:errors}
	\end{table}
	
	\begin{figure}
		\centering
		\noindent\makebox[\textwidth]{\includegraphics[width=0.5\textwidth]{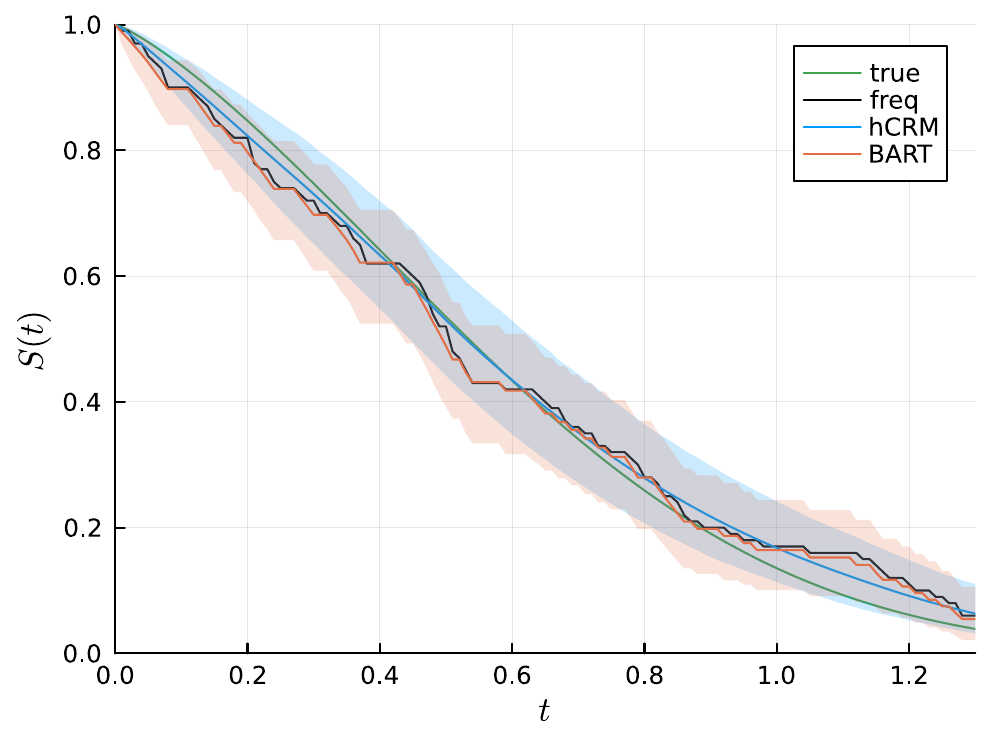}
			\includegraphics[width=0.5\textwidth]{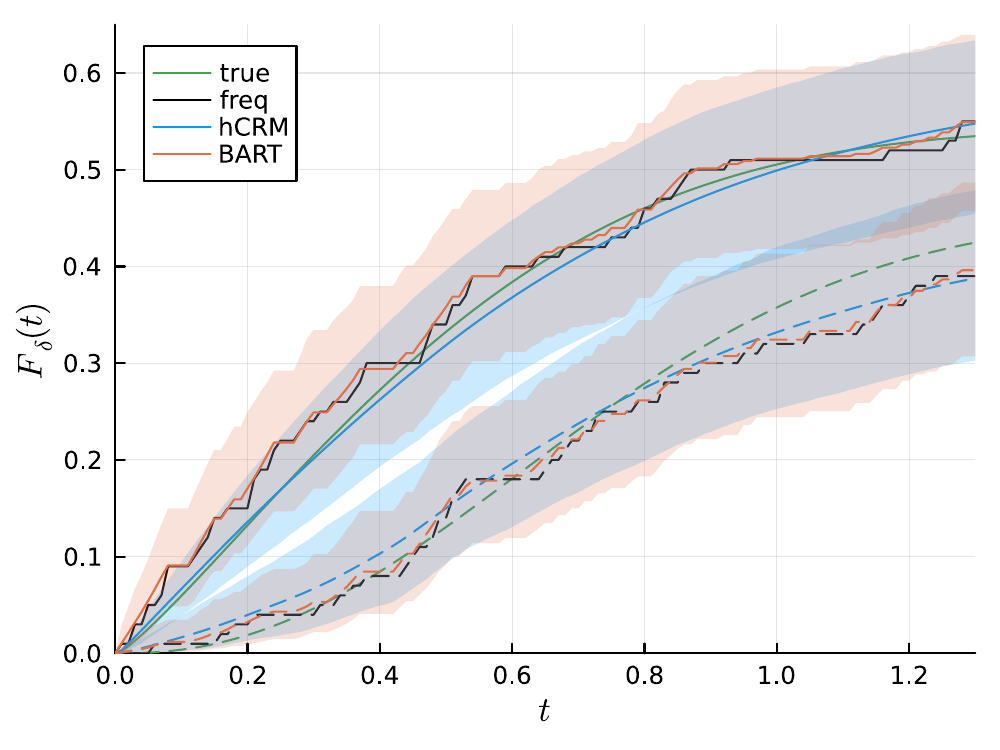}
		}
		\captionsetup{font=small}
		\caption{Posterior estimates of the survival function (left) and subdistribution functions (right), for the synthetic dataset in Section~\ref{supsec:illustration_bart}; comparison of the proposed method (hCRM, blue), the tree--based method (BART, orange), and the frequentist estimates (black) with the true curves (green); the pointwise $0.95$ credible bands are also displayed.}
		\label{supfig:comparison}
	\end{figure}
	
	\subsection{Application to clinical datasets}
	
	The comparison with the tree--based method of \citet{sparapani2020} is also conducted on the two clinical datasets considered in this paper: the EBMT dataset (Section~\ref{sec:applications}) and the melanoma dataset (Section~\ref{supsec:melanoma}). We refer the reader to those sections for details on our modeling choices and the specification of hyperparameters. For the BART model, we use the default parameters with survival times rounded to the nearest week.
	Figures~\ref{supfig:comparison_transplantation} and~\ref{supfig:comparison_melanoma} display the posterior estimates of the survival and subdistribution functions, along with their pointwise $0.95$ credible bands, obtained using the two alternative methods across different levels of the binary predictors; frequentist estimates are also reported for comparison. The considerations made above for the synthetic dataset continue to hold in the case of the melanoma data: our method yields smoother curves, whereas the BART approach tends to more closely follow the frequentist estimates. Instead, no clear or consistent pattern emerges for the EBMT dataset. The posterior estimates from the BART method appear relatively smooth for larger survival times, but occasionally show noticeable departures from the frequentist estimates. As for the credible bands, they are generally comparable across methods, with a remarkable exception for the subdistribution functions for the EBMT dataset: the credible bands obtained via the BART method are substantially wider. 
	
	\begin{figure}
		\centering
		\noindent\makebox[\textwidth]{\includegraphics[width=0.5\textwidth]{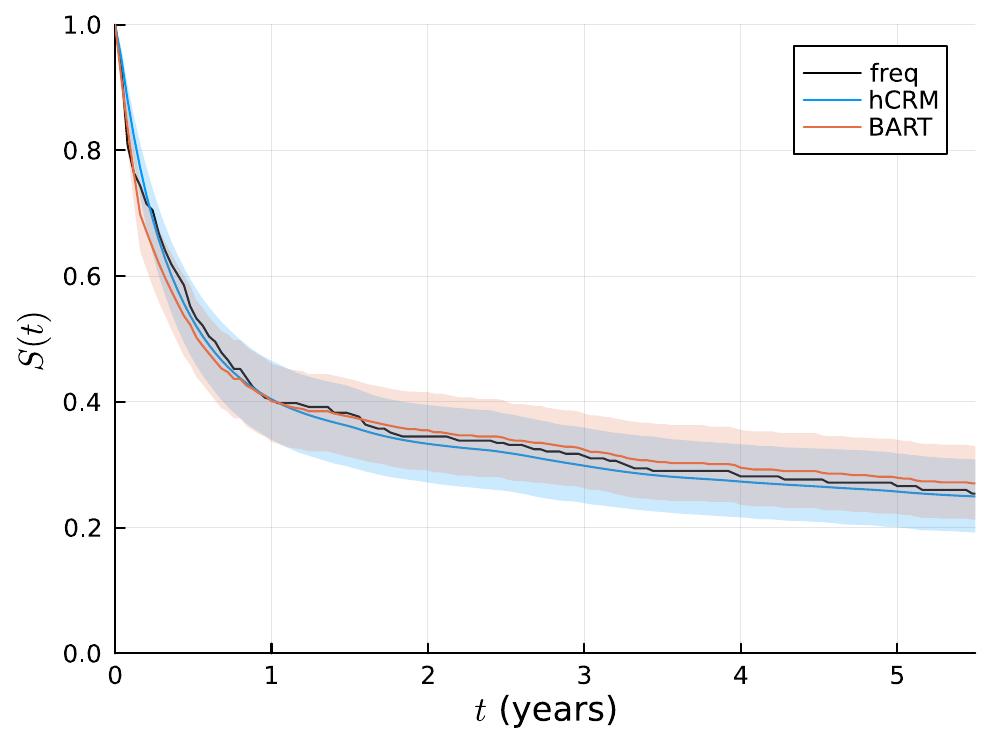}
			\includegraphics[width=0.5\textwidth]{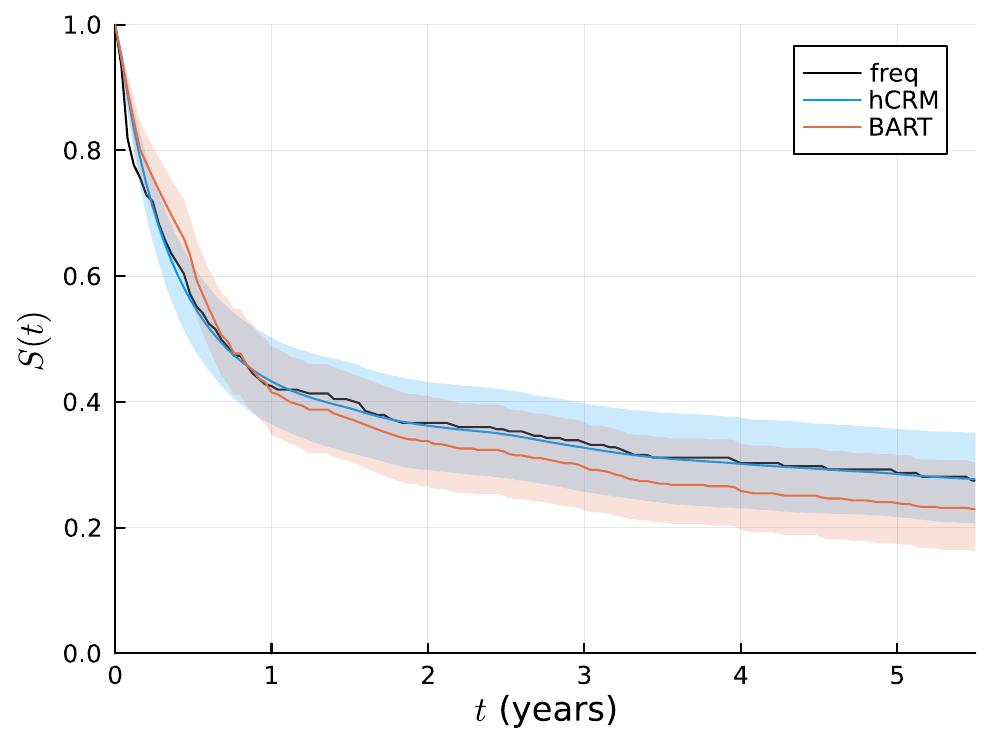}
			
		}
		\noindent\makebox[\textwidth]{\includegraphics[width=0.5\textwidth]{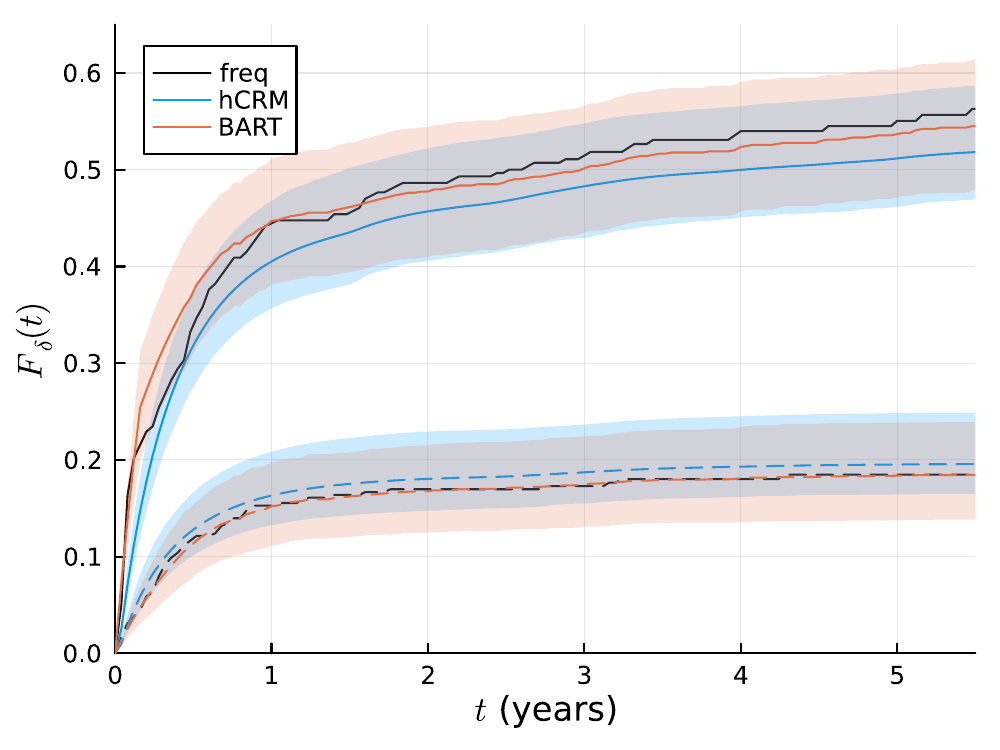}
			\includegraphics[width=0.5\textwidth]{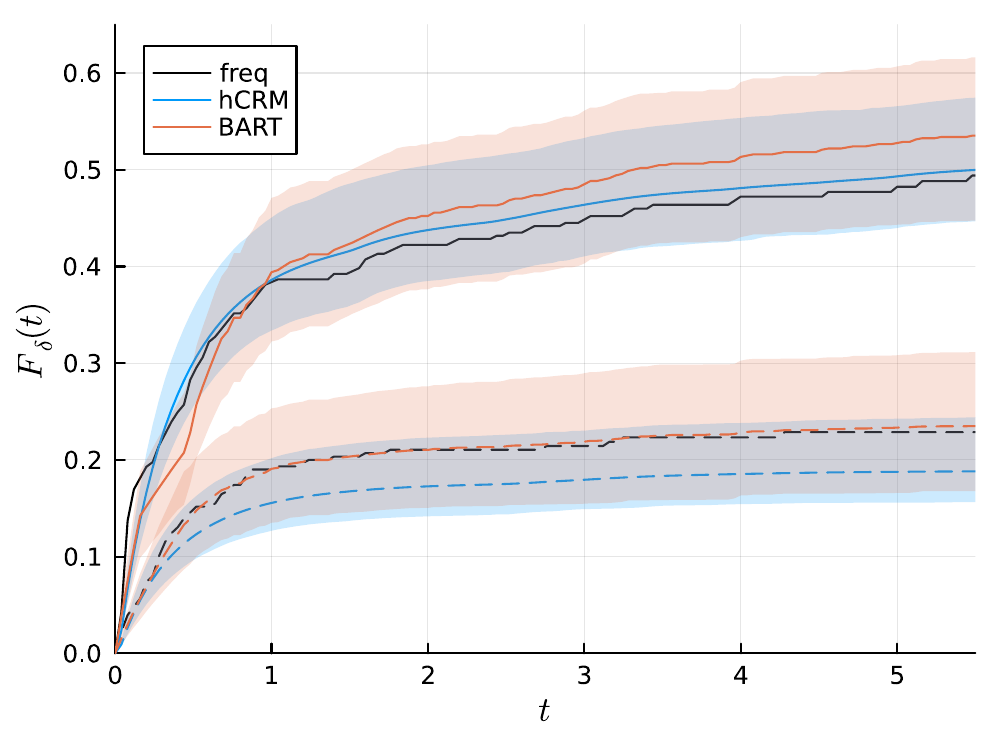}
		}
		\captionsetup{font=small}
		\caption{\emph{EBMT dataset}. Posterior estimates of the survival function (top) and subdistribution functions (bottom), for patients who experienced graft from bone marrow (left) and peripheral blood (right) cells; comparison of the proposed method (hCRM, blue), the tree--based method (BART, orange), and the frequentist estimates (black); the pointwise $0.95$ credible bands are also displayed.}
		\label{supfig:comparison_transplantation}
	\end{figure}
	
	\begin{figure}
		\centering
		\noindent\makebox[\textwidth]{\includegraphics[width=0.5\textwidth]{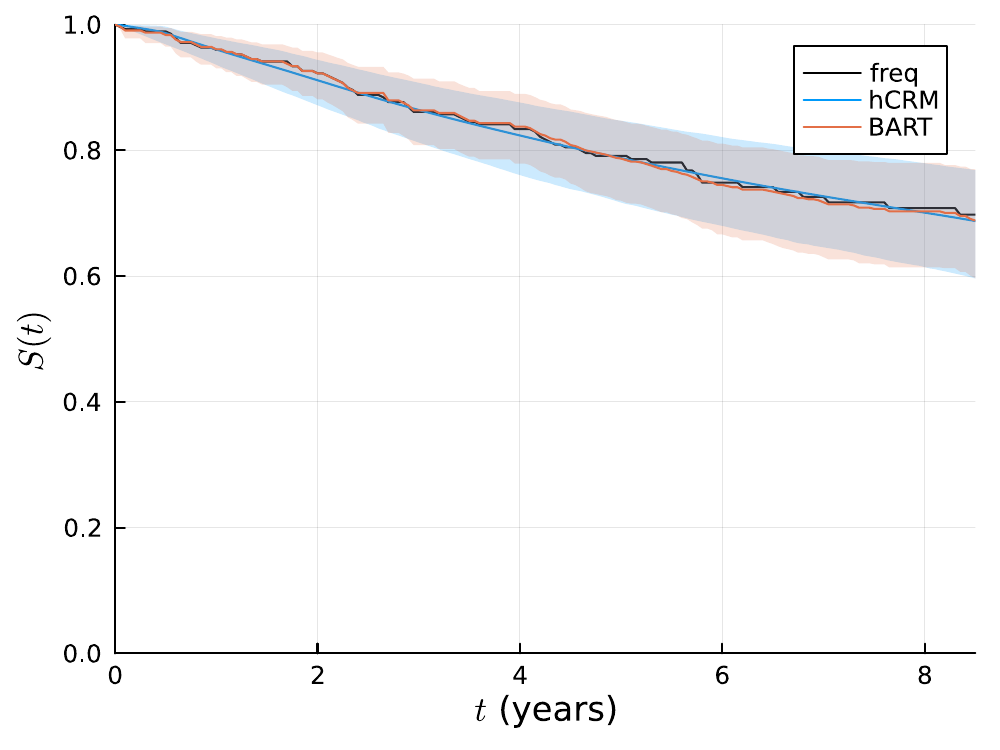}
			\includegraphics[width=0.5\textwidth]{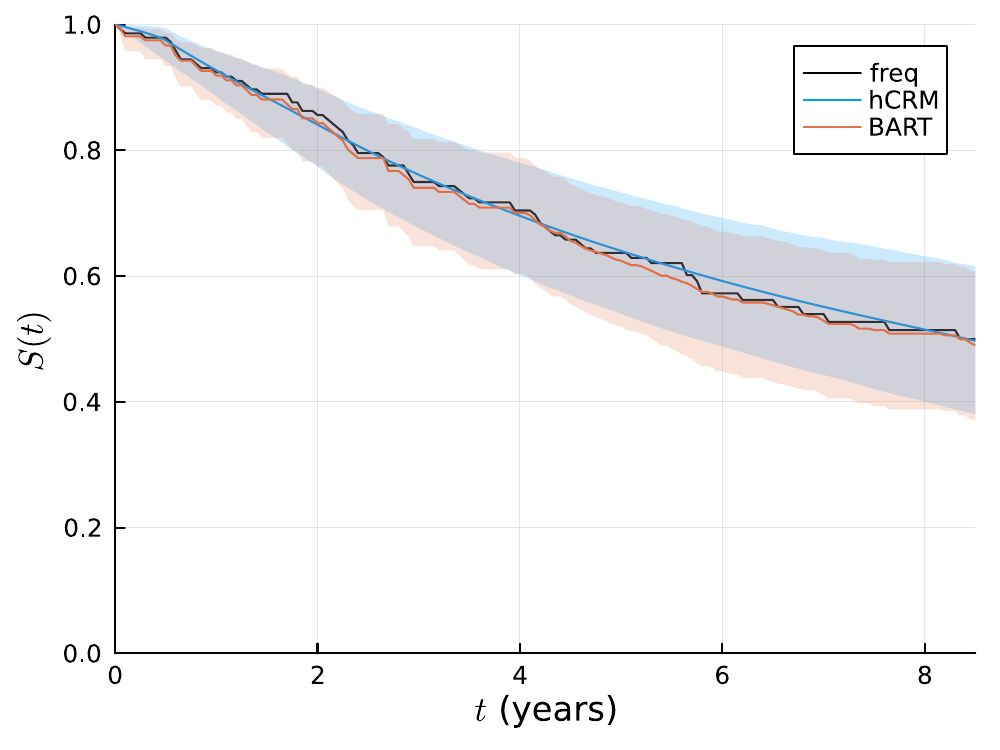}
		}
		\noindent\makebox[\textwidth]{\includegraphics[width=0.5\textwidth]{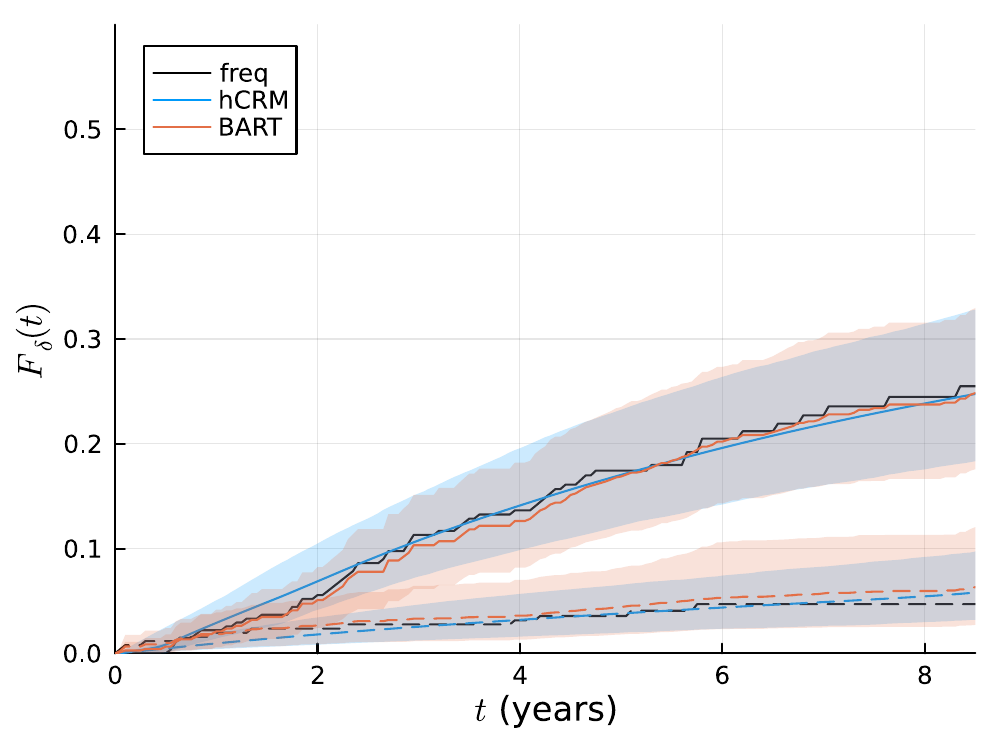}
			\includegraphics[width=0.5\textwidth]{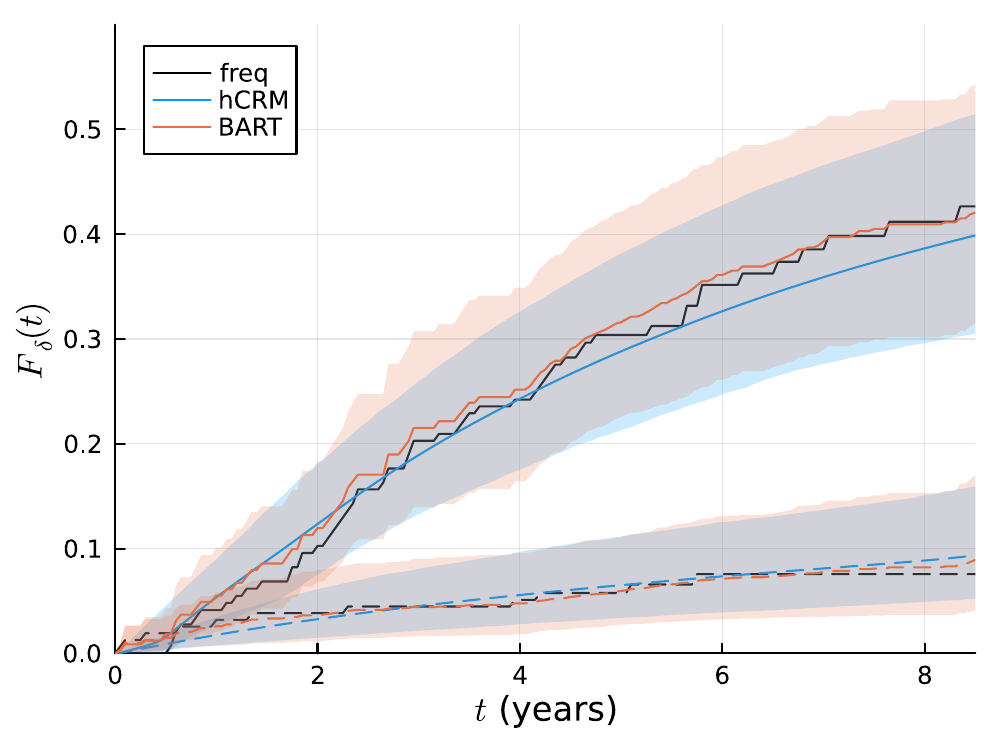}
		}
		\captionsetup{font=small}
		\caption{\emph{Melanoma dataset}. Posterior estimates of the survival function (top) and subdistribution functions (bottom), for female (left) and male (right) patients; comparison of the proposed method (hCRM, blue), the tree--based method (BART, orange), and the frequentist estimates (black); the pointwise $0.95$ credible bands are also displayed.}
		\label{supfig:comparison_melanoma}
	\end{figure}

\end{document}